\documentclass[dvipsnames,11pt]{article}

\usepackage{amsmath,amssymb,amsthm,bm,color,mathrsfs,extarrows,mathtools,enumitem,fancybox,makecell}
\usepackage[square,sort,comma,numbers]{natbib}
\usepackage{subcaption}
\usepackage[margin=1.in]{geometry}
\usepackage{xcolor,bbm}
\usepackage[pagebackref]{hyperref}
\usepackage{bm}
\usepackage{bbm}
\usepackage{comment}
\allowdisplaybreaks
\newif\ifsubmit   

\usepackage[normalem]{ulem}

\setlist[itemize]{itemsep=0pt}
\setlist[enumerate]{itemsep=0pt}
\hypersetup{colorlinks=true,urlcolor=blue,linkcolor=blue,citecolor=[rgb]{.42,.56,.14},}

\usepackage[mathcal]{eucal}

\usepackage[capitalise,nameinlink]{cleveref}
\Crefname{lemma}{Lemma}{Lemmas}
\Crefname{fact}{Fact}{Facts}
\Crefname{theorem}{Theorem}{Theorems}
\Crefname{corollary}{Corollary}{Corollaries}
\Crefname{claim}{Claim}{Claims}
\Crefname{example}{Example}{Examples}
\Crefname{problem}{Problem}{Problems}
\Crefname{definition}{Definition}{Definitions}
\Crefname{notation}{Notation}{Notations}
\Crefname{assumption}{Assumption}{Assumptions}
\Crefname{subsection}{Subsection}{Subsections}
\Crefname{section}{Section}{Sections}
\Crefformat{equation}{(#2#1#3)}

\newtheorem{theorem}{Theorem}[section]
\newtheorem*{theorem*}{Theorem}

\newtheorem{proposition}[theorem]{Proposition}
\newtheorem*{proposition*}{Proposition}
\newtheorem{lemma}[theorem]{Lemma}
\newtheorem*{lemma*}{Lemma}
\newtheorem{corollary}[theorem]{Corollary}
\newtheorem*{corollary*}{Corollary}
\newtheorem*{conjecture*}{Conjecture}
\newtheorem{fact}[theorem]{Fact}
\newtheorem*{fact*}{Fact}

\newtheorem*{exercise*}{Exercise}

\newtheorem*{hypothesis*}{Hypothesis}
\newtheorem{conjecture}[theorem]{Conjecture}

\theoremstyle{definition}
\newtheorem{definition}[theorem]{Definition}

\newtheorem{example}[theorem]{Example}

\newtheorem{assumption}[theorem]{Assumption}
\newtheorem{claim}[theorem]{Claim}
\newtheorem*{claim*}{Claim}

\newtheorem{remark}[theorem]{Remark}
\newtheorem*{remark*}{Remark}

\newtheorem*{observation*}{Observation}
\numberwithin{equation}{section}

\newcommand{\sd}[1]{\mathrm{d}#1}
\DeclareMathOperator*{\E}{\mathbb E}

\renewcommand{\Pr}{\operatorname*{\mathbf{Pr}}}

\newcommand{\eps}{\varepsilon}
\newcommand{\abs}[1]{\left| #1 \right|}
\newcommand{\vabs}[1]{\left\| #1 \right\|}
\newcommand{\abra}[1]{\left\langle #1 \right\rangle}
\newcommand{\sabra}[1]{\langle #1 \rangle}
\newcommand{\pbra}[1]{\left( #1 \right)}
\newcommand{\sbra}[1]{\left[ #1 \right]}
\newcommand{\cbra}[1]{\left\{ #1 \right\}}

\newcommand{\floorbra}[1]{\left\lfloor #1 \right\rfloor}

\renewcommand{\mid}{\,\middle\vert\,}
\newcommand{\bin}{\{0,1\}}
\newcommand{\binpm}{\{\pm1\}}

\newcommand{\poly}{\mathrm{poly}}
\newcommand{\polylog}{\mathrm{polylog}}
\newcommand{\sgn}{\mathrm{sgn}}
\newcommand{\indicator}{\mathbf{1}}

\newcommand{\unit}{\mathrm{unit}}
\newcommand{\frob}[1]{\vabs{#1}}
\newcommand{\SendReal}{\mathrm{trunc}}
\newcommand{\opnorm}[1]{\vabs{#1}_\mathrm{op}}
\newcommand{\Forr}{\mathrm{Forr}}
\newcommand{\com}{\mu}
\newcommand{\comtwo}{\sigma}

\makeatletter
\newcommand*\bigcdot{\mathpalette\bigcdot@{.5}}
\newcommand*\bigcdot@[2]{\mathbin{{\hbox{\scalebox{#2}{$\m@th#1\bullet$}}}}}
\makeatother

\newcommand{\tensor}{\overset{\bigcdot}{\otimes}}

\newcommand{\Nbb}{\mathbb{N}}
\newcommand{\Rbb}{\mathbb{R}}
\newcommand{\Sbb}{\mathbb{S}}

\newcommand{\Bcal}{\mathcal{B}}
\newcommand{\Ccal}{\mathcal{C}}

\newcommand{\Ecal}{\mathcal{E}}
\newcommand{\Fcal}{\mathcal{F}}
\newcommand{\Gcal}{\mathcal{G}}

\newcommand{\Tcal}{\mathcal{T}}
\newcommand{\Ucal}{\unif}

\newcommand{\rbm}{\bm{r}}
\newcommand{\xbm}{\bm{x}}
\newcommand{\ybm}{\bm{y}}
\newcommand{\zbm}{\bm{z}}

\newcommand{\Sbm}{\bm{S}}
\newcommand{\Tbm}{\bm{T}}

\usepackage{braket}

\renewcommand{\tilde}{\widetilde}
\renewcommand{\bar}{\overline}
\renewcommand{\hat}{\widehat}

\newcommand{\rel}{\mathrm{rel}}

\newcommand{\N}{\mathbb{N}}
\newcommand{\X}{\bm{X}}
\newcommand{\Y}{\bm{Y}}

\newcommand{\U}{\bm{u}}
\newcommand{\V}{\bm{v}}

\newcommand{\D}{\bm{d}}
\newcommand{\balpha}{\bm{\beta}}

\newcommand{\Q}{\bm{Q}}

\newcommand{\lX}{\bm{x}}
\newcommand{\lY}{\bm{y}}
\newcommand{\lZ}{\bm{z}}
\newcommand{\lA}{\bm{a}}
\newcommand{\lB}{\bm{b}}
\newcommand{\bell}{\bm{\ell}}
\newcommand{\la}{\lA}
\newcommand{\lb}{\lB}
\newcommand{\lc}{\bm{c}}
\newcommand{\lx}{\lX}
\newcommand{\ly}{\lY}
\newcommand{\lz}{\lZ}
\newcommand{\lQ}{\bm{q}}

\newcommand{\btau}{\bm{\tau}}

\newcommand{\midd}{~\vert~}

\newcommand{\pmones}{\{\pm1\}^n}
\newcommand{\pmone}{\{\pm1\}}
\newcommand{\unif}{\nu}
\newcommand{\BE}{\E}

\newcommand{\supu}[1]{\bm{u}^{(#1)}}

\newcommand{\supv}[1]{\bm{v}^{(#1)}}

\newcommand{\supa}[1]{\lA^{(#1)}}
\newcommand{\supb}[1]{\lB^{(#1)}}
\newcommand{\supcbar}[1]{\bar\lc^{(#1)}}
\newcommand{\supX}[1]{\X^{(#1)}}
\newcommand{\supY}[1]{\Y^{(#1)}}
\newcommand{\supZ}[1]{\lZ^{(#1)}}

\newcommand{\supF}[1]{\Fcal^{(#1)}}
\newcommand{\F}{\mathcal{F}}
\newcommand{\ip}[2]{\abra{#1, #2}}
\newcommand{\R}{\mathbb{R}}
\renewcommand{\Lambda}{\eta}
\newcommand{\ind}{\mathbf{1}}

\newcommand{\BH}{\bm{H}}
\newcommand{\bPi}{\bm{\Pi}}
\newcommand{\bP}{\bm{p}}
\renewcommand{\Q}{\bm{q}}
\newcommand{\bQ}{\bm{q}}
\newcommand{\K}{\bm{k}}
\newcommand{\lw}{\bm{w}}
\newcommand{\bc}{\bm{c}}
\newcommand{\br}{\bm{r}}
\newcommand{\bs}{\bm{s}}
\newcommand{\Xell}{\X_{\bell}}
\newcommand{\Yell}{\Y_{\bell}}

\newcommand{\biased}[1]{\pi^{\otimes n}_{#1}}

\usepackage{xspace}

\newcommand{\flift}{f\circ{g}}

\newcommand{\Goos}{G\"{o}\"{o}s\xspace}

\title{Fourier Growth of Communication Protocols for XOR Functions}
\author{
Uma Girish\thanks{Princeton University. Email: \texttt{ugirish@cs.princeton.edu}}
\and
Makrand Sinha\thanks{Simons Institute and University of California at Berkeley. Email: \texttt{makrand@berkeley.edu}}
\and
Avishay Tal\thanks{University of California at Berkeley. Email: \texttt{atal@berkeley.edu}}
\and
Kewen Wu\thanks{University of California at Berkeley. Email: \texttt{shlw\_kevin@hotmail.com}}
}
\date{}

\begin{document}
\maketitle

\begin{abstract}
The level-$k$ $\ell_1$-Fourier weight of a Boolean function refers to the sum of absolute values of its level-$k$ Fourier coefficients. Fourier growth refers to the growth of these weights as $k$ grows.
It has been extensively studied for various computational models, and bounds on the Fourier growth, even for the first few levels, have proven useful in learning theory, circuit lower bounds, pseudorandomness, and quantum-classical separations.

In this work, we investigate the Fourier growth of certain functions that naturally arise from communication protocols for XOR functions (partial functions evaluated on the bitwise XOR of the inputs $x$ and $y$ to Alice and Bob). If a protocol $\mathcal{C}$ computes an XOR function, then $\mathcal{C}(x, y)$ is a function of the parity $x \oplus y$. This motivates us to analyze the \textit{XOR-fiber} of the communication protocol $\mathcal{C}$, defined as $h(z) := \mathbb{E}_{\bm{x},\bm{y}}[\mathcal{C}(\bm{x}, \bm{y}) | \bm{x}\oplus \bm{y} = z]$.

We present improved Fourier growth bounds for the XOR-fibers of randomized protocols that communicate $d$ bits. For the first level, we show a tight $O(\sqrt{d})$ bound and obtain a new coin theorem, as well as an alternative proof for the tight randomized communication lower bound for the Gap-Hamming problem. For the second level, we show an $d^{3/2} \cdot \polylog(n)$ bound, which improves the previous $O(d^2)$ bound by Girish, Raz, and Tal (ITCS 2021) and implies a polynomial improvement on the randomized communication lower bound for the XOR-lift of the Forrelation problem, which extends the quantum-classical gap for this problem.

Our analysis is based on a new way of adaptively partitioning a relatively large set in Gaussian space to control its moments in all directions. We achieve this via martingale arguments and allowing protocols to transmit real values. We also show a connection between Fourier growth and lifting theorems with constant-sized gadgets as a potential approach to prove optimal bounds for the second level and beyond.
\end{abstract}

\newpage 
\tableofcontents
\thispagestyle{empty}
\newpage
\setcounter{page}{1}

\section{Introduction}\label{sec:intro}

The Fourier spectrum of Boolean functions and their various properties have played an important role in many areas of mathematics and theoretical computer science. In this work, we study a notion called $\ell_1$-Fourier growth, which captures the scaling of the sum of absolute values of the level-$k$ Fourier coefficients of a function. In a nutshell, functions with small Fourier growth cannot aggregate many weak signals in the input to obtain a considerable effect on the output. In contrast, the Majority function, which can amplify weak biases, is an example of a Boolean function with extremely {\em high} Fourier growth.

To formally define Fourier growth, we recall that every Boolean function $f: \binpm^n \to [-1,1]$ can be uniquely represented as a multilinear polynomial 
$$
f(x) = \sum_{S \subseteq [n]} \hat{f}(S) \cdot \prod_{i\in S} x_i
$$ 
where the coefficients of the polynomial $\hat{f}(S)\in \Rbb$ are called the Fourier coefficients of $f$, and they satisfy $\hat{f}(S) = \E[f(\bm{x}) \cdot \prod_{i\in S} \bm{x}_i]$ for a uniformly random $\bm{x} \in \pmone^n$.
The level-$k$ $\ell_1$-Fourier growth of $f$ is the sum of the {\em absolute values} of its level-$k$ Fourier coefficients, 
$$
L_{1,k}(f) := \sum_{S\subseteq[n]:|S|=k}\abs{\hat{f}(S)}.
$$

The study of Fourier growth dates back to the work of Mansour \cite{Mansour95} who used it in the context of learning algorithms.
Since then, several works have shown that upper bounds on the Fourier growth, even for the first few Fourier levels, have applications to pseudorandomness, circuit complexity, and quantum-classical separations.
For example:
\begin{itemize}
\item A bound on the level-one Fourier growth is sufficient to control the advantage of distinguishing biased coins from unbiased ones \cite{agarwal20}.
\item A bound on the level-two Fourier growth already gives  pseudorandom generators \cite{CHLT19}, oracle separations between BQP and PH \cite{RT19,Wu22}, and separations between efficient quantum communication and randomized classical communication \cite{GRT21}. 
\end{itemize}
Meanwhile, Fourier growth bounds have been extensively studied and established for various computational models, including small-width DNFs/CNFs \cite{Mansour95}, $\mathsf{AC}^0$ circuits \cite{Tal17}, low-sensitivity Boolean functions \cite{GSTW16}, small-width branching programs \cite{RSV13,SteinkeVW17,CHRT18,LPV22}, small-depth decision trees \cite{OS07,Tal20,SSW21}, functions related to small-cost communication protocols \cite{GRZ21,GRT21}, low-degree $\mathsf{GF}(2)$ polynomials \cite{CHHL19,CHLT19,blasiok2021fourier}, product tests \cite{Lee19}, small-depth parity decision trees \cite{DBLP:journals/corr/BlaisTW15,GTW21}, low-degree bounded functions \cite{iyer2021tight}, and more.

For any Boolean function $f$ with outputs in $[-1,1]$, the level-$k$ Fourier growth $L_{1,k}(f)$ is at most $\sqrt{\binom nk}$. However, for many natural classes of Boolean functions, this bound is far from tight and not good enough for applications. Establishing better bounds require exploring structural properties of the specific class of functions in question. Even for low Fourier levels, this can be highly non-trivial and tight bounds remain elusive in many cases. For example, for degree-$d$ $\mathsf{GF}(2)$ polynomials
(which well-approximate $\mathsf{AC}^0[\oplus]$ when we set $d=\polylog(n)$ \cite{razborov1987lower,DBLP:conf/stoc/Smolensky87}),
while we know a level-one bound of $L_{1,1}(f)\le O(d)$ due to~\cite{CHLT19}, the current best bound for levels $k\ge2$ is roughly $2^{O(dk)}$ \cite{CHHL19}, whereas the conjectured bound is $d^{O(k)}$. Validating such a bound, even for the second level $k=2$, will imply unconditional pseudorandom generators of polylogarithmic seed length for $\mathsf{AC}^0[\oplus]$ \cite{CHLT19}, a longstanding open problem in circuit complexity and pseudorandomness.

\paragraph*{XOR Functions.}  
In this work, we study the Fourier growth of certain functions that naturally arise from communication protocols for XOR-lifted functions, also referred to as XOR functions. XOR functions are an important and well-studied class of functions in communication complexity with connections to the log-rank conjecture and quantum versus classical separations~\cite{MO09,HHL18,TWXZ13,SZ08,Zha13}. 

In this setting, Alice gets an input $x\in \binpm^n$ and Bob gets an input $y\in \binpm^n$ and they wish to compute $f(x\odot y)$ where $f$ is some partial Boolean function and $x\odot y$ is in the domain of $f$. Here, $x\odot y$ denotes the pointwise product of $x$ and $y$. Given any communication protocol $\Ccal$ that computes an XOR function exactly, the output $\Ccal(x,y)$ of the protocol depends only on the parity $x \odot y$, whenever $f$ is defined on $x \odot y$. This gives a natural motivation to analyze the XOR-fiber of a communication protocol defined below. We note that a similar notion first appeared in an earlier work of Raz \cite{DBLP:journals/cc/Raz95}.

\begin{definition}\label{eqn:fiber}
Let $\Ccal: \pmones \times \pmones \to \pmone$ be any deterministic communication protocol. The XOR-fiber of the communication protocol $\Ccal$ is the function $h\colon\binpm^n\to[-1,1]$ defined at $z\in\binpm^n$ as 
\[    h(z) = \BE_{\lx,\ly \sim \unif}[\Ccal(\lx,\ly) ~|~ \lx\odot \ly = z],\]
where $\odot$ is the entrywise product and $\unif$ is the uniform distribution over $\binpm^n$. 
\end{definition}

We remark that XOR-fiber is the ``inverse'' of XOR-lift of a function: If $\Ccal$ computes the XOR function of $f$, then the XOR-fiber $h$ of $\Ccal$ is equal to $f$ on the domain of $f$.

\vspace{10pt}

In this work, we investigate the Fourier growth of XOR-fibers of small-cost communication protocols and apply these bounds in several contexts. Before stating our results, we first discuss several related works.  

\paragraph*{Related Works.}
Showing optimal Fourier growth bounds for XOR-fibers is a complex undertaking in general and a first step towards this end is to obtain optimal Fourier growth bounds for parity decision trees. This is because a parity decision tree for a Boolean function $f$ naturally gives rise to a structured communication protocol for the XOR-function corresponding to $f$. This protocol perfectly simulates the parity decision tree by having Alice and Bob exchange one bit each to simulate a parity query. Moreover, the XOR-fiber of this protocol exactly computes the parity decision tree. As such, parity decision trees can be seen as a special case of communication protocols, and Fourier growth bounds on XOR-fibers of communication protocols immediately imply Fourier growth bounds on parity decision trees.

Fourier growth bounds for decision trees and parity decision trees are well-studied. It is not too difficult to obtain a level-$k$ bound of $O(d)^k$ for parity decision trees of depth $d$, however, obtaining improved bounds is significantly more challenging. For decision trees of depth $d$ (which form a subclass of parity decision trees of depth $d$), O'Donnell and Servedio~\cite{OS07} proved a tight bound of $O(\sqrt{d})$ on the level-one Fourier growth. 
By inductive tree decompositions, Tal~\cite{Tal20} obtained bounds for the higher levels of the form $L_{1,k}(f)\le\sqrt{d^k\cdot O(\log(n))^{k-1}}$.
This was later sharpened by Sherstov, Storozhenko, and Wu~\cite{SSW21} to the asymptotically tight bound of $L_{1,k}(f)\le\sqrt{\binom dk\cdot O(\log(n))^{k-1}}$ using a more sophisticated layered partitioning strategy on the tree.

When it comes to parity decision trees, despite all the similarities, the structural decomposition approach does not seem to carry over due to the correlations between the parity queries. For parity decision trees of depth $d$, Blais, Tan, and Wan~\cite{DBLP:journals/corr/BlaisTW15} proved a tight level-one bound of $O(\sqrt{d})$. For higher levels, Girish, Tal, and Wu~\cite{GTW21} showed that $L_{1,k}(f)\le\sqrt{d^k\cdot O(k\log(n))^{2k}}$.
These works imply almost tight Fourier growth bounds on the XOR-fibers of structured protocols that arise from simulating decision trees or parity decision trees. 

For the case of XOR-fibers of arbitrary deterministic/randomized communication protocols (which do not necessarily simulate parity decision trees or decision trees), Girish, Raz, and Tal \cite{GRT21} showed an ${O}(d^k)$ Fourier growth\footnote{Technically, \cite{GRT21} only proved a level-two bound (as it suffices for their analysis), but a level-$k$ bound follows easily from their proof approach, as noted by~\cite{GRZ21}} for level-$k$. For level-one and level-two, these bounds are  $O(d)$ and $O(d^2)$ respectively and are sub-optimal --- as mentioned previously, such weaker bounds for parity decision trees are easy to obtain, while obtaining optimal bounds (for parity decision trees) of $O(\sqrt{d})$ for level one and $d \cdot \polylog(n)$ for level two already requires sophisticated ideas.

The bounds in~\cite{GRT21} follow by analyzing the Fourier growth of XOR-fibers of communication rectangles of measure $\approx 2^{-d}$ and then adding up the contributions from all the leaf rectangles induced by the protocol. Such a per-rectangle-based approach cannot give better bounds than the ones in~\cite{GRT21}, while  they also conjectured that the optimal Fourier growth of XOR-fibers of arbitrary protocols should match the growth for parity decision trees. 

Showing the above is a challenging task even for the first two Fourier levels. The difficulty arises primarily since in the absence of a per-rectangle-based argument, one has to crucially leverage cancellations between different  rectangles induced by the communication protocol. In the simpler case of parity decision trees (or protocols that exchange parities), such cancellations are leveraged in \cite{GTW21} by ensuring $k$-wise independence at each node of the tree --- this can be achieved by adding extra parity queries. In a general protocol, the parties can send arbitrary partial information about their inputs and correlate the coordinates in complicated ways that such methods break down.
This is one of the key difficulties we face in this paper. 

\subsection{Main Results}

We prove new and improved bounds on the Fourier growth of the XOR-fibers associated with small-cost protocols for levels $k=1$ and $k=2$.

\begin{theorem}\label{thm:boolean_bound_level_one} 
Let $\Ccal:\binpm^n\times \binpm^n\to \binpm$ be a deterministic communication protocol with at most $d$ bits of communication. Let $h$ be its XOR-fiber as in \Cref{eqn:fiber}. Then,
$L_{1,1}(h) = O\pbra{\sqrt d}$. 
\end{theorem}

\begin{theorem}\label{thm:boolean_bound_level_two}
Let $\Ccal:\binpm^n\times \binpm^n\to \binpm$ be a deterministic protocol communicating at most $d$ bits. Let $h$ be its XOR-fiber as in \Cref{eqn:fiber}. Then, $L_{1,2}(h) = O\pbra{d^{3/2} \log^3(n)}$.
\end{theorem}

Our bounds in \Cref{thm:boolean_bound_level_one,thm:boolean_bound_level_two} extend directly to randomized communication protocols. This is because $L_{1,k}$ is convex and any randomized protocol is a convex combination of deterministic protocols with the same cost. Moreover, we can use Fourier growth reductions, as described in \Cref{sec:applications_gadgets}, to demonstrate that these bounds apply to general constant-sized gadgets $g$ and the corresponding $g$-fiber.


Our level-one and level-two bounds improve previous bounds in \cite{GRT21} by polynomial factors. Additionally, our level-one bound is tight since a deterministic protocol with $d+1$ bits of communication can compute the majority vote of $x_1 \cdot y_1, \ldots, x_d \cdot y_d$, which corresponds to $h(z) = \mathrm{MAJ}(z_1, \ldots, z_{d})$ with $L_{1,1}(h) = \Theta(\sqrt{d})$. Furthermore, as we discuss later in \Cref{sec:apps}, level-one and level-two bounds are already sufficient for many interesting applications.

In terms of techniques, our analysis presents a key new idea that enables us to exploit cancellations between different rectangles induced by the protocol. This idea involves using a novel process to adaptively partition a relatively large set in Gaussian space, which enables us to control its $k$-wise moments in all directions --- this can be thought of as a spectral notion of almost $k$-wise independence. We achieve this by utilizing martingale arguments and allowing protocols to transmit \emph{real values} rather than just discrete bits. This notion and procedure may be of independent interest. See \Cref{sec:overview} for a detailed discussion.

\subsection{Applications and Connections}
\label{sec:apps}

Our main theorem has applications to XOR functions, and in more generality to functions lifted with constant-sized gadgets. In this setting, there is a simple gadget $g:\Sigma \times \Sigma \to \pmone$ and a Boolean function $f$ defined on inputs $z\in \pmone^n$. 
The lifted function $\flift$ is defined on $n$ pairs of symbols $(x_1, y_1), \ldots, (x_n, y_n) \in \Sigma \times \Sigma$ such that $(\flift)(x, y) = f(g(x_1, y_1), \ldots, g(x_n, y_n))$.
The function $\flift$ naturally defines a communication problem where Alice is given $x = (x_1, \ldots, x_n)$, Bob is given $y = (y_1, \ldots, y_n)$, and they are asked to compute $(\flift)(x, y)$. 

Since XOR functions are functions lifted with the XOR gadget, our main theorem implies lower bounds on the communication complexity of specific XOR functions. Additionally, we also show connections between XOR-lifting and lifting with any constant-sized gadget. Next, we describe these lower bounds and connections, with further context.


\subsubsection{The Coin Problem and the Gap-Hamming Problem}\label{sec:applications_level_one}

The coin problem studies the advantage that a class of Boolean functions has in distinguishing biased coins from unbiased ones. 
More formally, let $\Fcal$ be a class of $n$-variate Boolean functions.
Let $\rho \in[-1,1]$ and $\biased{\rho}$ denote the product distribution over $\binpm^n$ where each coordinate has expectation $\rho$. 
The Coin Problem asks what is the maximum advantage that functions in $\Fcal$ have in distinguishing $\biased{\rho}$ from the uniform distribution $\biased{0}$. 

This quantity essentially captures how well $\Fcal$ can approximate threshold functions, and in particular, the majority function. The coin problem has been studied for various models of computation including branching programs~\cite{BV10}, $\mathsf{AC}^0$ and $\mathsf{AC}^0[\oplus]$ circuits~\cite{CGR14,LSSTV19}, product tests~\cite{Lee18}, and more. 
Recently, Agrawal \cite{agarwal20} showed that the coin problem is closely related to the level-one Fourier growth of functions in $\Fcal$.

\begin{lemma}[{\cite[Lemma 3.2]{agarwal20}}] \label{lem:coin_problem}
Assume that $\Fcal$ is closed under restrictions and satisfies $L_{1,1}(f) \le t$ for all $f\in \Fcal$. Then, for all $\rho\in(-1,1)$ and $f\in \Fcal$, 
\[ 
\abs{\E_{z\sim \biased{\rho}}[f(z)]-\E_{z\sim \biased{0}}[f(z)]}\le \ln\pbra{\tfrac{1}{1-|\rho|}}\cdot t.
\]
\end{lemma}

Note that communication protocols of small cost are closed under restrictions, so are their XOR-fibers (see \cite[Lemma 5.5]{GRT21}).
By noting that  $\ln\pbra{\frac1{1-|\rho|}} \approx |\rho|$ for small values of $\rho$, we obtain the following corollary.\footnote{Here we also use the fact that the upper bound $O(|\rho|\cdot \sqrt{d})$ is vacuous for large enough $\rho$ as it is larger than $1$.}
We also remark that, using the Fourier growth reductions (see \Cref{sec:applications_gadgets}), \Cref{thm:coin_problem} can be established for general gadgets of small size.

\begin{theorem}\label{thm:coin_problem} 
Let $h$ be the XOR-fiber of a protocol with total communication $d$. Then for all $\rho$,
\[ \abs{\E_{z\sim \biased{\rho}}[h(z)]-\E_{z\sim \biased{0}}[h(z)]} \le O\!\pbra{ |\rho|\cdot \sqrt{d}} .\]
\end{theorem}

In particular, consider the following distinguishing task: 
Alice and Bob either receive two uniformly random strings in $\binpm^n$ or they receive two uniformly random strings in $\binpm^n$ conditioned on their XOR distributed according to $\biased{\rho}$ for $\rho = 1/\sqrt{n}$ (the latter is often referred to as \emph{$\rho$-correlated strings}).
\Cref{thm:coin_problem} implies that any protocol communicating $o(n)$ bits cannot distinguish these two distributions with constant advantage. This is essentially a communication lower bound for the well-known Gap-Hamming Problem.

\paragraph*{The Gap-Hamming Problem.}
In the Gap-Hamming Problem, Alice and Bob receive strings $x,y\in\binpm^n$ respectively and they want to distinguish if $\ip{x}{y} \le -\sqrt{n}$ or $\ip{x}{y} \ge \sqrt{n}$.

This is essentially the XOR-lift of the Coin Problem with $\rho=\pm 1/\sqrt{n}$ because the distribution of $(x,y)$ conditioned on $x \odot y\sim\biased{\rho}$ with $\rho=-1/\sqrt{n}$ and $\rho=1/\sqrt{n}$ is mostly supported on the \textsc{Yes} and \textsc{No} instances of Gap-Hamming respectively.
Thus immediately from \Cref{thm:coin_problem}, we derive a new proof for the $\Omega(n)$ lower bound on the communication complexity of the Gap-Hamming Problem.
The proof is deferred to \Cref{app:thm:gap_hamming}.

\begin{theorem}\label{thm:gap_hamming}
The randomized communication complexity of the Gap-Hamming Problem is $\Omega(n)$.	
\end{theorem}

We note that there are various different proofs~\cite{DBLP:journals/siamcomp/ChakrabartiR12,DBLP:journals/toc/Sherstov12,DBLP:journals/cjtcs/Vidick12,RY22} that obtain the above lower bound but the perspective taken here is perhaps conceptually simpler: (1) Gap-Hamming is essentially the XOR-lift of the Gap-Majority function, and (2) any function that approximates the Gap-Majority function must have large level-one Fourier growth, whereas XOR-fibers of small-cost protocols have small Fourier growth. 

\subsubsection{Quantum versus Classical Communication Separation via Lifting}
\label{sec:applications_level_two}

One natural approach to proving quantum versus classical separations in communication complexity is via lifting: 
Consider a function $f$ separating quantum and classical query complexity and lift it using a gadget $g$. Naturally, an algorithm computing $f$ with few queries to $z$ can be translated into a communication protocol computing $\flift$ where we replace each query to a bit $z_i$ with a short conversation that allows the calculation of $z_i=g(x_i, y_i)$. \Goos, Pitassi, and Watson~\cite{GPW20} showed that for randomized query/communication complexity and for various gadgets, this is essentially the best possible. Such results are referred to as {\em lifting theorems}. 

Lifting theorems apply to different models of computation, such as deterministic decision trees~\cite{RM99,GPW15}, randomized decision trees~\cite{GPW20, CFKMP19}, and more. A beautiful line of work shows how to ``lift''  many lower bounds in the query model to the communication model
\cite{RM99,GPW15,GLMWZ15,Goos15,RezendeNV16,HHL18,WYY17,CKLM19,KMR17,SZ09,Sher11,RS10,RPRC16,GKPW19,LRS15}. 
For quantum query complexity, only one direction (considered the ``easier'' direction) is known:
Any quantum query algorithm for $f$ can be translated to a communication protocol for $\flift$ with a small logarithmic overhead \cite{BCW}. 
It remains widely open whether the other direction holds as well.
However, this query-to-communication direction for quantum, combined with the communication-to-query direction for classical, is already sufficient for lifting quantum versus classical separations from the query model to the communication model.

One drawback of this approach to proving communication complexity separations is that the state-of-the-art lifting results~\cite{CFKMP19,lovett2022lifting} work for gadgets with alphabet size at least $n$ (recall that $n$ denotes $f$'s input length) and it is a significant challenge to reduce the alphabet size to $O(1)$ or even $\polylog(n)$.
These large gadgets will usually result in larger overheads in terms of communication rounds, communication bits, and computations for both parties.
As demonstrated next, lifting with simpler gadgets like XOR allows for a simpler quantum protocol for the lifted problem.

\paragraph*{Lifting Forrelation with XOR.}
The Forrelation function introduced by \cite{Aaronson10} is defined as follows: on input $x=(x_1,x_2)\in\binpm^{n}$ where $n$ is a power of $2$,
$$
\Forr(x)=\frac{2}{n}\ip{H x_1}{x_2},
$$
where $H$ denotes the $(n/2)\times (n/2)$ (unitary) Hadamard matrix.

Girish, Raz, and Tal~\cite{GRT21} studied the XOR-lift of the Forrelation problem and obtained new separations between quantum and randomized communication protocols. 
In more detail, they considered the partial function\footnote{We are overloading the notation here: technically, $\Forr \circ \mathrm{XOR}$ is the XOR-lift of the partial boolean function which on input $x$ outputs $1$ if $\Forr(x)$ is large and $-1$ if $\Forr(x)$ is small.} $\Forr \circ \mathrm{XOR} \colon\binpm^{n}\times\binpm^{n}\to\binpm$ defined as
$$
\Forr \circ \mathrm{XOR}(x,y)=\begin{cases}
1 & \Forr(x\odot y)\ge\frac1{200\ln(n/2)},\\
-1 & \Forr(x\odot y)\le\frac1{400\ln(n/2)},
\end{cases}
$$
and showed that if Alice and Bob use a randomized communication protocol, then they must communicate at least $\tilde{\Omega}(n^{1/4})$ bits to compute $\Forr \circ{\mathrm{XOR}}$; while it can be solved by two entangled parties in the quantum simultaneous message passing model with a $\polylog(n)$-qubit communication protocol and additionally the parties can be implemented with efficient quantum circuits. 

The lower bound in~\cite{GRT21} was obtained from a second level Fourier growth bound (higher levels are not needed) on the XOR-fiber of classical communication protocols.
Our level-two bound strengthens their bound and immediately gives an improved communication lower bound.

\begin{theorem}\label{thm:rcc_xor-lifts_forrelation}
The randomized communication complexity of $\Forr \circ \mathrm{XOR}$ is $\tilde{\Omega}(n^{1/3})$.
\end{theorem}

\Cref{thm:rcc_xor-lifts_forrelation} above gives an $\polylog(n)$ versus $\tilde\Omega(n^{1/3})$ separation between the above quantum communication model and the randomized two-party communication model, improving upon the $\polylog(n)$ versus $\tilde\Omega(n^{1/4})$ separation from \cite{GRT21}. 
We emphasize that our separations are for players with \emph{efficient quantum} running time, where the only prior separation was shown by the aforementioned work~\cite{GRT21}.
Such efficiency features can also benefit real-world implementations to demonstrate quantum advantage in experiments; for instance, one such proposal was introduced recently by Aaronson, Buhrman, and Kretschmer~\cite{aaronson2023qubit}.
Without the efficiency assumption, a better $\polylog(n)$ versus $\tilde\Omega(\sqrt n)$ separation is known \cite{DBLP:journals/tit/Gavinsky20} (see \cite[Section 1.1]{GRT21} for a more detailed comparison). Optimal Fourier growth bounds of $d \cdot \polylog(n)$ for level two, which we state later in \Cref{conjecture:higher_level},  would also imply such a separation with XOR-lift of Forrelation.

\paragraph*{Lifting $k$-Fold Forrelation with XOR.} 

$k$-Fold Forrelation~\cite{DBLP:journals/siamcomp/AaronsonA18} is a generalization of the Forrelation problem and was originally conjectured to be a candidate that exhibits a maximal separation between quantum and classical query complexity. 
In a recent work, \cite{BS21} showed that the randomized query complexity of $k$-Fold Forrelation is $\tilde{\Omega}(n^{1-1/k})$, confirming this conjecture, and a similar separation was proven in \cite{SSW21} for variants of $k$-Fold Forrelation. 
These separations, together with lifting theorems with the \emph{inner product} gadget~\cite{CFKMP19}, imply an $O(k\log(n))$ vs $\tilde\Omega(n^{1-1/k})$ separation between two-party quantum and classical communication complexity, where additionally, the number of rounds\footnote{We remark that for $k=2$, this is exactly the XOR-lift of the Forrelation problem and can even be computed in the quantum simultaneous model, as shown in \cite{GRT21}.} in the two-party quantum protocol is $2\cdot\lceil k/2\rceil$.

Replacing the inner product gadget with the XOR gadget above would yield an improved quantum-classical communication separation where the gadget is simpler and the number of rounds required by the quantum protocol to achieve the same quantitative separation is reduced by half. 
Bansal and Sinha~\cite{BS21} showed that for any computational model, small Fourier growth for the first $O(k^2)$-levels implies hardness of $k$-Fold Forrelation in that particular model. Thus, in conjunction with their results, to prove the above XOR lifting result for the $k$-Fold Forrelation problem, it suffices to prove the following Fourier growth bounds for XOR-fibers. 
\begin{conjecture}\label{conjecture:higher_level}
Let $\Ccal:\binpm^n\times \binpm^n\to \binpm$ be a deterministic communication protocol with at most $d$ bits of communication. Let $h$ be its XOR-fiber as in \Cref{eqn:fiber}. Then for all $k\in \Nbb$, we have that
$L_{1,k}(h) \le (\sqrt{d} \cdot \poly(k,\log(n)))^k$.
\end{conjecture}

Note that these bounds are consistent with the Fourier growth of parity decision trees (or protocols that only send parities) as shown in \cite{GTW21}.

We prove the above conjecture for the case $k=1$ and make progress for the case $k=2$. While our techniques can be extended to higher levels in a straightforward manner, the bounds  obtained are farther from the conjectured ones. Thus, we decided to defer dealing with higher levels to future work as we believe one needs to first prove the \emph{optimal} bound for level $k=2$.

In the next subsection, we give another motivation to study the above conjecture by showing a connection to lifting theorems for constant-sized gadgets.

\subsubsection{General Gadgets and Fourier Growth from Lifting}\label{sec:applications_gadgets}

Our main results are Fourier growth bounds for XOR-fibers, which corresponds to XOR-lifts of functions.
To complement this, we show that similar bounds hold for general lifted functions.

Let $g\colon\Sigma\times\Sigma\to\binpm$ be a gadget and $\Ccal\colon\Sigma^n\times\Sigma^n\to\binpm$ be a communication protocol.
Define the $g$-fiber of $\Ccal$, denoted by $\Ccal_{\downarrow g}\colon\binpm^n\to[-1,1]$, as
$$
\Ccal_{\downarrow g}(z)=\E\sbra{\Ccal(\xbm,\ybm)\mid g(\xbm_i,\ybm_i)=z_i,~\forall i},
$$
where $\xbm$ and $\ybm$ are uniform over $\Sigma$.
We use $L_{1,k}(g,d)$ to denote the upper bound of the level-$k$ Fourier growth for the $g$-fibers of protocols with at most $d$ bits of communication.
Using this notation, the XOR-fiber of $\Ccal$ is simply $\Ccal_{\downarrow\mathrm{XOR}}$, and our main results \Cref{thm:boolean_bound_level_one,thm:boolean_bound_level_two} can be rephrased as
$$
L_{1,1}(\mathrm{XOR},d)\le O\pbra{\sqrt d}
\quad\text{and}\quad
L_{1,2}(\mathrm{XOR},d)\le O\pbra{d^{3/2}\log^3(n)}.
$$

In \Cref{sec:gadget}, we relate $L_{1,k}(g,d)$ to $L_{1,k}(\mathrm{XOR},d)$, and the main takeaway is, in the study of Fourier growth bounds, constant-sized gadgets are all equivalent.

\begin{theorem}[Informal, see \Cref{thm:xor_to_g} and \Cref{thm:g_to_xor}]\label{thm:informal_general_gadget}
Let $g\colon\Sigma\times\Sigma\to\binpm$ be a ``balanced'' gadget.
Then 
$$
|\Sigma|^{-k}\cdot L_{1,k}(\mathrm{XOR},d)\le L_{1,k}(g,d)\le|\Sigma|^k\cdot L_{1,k}(\mathrm{XOR},d).
$$
\end{theorem}

\Cref{thm:informal_general_gadget} also proposes a different approach towards \Cref{conjecture:higher_level}: it suffices to establish tight Fourier growth bound for $g$-fibers for some constant-sized (actually, polylogarithmic size suffices) gadget $g$, and then apply the reduction.
The benefit of switching to a different gadget is that we can perhaps first prove a lifting theorem, and then appeal to the known Fourier growth bounds of (randomized) decision trees \cite{Tal20,SSW21}.
See \Cref{sec:lift} for detail.

As mentioned earlier, lifting theorems show how to simulate communication protocols of cost $d$ for lifted functions with decision trees of depth at most $O(d)$ (see e.g., \cite{GPW20}).
A problem at the frontier of this fruitful line of work has been establishing lifting theorems for decision trees with constant-sized gadgets.
Note that the XOR gadget itself cannot have such a generic lifting result: Indeed, the parity function serves as a counterexample. 
Nevertheless, it is speculative that some larger gadget works, which suffices for our purposes.\footnote{In terms of the separations between quantum and classical communication, even restricted lifting results for the specific outer function being the Forrelation function would suffice.}
On the other hand, for lifting from \emph{parity} decision trees, we do know an XOR-lifting theorem~\cite{HHL18}. However, it only holds for deterministic communication protocols and has a sextic blowup in the cost. 

Thus, one can see \Cref{conjecture:higher_level} as either a further motivation for establishing lifting results for decision trees with constant-sized gadgets, or as a necessary milestone before proving such lifting results.

\subsubsection{Pseudorandomness for Communication Protocols}\label{sec:applications_prg}

We say $G\colon\binpm^\ell\to\binpm^n\times\binpm^n$ is a pseudorandom generator (PRG) for a (randomized) communication protocol $\Ccal\colon\binpm^n\times\binpm^n\to[-1,1]$ with error $\eps$ and seed length $\ell$ if
$$
\abs{\E_{\xbm,\ybm\sim\unif}[\Ccal(\xbm,\ybm)]-\E_{\rbm\sim\binpm^\ell}[\Ccal(G(\rbm))]}\le\eps.
$$
\cite{INW94} showed that for the class of protocols sending at most $d$ communication bits, there exists an explicit PRG of error $2^{-d}$ and seed length $n+O(d)$ from expander graphs.
Note that the overhead $n$ is inevitable even if the protocol is only sending one bit, since it can depend arbitrarily on Alice/Bob's input.

Combining \Cref{conjecture:higher_level} 
and the PRG construction from \cite[Theorem 4.5]{CHHL19}, we would obtain a completely different explicit PRG for this class with error $\eps$ and seed length $n+d\cdot\polylog(n/\eps)$.

\paragraph*{Paper Organization.}
An overview of our proofs is given in \Cref{sec:overview}.
In \Cref{sec:prelim} we define necessary notation and recall useful inequalities.
\Cref{sec:fourier_via_martingale} explains a way to associate the Fourier growth to a martingale process. The proof of level-one bound (\Cref{thm:boolean_bound_level_one}) is given in \Cref{sec:proof_of_level_one}, and the level-two bound (\Cref{thm:boolean_bound_level_two}) in \Cref{sec:proof_of_level_two}.
The Fourier growth reductions between general gadgets are presented in \Cref{sec:gadget}.
The future directions are discussed in \Cref{sec:future}.
Missing proofs can be found in the appendix.
\section{Proof Overview}\label{sec:overview}

We first briefly outline the proof strategy, which consists of three main components:  

\begin{itemize}
    \item First, we show that the level-one bound can be characterized as the expected absolute value of a martingale defined as follows: Consider the random walk induced on the protocol tree when Alice and Bob are given inputs $\lx$ and $\ly$ uniformly from $\pmones$.
    Let $\X^{(t)} \times \Y^{(t)}$ be the rectangle associated with the random walk at time $t$. The martingale process tracks the inner product $\ip{\com(\X^{(t)})}{\com(\Y^{(t)})}$ where $\com(\X^{(t)}) = \BE\sbra{\lX \mid \lX \in \X^{(t)}}$ and $\com(\Y^{(t)}) = \BE\sbra{\lY \mid \lY \in \Y^{(t)}}$ are Alice's and Bob's center of masses. 
    \item Second, to bound the value of the martingale, it is necessary to ensure that neither $\X^{(t)}$ nor $\Y^{(t)}$ become excessively elongated in any direction during the protocol execution. To measure the length of $\X^{(t)}$ in a particular direction $\theta \in \mathbb{S}^{n-1}$, we calculate the variance $\mathbb{V}\mathrm{ar}\sbra{\abra{\xbm,\theta} \mid \lx\in\X^{(t)}}$, i.e. the variance of a uniformly random $\lx \in \X^{(t)}$ in the direction $\theta$. If the set is not elongated in any direction, this can be thought of as a spectral notion of almost pairwise independence. Such a notion also generalizes to almost $k$-wise independence by considering higher moments. 
 
    To achieve the property that the sets are not elongated, one of the main novel ideas in our paper is to modify the original protocol to a new one that incorporates additional cleanup steps where the parties communicate \emph{real values} $\abra{\xbm,\theta}$. Through these communication steps, the sets $\X^{(t)}$ and $\Y^{(t)}$ are recursively divided into affine slices along problematic directions. 
    
    
    \item Last, one needs to show that the number of cleanup steps are small in order to bound the value of the martingale for the new protocol. This is the most involved part of our proof and requires considerable effort because the cleanup steps are real-valued and adaptively depend on the entire history, including the previous real values communicated.
\end{itemize}

The strategy outlined above also generalizes to level-two Fourier growth by considering higher moments and sending values of quadratic forms in the inputs. We also remark that since we view the sets $\X^{(t)}$ and $\Y^{(t)}$ above as embedded in $\R^n$ and allow the protocol to send real values, it is more natural for us to work in Gaussian space by doing a standard transformation. The rotational invariance of the Gaussian space also seems to be essential for us to obtain optimal level-one bound without losing additional polylogarithmic factors. 

We now elaborate on the above components in detail and also highlight the differences between the level-one and level-two settings. For conciseness, in the following overview we use $f\lesssim g$ to denote $f=O(g)$ and $f\gtrsim g$ to denote $f=\Omega(g)$ where $O$ and $\Omega$ only hide absolute constants.

\subsection{Level-One Fourier Growth}\label{sec:overview_level_one}

The level-one Fourier growth of the XOR-fiber $h$ is given by 
\[ 
L_{1,1}(h) = \sum_{i=1}^n \abs{\hat{h}(\{i\})} 
= \sum_{i=1}^n \abs{\BE_{\lZ \sim \unif}[h(\lZ) \lZ_i]} 
= \sum_{i=1}^n \abs{\BE_{\lX,\lY \sim \unif}[\Ccal(\lX,\lY)\lX_i\lY_i]}.
\]

To bound the above, it suffices to bound $\sum_{i=1}^n \eta_i \cdot \BE [\Ccal(\lX,\lY)\lX_i\lY_i]$ for any sign vector $\eta \in \pmones$. 
Here for simplicity we assume $\eta_i\equiv1$ and the probability of reaching every leaf is $\approx2^{-d}$.

\paragraph*{A Martingale Perspective.} 
To evaluate the quantity $\sum_{i=1}^n \BE [\Ccal(\lX,\lY)\lX_i\lY_i]$, consider a random leaf $\bell$ of the protocol and let $\Xell \times \Yell$ be the corresponding rectangle. Since the leaf determines the answer of the protocol, denoted by $\Ccal(\bell)$, the quantity above equals
\[ 
\sum_{i=1}^n \BE_{\bell}\sbra{\Ccal(\bell) \cdot \BE [\lX_i \mid \lX \in \Xell] \cdot \BE[\lY_i \mid \lY \in \Yell]} 
= \BE_{\bell}[ \Ccal(\bell) \cdot \ip{\com(\Xell)}{\com(\Yell)}] 
\le \BE_{\bell}[ |\ip{\com(\Xell)}{\com(\Yell)}|],
\]
where $\com(\Xell) = \BE\sbra{\lX \mid \lX \in \Xell}$ and $\com(\Yell) = \BE\sbra{\lY \mid \lY \in \Yell}$ are the center of masses of the rectangle. 
Our goal is to bound the magnitude of the random variable $\lZ = \ip{\com(\Xell)}{\com(\Yell)}$. 

We shall show that $\BE_{\bell}[|\lZ|] \lesssim\sqrt{d}$. Note that $|\lZ|$ can be as large as $d$ in the worst case --- for instance if the first $d$ coordinates of $\Xell$ and $\Yell$ are fixed to the same value --- thus we cannot argue for each leaf separately.  

To analyze it for a random leaf, we first characterize the above as a martingale process using the tree structure of the protocol. 
The martingale process is defined as $\pbra{\supZ{t}}_t$ where $\supZ{t}:=\abra{\com(\supX{t}),\com(\supY{t})}$ tracks the inner product between the center of masses $\com(\supX{t})$ and $\com(\supY{t})$ of the current rectangle $\supX{t} \times \supY{t}$ at step $t$. 
Denote the martingale differences by $\Delta \supZ{t+1} = \supZ{t+1} - \supZ{t}$ and note that if in the $t^{\text{th}}$ step Alice sends a message, then 
\[ \Delta \supZ{t+1} = \ip{\Delta \com(\supX{t+1})}{ \com(\supY{t+1})},\]
where $\Delta \com(\supX{t+1}) = \com(\supX{t+1}) - \com(\supX{t})$ is the change in Alice's center of mass. A similar expression holds if Bob sends a message.
Then it suffices to bound the expected quadratic variation (see \Cref{sec:prelim}) since
\begin{equation}\label{eqn:martingale}
     \pbra{\BE\sbra{\abs{\supZ{d}}}}^2 \le \BE\sbra{\pbra{\supZ{d}}^2} =  \BE\sbra{\sum_{t = 0}^{d-1}\pbra{\Delta \supZ{t+1}}^2 },
\end{equation}
where the equality holds due to the martingale property: $\BE\sbra{\Delta \supZ{t+1} \mid 
\supZ{1}, \ldots \supZ{t}} = 0$.

To obtain the desired bound, we need to bound the expected quadratic variation by $O(d)$. Note that it could be the case that a single $\Delta \supZ{t+1}$  scales like $\sqrt{d}$. For instance, if Bob first announces his first $d$ coordinates, $y_1, \ldots, y_d$, and then Alice sends a majority of $x_1 \cdot y_1, \ldots, x_d \cdot y_d$, then in the last step Alice's center of mass $\com(\supX{t+1})$ changes by $\approx1/\sqrt{d}$ in each of the first $d$ coordinates, and the inner product with Bob's center of mass changes by $\approx \sqrt{d}$ in a single step.

Such cases make it difficult to directly control the individual step sizes of the martingale and we will only be able to obtain an amortized bound. It turns out, as we explain later, that such an amortized bound on the martingale can be obtained if Alice and Bob's sets are not elongated in any direction. Therefore, we will transform the original protocol into a \emph{clean} protocol by introducing real communication steps that slice the elongated directions. For this, it will be convenient to work in Gaussian space which also turns out to be essential in proving the optimal $O(\sqrt{d})$ bound.

\paragraph*{Protocols in Gaussian Space.} 
A communication protocol in Gaussian space takes as inputs $\lx, \ly \in \Rbb^n$ where $\lx, \ly$ are independently sampled from the Gaussian distribution $\gamma_n$. One can embed the original Boolean protocol in the Gaussian space by running the protocol on the uniformly distributed Boolean inputs $\sgn(\lx)$ and $\sgn(\ly)$ where $\sgn(\cdot)$ takes the sign of each coordinate. Note that any node of the protocol tree in the Gaussian space corresponds to a rectangle $X \times Y$ where $X, Y \subseteq \Rbb^n$. 
Abusing the notation and defining their \emph{Gaussian} centers of masses as $\com(X) = \BE_{\lx \sim \gamma_n}\sbra{\lX \mid \lX \in X}$ and $\com(Y) = \BE_{\ly \sim \gamma_n}\sbra{\lY \mid \lY \in Y}$, one can associate the same martingale $(\supZ{t})_t$ with the protocol in the Gaussian space:
\[ 
\supZ{t} = \ip{\com(\supX{t})}{ \com(\supY{t})}.
\] 
It turns out that bounding the quadratic variation of this martingale suffices to give a bound on $L_{1,2}(h)$ (see \Cref{sec:fourier_via_martingale}), so we will stick to the Gaussian setting.
We now describe the ideas behind the cleanup process so that the step sizes can be controlled more easily.  

\paragraph*{Cleanup with Real Communication.} 
The cleanup protocol runs the original protocol interspersed with some cleanup steps where Alice and Bob send real values. As outlined before, one of the goals of these cleanup steps is to ensure that the sets are not elongated in any direction, in order to control the martingale steps.
In more detail, recall that we want to control
\[
\BE\sbra{(\Delta \supZ{t+1})^2 \mid \supZ{1},\ldots,\supZ{t}} = \BE\sbra{\abra{\Delta \com(\supX{t+1}), \com(\supY{t+1})}^2 \mid \supZ{1},\ldots,\supZ{t}}
\]
in the $t^{\text{th}}$ step where Alice speaks. There are two key underlying ideas for the cleanup steps:

\begin{itemize}
    \item \textbf{Gram-Schmidt Orthogonalization:} 
    At each round, if the current rectangle is $\X \times \Y$, before Alice sends the actual message, she sends the inner product $\ip{x}{\com({\Y})}$ between her input and Bob's current center of mass $\com({\Y})$. This partitions Alice's set $\X$ into affine slices orthogonal to Bob's current center of mass $\com(\Y)$. 
    Thus the change in Alice's center of mass in later rounds is orthogonal to $\com(\Y)$ since it only takes place inside the affine slice.
    
    Recall that the martingale $\supZ{t}$ is the inner product of Alice and Bob's center of masses, and Bob's center of mass does not change when Alice speaks.
    The original communication steps now do not contribute to the martingale and only the steps where the inner products are revealed do. In particular, if $t_{\mathrm{prev}} < t$ are two consecutive times where Alice revealed the inner product, then the change in Alice's center of mass is orthogonal to change in Bob's center of mass between time $t_{\mathrm{prev}}$ and $t$. Thus, conditioned on the rectangle $\supX{t} \times \supY{t}$ fixed by the messages until time $t$, we have, by Jensen's inequality,
    \begin{align}\label{eqn:overview}
    \BE\sbra{(\Delta \supZ{t+1})^2 \mid \supX{t},\supY{t}} 
    &= \BE\sbra{\ip{\Delta \com(\supX{t+1})}{ \com(\supY{t})- \com(\supY{t_\mathrm{prev}})}^2 \mid\supX{t}, \supY{t}}\notag\\
    &\le \BE\sbra{\ip{\lX - \com(\supX{t})}{ \com(\supY{t})- \com(\supY{t_{\mathrm{prev}}})}^2 \mid\supX{t}, \supY{t}}.
    \end{align}

    Note that the quantity on the right-hand side above is of the form $\ip{\lX - \BE[\lX]}{v}$. In other words, it is the variance of the random vector $\lX$ along direction $v$. To maintain a bound on this quantity, we introduce the notion of ``not being elongated in any direction''.    
    
    \item \textbf{Not elongated in any direction:}  We define the following notion to capture the fact that the random vector is not elongated in any direction: we say that a mean-zero random vector $\lX' = \lX - \BE[\lX]$ in $\R^n$ is $\lambda$-\emph{pairwise clean}, if for every $v \in \R^n$,
    \begin{equation}\label{eq:clean_definition}
    \BE\sbra{\abra{\lX' ,v}^2} \le \lambda\cdot \|v\|^2, 
    \end{equation}
    or equivalently, the operator norm of the covariance matrix $\BE[\lX'\lX'^\top]$ is at most $\lambda$. This can be considered a spectral notion of almost pairwise independence, since the pairwise moments are well-behaved in every direction. 

\end{itemize}


If the input distribution conditioned on Alice's set $\supX{t}$ is $O(1)$-pairwise clean, we say that her set is \emph{pairwise clean}. Based on the above ideas, after Alice sends the initial message, if her set is not yet clean, she partitions it recursively by taking affine slices and transmitting real values. More precisely, while there is direction $\theta\in \Sbb^{n-1}$ violating \Cref{eq:clean_definition}, Alice does a cleanup of her set by sending the inner product $\ip{x}{\theta}$. 
This direction is known to Bob as it only depends on Alice's current space. In addition, this cleanup does not contribute to the martingale \emph{in the future} because the inner product along this direction is fixed now.


The resulting protocol is pairwise clean in the sense that at each step\footnote{We remark that the sets are only clean at intermediate steps where a cleanup phase ends, but we show that because of the orthogonalization step, the other steps do not contribute to the value of the martingale.}, Alice's current set is pairwise clean.
Similar arguments work for Bob.
    
Let $\D$ be the total number of communication rounds including all the cleanup steps. Then, by the above argument, and denoting by $(\btau_m)_m$ and $(\btau'_m)_m$ the indices of the inner product steps for Alice and Bob, we can ultimately bound
\begin{align}\label{eqn:qv}
\BE\sbra{(\supZ{\D})^2} & \lesssim \BE\sbra{\sum_{m} \vabs{\com(\supX{\btau_m})- \com(\supX{\btau_{m-1}})}^2 + \vabs{\com(\supY{\btau'_m})- \com(\supY{\btau'_m-1})}^2}\notag\\
&= \BE\sbra{\vabs{\com(\supX{\D})}^2 + \vabs{\com(\supY{\D})}^2},
\end{align}
where again, the last equality follows from the martingale property. 
The right hand side above can be bounded by the expected number of communication rounds $\BE[\D]$ using the level-one inequality (see \Cref{thm:level_k_ineq}) --- this inequality bounds the Euclidean norm of the center of mass of a set in terms of its Gaussian measure.
    
\paragraph*{Expected Number of Cleanup steps.} 
Since the original communication only consists of $d$ rounds, the analysis essentially reduces to bounding the expected number of cleanup steps by $O(d)$, which is technically the most involved part of the proof.


It is implicit in the previous works on the Gap-Hamming Problem \cite{DBLP:journals/siamcomp/ChakrabartiR12,DBLP:journals/cjtcs/Vidick12} that large sets are not elongated in many directions: if a set $X \subseteq \Rbb^n$ has Gaussian measure $\approx2^{-d}$, then for a random vector $\lx$ sampled from $X$, there are at most $m\lesssim d$ orthogonal directions $\theta_1, \ldots, \theta_m$ such that $\BE[\ip{\lX'}{\theta_i}^2]\gtrsim1$ where $\lX' = \lX - \BE[\lX]$.
This is a consequence of the fact that the expectation of $\lQ = \sum_{i=1}^m \ip{\lX'}{\theta_i}^2$ can be bounded by $O(d)$ provided that $X$ has measure $\approx2^{-d}$.

The above argument suggests that maybe we can clean up the set $X$ along these $O(d)$ bad orthogonal directions.
However this is not enough for our purposes: after taking an affine slice, the set may not be clean in a direction where it was clean before.
Moreover, since the parties take turns to send messages and clean up, the bad directions will also depend on the entire history of the protocol, including the previous real and Boolean communication.
This adaptivity makes the analysis more delicate and to prove the optimal bound we crucially utilize the rotational symmetry of the Gaussian distribution.
Indeed, the fact that a large set is not elongated in many directions also holds even when we replace the Gaussian distribution with the uniform distribution on $\pmones$, but it is unclear how to obtain an optimal level-one bound using the latter.

In the final protocol, since the parties only send Boolean bits and linear forms of their inputs, conditioned on the history of the martingale, one can still say what the distribution of the next cleanup $\abra{\lx, \theta}$ looks like, as the Gaussian distribution is well-behaved under linear projections. We then use martingale concentration and stopping time arguments to show that the expected number of cleanup steps is indeed bounded by $O(d)$ even if the cleanup is adaptive.

We make two remarks in passing: 
First, we can also prove the optimal level-one bound using information-theoretic ideas but they do not seem to generalize to the level-two setting, so we adopt the alternative concentration-based approach here and they are similar in spirit. 
Second, it is possible from our proof approach (in particular, the approach for level two described next) to derive a weaker upper bound of $\sqrt{d}\cdot\polylog(n)$ for the level one while directly working with the uniform distribution on the hypercube.

\subsection{Level-Two Fourier Growth}\label{sec:overview_level_two}


We start by noting that the level-two Fourier growth of the XOR-fiber $h$ is given by 
\[ 
L_{1,2}(h) = \sum_{i\neq j} \abs{\hat{h}(\{i,j\})} = \sum_{i\neq j} \abs{\BE_{\lZ \sim \unif}[h(\lZ) \lZ_i\lZ_j]} = \sum_{i\neq j} \abs{\BE_{\lX,\lY \sim \unif}[\Ccal(\lX,\lY)\lX_i\lX_j\lY_i\lY_j]}.
\]

To bound the above, it suffices to bound $\sum_{i\neq j} \eta_{ij} \cdot \BE [\Ccal(\lX,\lY)\lX_i\lX_j \lY_i\lY_j]$ for any symmetric sign matrix $(\eta_{ij})$. For this proof overview, we assume for simplicity that $\eta_{ij}\equiv1$.

\paragraph*{Martingales and Gram-Schmidt Orthogonalization.} 
Similar to the case of level one, the level-two Fourier growth also has a martingale formulation.
In particular, let $\supX{t}$ and $\supY{t}$ be Alice and Bob's sets at time $t$ as before and define $\comtwo(\supX{t}) = \BE\sbra{\lx \tensor \lx \mid \lx \in \supX{t}},\comtwo(\supY{t}) = \BE\sbra{\lY \tensor \lY \mid \lY \in \supY{t}}$ to be the $n\times n$ matrices that represent the \emph{level-two center of masses} of the two sets. 
Here $\lx \tensor \ly$ 
denotes the tensor product $\lx\otimes\ly$ with the diagonal zeroed out.%
\footnote{Here $x \tensor y$ is an $n \times n$ matrix. We will  also interchangeably view $n \times n$ matrices as $n^2$-length vectors.}
To bound the level-two Fourier growth, it suffices to bound the expected quadratic variation of the martingale $\pbra{\supZ{t}}_t$ defined by taking the inner product of the level-two center of masses $\supZ{t} := \ip{\comtwo(\supX{t})}{ \comtwo(\supY{t})}$ where $\ip{\cdot}{\cdot}$ is the inner product of two matrices viewed as vectors.

To this end, we again move to Gaussian space where the inputs $x, y \in \Rbb^n$ and transform the protocol to a clean protocol. First, we need an analog of the \emph{Gram-Schmidt orthogonalization} step --- this is achieved in a natural way by Alice sending inner product $\ip{x \tensor x}{\comtwo(\supY{t})}$ with Bob's level-two center of mass, and Bob does the same. 
Note that Alice and Bob are now exchanging values of quadratic polynomials in their inputs. Thus, to control the step sizes, we now need to control the second moment of quadratic forms which naturally motivates the following spectral analogue of $4$-wise independence.

\paragraph*{4-wise Cleanup with Quadratic Forms.} 
We say a random vector $\lX$ is $4$-wise clean with parameter $\lambda$ if the operator norm of the $n^2 \times n^2$ covariance matrix 
$$
\BE\sbra{\pbra{\lX\tensor \lX - \BE\sbra{\lX \tensor \lX}}\pbra{\lX\tensor \lX - \BE\sbra{\lX \tensor \lX}}^\top}
$$
is at most $\lambda$ where we view $\lX\tensor \lX - \BE[\lX \tensor \lX]$ as an $n^2$-dimensional vector.
This is equivalent to saying that for any quadratic form $\abra{M,\lX\tensor \lX}$,
\begin{equation}\label{eq:overview_level_two}
\E\sbra{\abra{M,\lX\tensor \lX - \E\sbra{\lX\tensor \lX} }^2} \le \lambda \frob{M}^2, 
\end{equation}
where $\frob{M}$ denotes the Euclidean norm of $M$ when viewed as a vector.
Thus, this allows us to control the second moment of any quadratic polynomial (and in particular, fourth moments of linear functions). 
We note that one can generalize the above spectral notion to $k$-wise independence in the natural way by looking at the covariance matrix of the tensor $\lX^{\tensor k}$.

We say a set is \emph{$4$-wise clean} with parameter $\lambda$ if \Cref{eq:overview_level_two} is preserved for all $M$ with zero diagonal\footnote{The requirement of zero diagonal is for analysis purposes only and can be assumed without loss of generality since $\lx\tensor\lx$ is zero diagonal anyway.}.
Combined with this notion, one can define the cleanup in an analogous way to the level-one cleanup: While there exists some $M \in \R^{n \times n}$ violating \Cref{eq:overview_level_two}, Alice sends the quadratic form $\ip{x \tensor x}{M}$ to Bob until her set is 4-wise clean with parameter $\lambda$. 

\paragraph*{Cleanup Analysis via Hanson-Wright Inequalities.} 
The crux of the proof is to bound the number of cleanup steps which, together with a similar analysis as in the level-one case, gives us the desired bound. 
We show that $m\lesssim d$ cleanup steps suffice in expectation to make the sets $4$-wise clean for $\lambda \le d \cdot \polylog(n)$. Analogous to \Cref{eqn:martingale} and \Cref{eqn:qv}, this gives a bound of $d^3 \cdot \polylog(n)$ on the expected quadratic variation and implies $L_{1,2}(h) \le d^{3/2}\cdot \polylog(n)$.

Since the parties send values of quadratic forms now, the analysis here is significantly more involved compared to the level-one case, even after moving to the Gaussian setting, where one could previously use the fact that the Gaussian distribution behaves nicely under linear projections. 
We rely on a powerful generalization of the Hanson-Wright inequality to a Banach-space-valued setting due to Adamczak, Latała, and Meller~\cite{A20}. 
This inequality gives a tail bound for sum of squares of quadratic forms: 
In particular if $M_1, \ldots, M_m$ are matrices with zero diagonal which form an orthonormal set when viewed as $n^2$ dimensional vectors,
then the random variable $\lQ = \sum_{i=1}^m \ip{\lX \tensor \lX}{M_i}^2$ satisfies $\Pr_{\lx \sim \gamma_n}[\lQ \ge t] \le e^{-\Omega(\sqrt{t})}$ for any  $t\gtrsim m^2$ (see \Cref{thm:quadratic_concentration} for a precise statement).
We remark that this tail bound relies on the orthogonality of the quadratic forms and is much sharper than, for example, the bound obtained from hypercontractivity or other standard polynomial concentration inequalities.

In our setting, the matrices are being chosen adaptively.
In addition, the parties are sending quadratic forms in their inputs, and the distribution of the next $\ip{\lX \tensor \lX}{M}$ conditioned on the history is hard to determine, unlike the level-one case. 
To handle this, we replace the real communication with Boolean communication of finite precision $\pm 1/\poly(n)$. 
This means that whenever Alice wants to perform cleanup $\abra{\lx\otimes\lx,M}$ for some $M$ known to both parties, she sends only $O(\log(n))$ bits. 
On the one hand, this modification is similar enough to the cleanup protocol with real messages so that most of the argument carries through. On the other hand, now the protocol is completely discrete, which allows us to condition on any particular transcript.

For intuition, assume we fix a transcript of $L=d + O(m\log(n))$ bits which has gone through $m$ cleanups.
Typically, this transcript should capture $\approx 2^{-L}$ of the probability mass. 
More crucially, the matrices $M_1, \ldots, M_m$ for the cleanups are also fixed along the transcript, and one can apply the aforementioned Hanson-Wright inequality on $\lQ = \sum_{i=1}^m \ip{\lX \tensor \lX}{M_i}^2$. 
Combining the two facts together, we can apply the non-adaptive tail bound above and then condition on obtaining such typical transcript. 
This shows $\E[\lQ]\le d^2\cdot\polylog(n)$.
However, each quadratic form comes from a violation of \Cref{eq:overview_level_two} and contributes at least $\lambda$ to $\lQ$ in expectation.
This implies that $\E[\lQ]\ge \lambda\cdot m$ and by taking $\lambda=d\cdot\polylog(n)$, we derive that the number of cleanup steps $m\lesssim d$. This shows that the level-two Fourier growth is $O((m+d)\cdot\sqrt\lambda)=d^{3/2}\cdot\polylog(n)$ completing the proof. 

Note that if we could take $\lambda = \polylog(n)$ while having the same number of cleanup steps $m=d\cdot\polylog(n)$, then we would obtain an optimal level-two bound of $d \cdot \polylog(n)$.
However, it is not clear how to use current approach to show this.
In \Cref{sec:improved-hw}, we identify examples showing the tightness of our current analysis and also discuss potential ways to circumvent the obstacles within.

We remark that by replacing the Hanson-Wright inequality with its higher-degree variants and performing level-$k$ cleanups, we can analyze level-$k$ Fourier growth in the similar way.
However, since the first two levels already suffice for our applications and we believe that our level-two bound can be further improved, we do not make the effort of generalizing it to higher levels here.
\section{Preliminaries}\label{sec:prelim}

\paragraph*{Notation.} 
Throughout, $\log(\cdot)$ and $\ln(\cdot)$ denote logarithms with base $2$ and $e$ respectively. We use $\N=\cbra{0,1,2,\ldots}$ to denote the set of natural numbers including 0.
For $n \in \mathbb{N}$, we write $[n]$ to denote the set $\cbra{1,2,\ldots,n}$. We use the standard $O(\cdot), \Omega(\cdot), \Theta(\cdot)$ notation, and emphasize that in this paper they only hide universal constants that do not depend on any parameter.

We write $\odot$ to denote the entrywise product for vectors and matrices: in particular, for any $x,y\in \Rbb^n$, we define $x\odot y\in \Rbb^n$ to be a vector where $(x\odot y)_i=x_iy_i$ for $i\in[n]$ and similarly for any $X,Y\in\Rbb^{n\times m}$, we define $X\odot Y\in \Rbb^{n\times m}$ to be a matrix where $(X\odot Y)_{ij}=X_{ij}Y_{ij}$ for $i \in [n], j\in [m]$. We use $\tensor$ to denote a tensor with zeros on the diagonal, i.e., for any $x\in\Rbb^n$, $x\tensor x$ is a $n\times n$ matrix where $\pbra{x\tensor x}_{ij}= x_ix_j$ if $i\neq j$ and zero if $i = j$.

For a vector $x\in\Rbb^n$, we use $\vabs{x}$ to denote its Euclidean norm. Similarly, for a matrix $X\in\Rbb^{n\times n}$, we use $\frob{X}$ to denote its Euclidean norm viewing the matrix $X$ as an $n^2$-dimensional vector. For nonzero $x \in \R^n$ or $X \in \R^{n\times n}$, we define $\unit(x) \in \R^n$ or $\unit(X) \in \R^{n\times n}$ as the unit vector along direction $x$ and $X$ respectively: $\unit(x) = x/\vabs{x}$ and $\unit(X) = X/\frob{X}$. 
We write $\mathbb{S}^{n-1}$ for the unit sphere in $\R^n$, and write $\Sbb^{n\times n-1}$ for the unit sphere in $\Rbb^{n\times n}$ where additionally the diagonal entries of the $n \times n$ matrices are zero. 
We use $\ip{x}{y}$ to denote the inner product between vectors $x,y \in \R^n$ and $\ip{X}{Y}$ to denote the inner product between matrices $X, Y \in \R^{n \times n}$ viewing them as $n^2$-dimensional vectors.

\paragraph*{Probability.}
A probability space is a triple $(\Omega, \F, \xi)$ where $\Omega$ is the sample space, $\F$ is a $\sigma$-algebra which describes the measurable sets (or events) in the probability space, and $\xi$ is a probability measure. We use $\lX\sim \xi$ to denote a random sample distributed according to $\xi$ and $\E_{\lX\sim \xi}[f(\lX)]$ to denote the expectation of a function $f$ under the measure $\xi$. For any event $S\in \F$, we use $\xi(S)$ to denote the measure of $S$ under $\xi$. We say an event $S$ holds \emph{almost surely} if $\xi(S)=1$, i.e., the exceptions to the event have measure zero. For a measurable event $\Ecal \in \F$, we write $\F \cap \{\Ecal\}$ to denote the intersection of the sigma-algebra $\F$ and the sigma-algebra generated by $\Ecal$.

We use $\Ucal_n$ to denote the uniform probability measure over $\binpm^n$ and $\gamma_n$ to denote the $n$-dimensional standard Gaussian measure in $\R^n$. We say a random variable $\lX \in \R^n$ is a  standard Gaussian in $\R^n$ if its probability distribution is $\gamma_n$. We will drop the subscript if the dimension is clear from context. We will also need lower dimensional Gaussian measures: given a linear subspace $V$ of dimension $k$, there is a $k$-dimensional standard Gaussian measure on it, which we denote by $\gamma_V$. For any measurable subset $S \subseteq \R^n$, we define its ambient space to be the smallest affine subspace $V+t$ that contains it where $V$ is a linear subspace of $\R^n$ and $t \in \R^n$. The relative Gaussian measure of $S$ denoted by $\gamma_{\rel}(S)$ is then defined to be the Gaussian measure of the set $S-t$ under $\gamma_V$.

\paragraph*{Martingales.} 
Given a sequence of real-valued random variables $\lX_1, \lX_2, \ldots, \lX_n$ in a probability space $(\Omega, \F, \xi)$ and a function $f(\lX_1,\ldots, \lX_n)$ satisfying $\BE\sbra{|f(\lX_1,\ldots,\lX_n)|} < \infty$, the sequence of random variables $\supZ{t} = \BE\sbra{f(\lX_1, \ldots, \lX_n) \mid \supF{t-1}}$ is called the \emph{Doob martingale} where $\supF{t-1}$ is the $\sigma$-algebra generated by $\lX_1,\ldots, \lX_{t-1}$ which should be viewed as a record of the randomness of the process until time $t-1$. The sequence $(\supF{t})_t$ is called a \emph{filtration}. A sequence of random variables $(\supZ{t})_t$ is called \emph{predictable} (or \emph{adapted}) with respect to $\supF{t}$ if $\supZ{t}$ is $\supF{t}$-measurable for every $t$, meaning that it is determined by the randomness in $\supF{t}$.

A discrete random variable $\btau \in \N$ is called a \emph{stopping time} with respect to the filtration $(\supF{t})_t$ if the event $\{\btau = t\} \in \supF{t}$ for all $t \in \Nbb$, or in words, whether the event $\btau=t$ occurs is determined by the history of the process until time $t$. All stopping times considered in this paper will be finite. The sigma-algebra $\supF{\btau}$ which contains all events that imply the stopping condition is defined as the set of all events $\Ecal$ such that $\Ecal \cap \{\btau = t\} \in \supF{t}$ for all $t \in \N$. We also note if one takes an increasing sequence of stopping times $(\btau_m)_m$ then the process defined by $(\supZ{\btau_m})_m$ is also a martingale.

Let $\Delta \supZ{t} := \supZ{t} - \supZ{t-1}$ be the martingale differences.
Note that $\BE\sbra{\Delta \supZ{t} \mid \supF{t-1}} = 0$ and thus
\begin{equation}\label{eqn:martingale-orthogonality}
    \BE\left[\left(\supZ{t}\right)^2\right] = \BE\left[\left(\sum_{t=1}^n \Delta \supZ{t}\right)^2\right] = \BE\left[\sum_{t=1}^n \left(\Delta \supZ{t}\right)^2\right], 
\end{equation}
where the cross terms disappear upon taking expectation. 
In other words, the martingale differences are orthogonal under taking expectations. The right hand side above is the \emph{expected quadratic variation} of the martingale $\pbra{\supZ{t}}_t$. If the sequence $(\supZ{t})_t$ is vector-valued (resp., matrix-valued) and satisfies $\BE\sbra{\Delta \supZ{t} \mid \supF{t-1}} = 0$ where $0$ is zero vector (resp., matrix), then we say it is a vector-valued (resp., matrix-valued) martingale with respect to $(\supF{t})_t$. Since each coordinate of a vector or matrix-valued martingale is itself a real-valued martingale, vector-valued or matrix-valued martingale differences are also orthogonal under Euclidean norms:
\begin{equation}\label{eqn:martingale-orthogonality-vec}
    \BE\left[\frob{\supZ{t}}^2\right] = \BE\left[\frob{\sum_{t=1}^n \Delta \supZ{t}}^2\right] = \BE\left[\sum_{t=1}^n \frob{\Delta \supZ{t}}^2\right]. 
\end{equation}

\paragraph*{Useful Inequalities.} 
We will use the well-known level-$k$ inequality \cite{DBLP:journals/combinatorica/Talagrand96, DBLP:conf/focs/KahnKL88} (see e.g., \cite[Level-$k$ Inequalities]{DBLP:books/daglib/0033652}). A statement in the Gaussian setting can be found in, e.g., \cite[Lemma 2.2]{eldan2022reduction}. We remark that we will only use the case for $k=1$ and $k=2$ here which we state below.\footnote{Our \Cref{thm:level_k_ineq} is slightly different from the references, where they additionally require $\mu\le1/e$. By Parseval's identity, the left hand side is always at most one. Therefore we use a slightly worse bound for the right hand side to allow for the whole range of $\mu$.}

Below we write $\ind_A$ for the indicator function of a set and $x_S = \prod_{i\in S} x_i$ for a monomial.

\begin{theorem}[Level-$k$ Inequality]\label{thm:level_k_ineq}
Let $k \in \{1,2\}$. Assume $A \subseteq \Rbb^n$ is measurable and $\mu:=\E_{\lX\sim\gamma}[\ind_A(\lX)]$. Then, we have
$$
\sum_{|S|=k}\pbra{\E_{\lX\sim\gamma}\sbra{\ind_A(\lX)\lX_S}}^2\le 2e^2\mu^2 \cdot \ln^k(e/\mu).
$$
\end{theorem}

In particular, if $\mu$ is non-zero, dividing both sides by $\mu^2$, we get the following more convenient form for $k \in \{1,2\}$: 
\[ \sum_{|S|=k}\pbra{\E_{\lX\sim\gamma}\sbra{\lX_S \mid \lx \in A}}^2 \le 2e^2 \cdot \ln^k(e/\mu). \]
We also make use of the following standard concentration inequality for sums of squares of independent standard Gaussians (see \cite{vershynin2018high}).

\begin{fact}\label{thm:chi_squared_concentration}
Let $m\in\N$ be arbitrary.
For any $r\ge 2m$, we have $\Pr_{\lX \sim \gamma_m} \sbra{\sum_{i=1}^m \lX_i^2 \ge r}\le e^{-{r}/{4}}$.
\end{fact}

We also need a concentration inequality for sums of squares of orthogonal quadratic forms over Gaussian random variables. 
In particular, we prove the following inequality which follows from a generalization of the Hanson-Wright inequality to a Banach space-valued setting~\cite[Theorem 6]{A20}. 
Since, we only need a special case that is easier to prove, we include a self-contained proof using the Gaussian isoperimetric inequality in \Cref{app:thm:quadratic_concentration} following~\cite[Proposition 23]{A20}.

\begin{theorem}\label{thm:quadratic_concentration}
Let $m\in\N$ be arbitrary.
Let $M_1,\ldots,M_m$ be $n \times n$ real matrices where each $M_i$ has zero diagonal, $\abra{M_i,M_i}=1$ and $\abra{M_i,M_j}=0$ for $i\neq j$. Then for any $r\ge98m$, we have
$$
\Pr_{\lX\sim\gamma_n}\sbra{\sum_{i=1}^m\abra{\lX\tensor \lX,M_i}^2\ge r}\le\exp\cbra{-\Omega\pbra{\frac r{m+\sqrt r}}}.
$$
\end{theorem} 

We remark that the tail bound above holds more generally for sub-Gaussian random variables $\lx$ (see~\cite{A20}). 
\section{Fourier Growth via Martingales in Gaussian Space}\label{sec:fourier_via_martingale}

In this section, we reduce the question of bounding the level-one and level-two Fourier growth to bounding the expected quadratic variation of certain martingales. 
To analyze these martingales and to prove the optimal bound for the level-one setting, it seems to be crucial to work in the Gaussian setting, so first we give a generic transformation from Boolean to Gaussian. We shall also additionally allow protocols that communicate real numbers to make the analysis easier.

\subsection{Communication Protocols in Gaussian Space}
\label{sec:boolean_to_real}

Let $\Ccal: \pmones \times \pmones \to \pmone$ be a communication protocol with total communication $d$ and $h$ be its XOR-fiber defined in \Cref{eqn:fiber}.

We embed the protocol in the Gaussian space by allowing Alice's and Bob's inputs, $x$ and $y$ respectively, to be real vectors in $\R^n$ --- the new protocol $\tilde{\Ccal}$ runs the original protocol $\Ccal$ with Boolean inputs $\sgn(x)$ and $\sgn(y)$ where $\sgn(v) = (\sgn(v_1), \ldots, \sgn(v_n))$ denotes the sign function applied pointwise to each coordinate for a vector $v \in \R^n$. The behavior of the communication protocol $\tilde{\Ccal}$ can be defined arbitrarily if any coordinate of $\sgn(x)$ or $\sgn(y)$ is zero since such points have zero measure under the standard $n$-dimensional Gaussian measure $\gamma_n$. 

This translation from the Boolean hypercube to the Gaussian space preserves the measure of sets: for any subset $S\subseteq \binpm^n$, we have $\Ucal_n(S)=\gamma_n(\cbra{x\in\Rbb^n\mid\sgn(x)\in S})$ where $\unif_n$ is the uniform measure over $\pmones$. 
Moreover, up to some normalizing factor, the Fourier coefficients of $h$ can also be computed by looking at Gaussian inputs. 
In particular, denoting by $x_S = \prod_{i \in S} x_i$ for a subset $S \subseteq [n]$, we have the following fact.
\begin{fact}\label{fct:boolean_to_real}
For all $S\subseteq [n]$, we have $		\E_{\lZ\sim\Ucal_n}\sbra{h(\lZ)\lZ_S}=(\pi/2)^{|S|}\E_{\lX,\lY\sim\gamma_n}\sbra{\tilde {\Ccal}(\lX,\lY) \lX_S \lY_S}$.
\end{fact}
\begin{proof}
Note that for $\lX \sim \gamma_n$, the random variable $\sgn(\lX)$ is distributed as $\unif_n$. Thus, by the definition of the XOR-fiber $h$ and the protocol $\tilde{\Ccal}$, we have
\begin{align*}
\E_{\lZ\sim\Ucal_n}\sbra{h(\lZ)\lZ_S}
&=\E_{\lX,\lY\sim\gamma_n} \left[\Ccal(\sgn(\lX), \sgn(\lY)) \cdot \prod_{i \in S} \sgn(\lX_i) \cdot \sgn(\lY_i)\right]\\
&= (\pi/2)^{|S|}\E_{\lX,\lY\sim\gamma_n} \left[\Ccal(\sgn(\lX), \sgn(\lY)) \cdot \prod_{i \in S} \lX_i \cdot \lY_i\right]\\
&=(\pi/2)^{|S|}\E_{\lX,\lY\sim\gamma_n}\sbra{\tilde {\Ccal}(\lX,\lY) \lX_S \lY_S},
\end{align*}
where the second line follows since the expected value of a standard Gaussian in $\R$ conditioned on its sign being fixed to $\eta$ is $\sqrt{\frac2\pi}\cdot \eta$ by the following calculation:
\begin{equation*}
\E_{\lX_i\sim\gamma}\sbra{\lX_i \mid\sgn(\lX_i)=\eta}= \eta \cdot\int_0^{\infty}\sqrt{\frac2\pi}\cdot r\cdot e^{-r^2/2} \sd r=\sqrt{\frac2\pi}\cdot \eta. \qedhere
\end{equation*}
\end{proof}

\begin{remark}\label{rem:symmetric}
    We remark that instead of the Gaussian distribution above, one can work with any distribution where the coordinates are i.i.d. and symmetric around zero.
    In particular, if $\xi$ is a symmetric probability measure on the real line, and $\lx,\ly$ are independently drawn vectors in $\Rbb^n$ where each coordinate is i.i.d. sampled from $\xi$, then 
    $\E_{\lZ\sim\Ucal_n}\sbra{h(\lZ)\lZ_S}=c_{\xi}^{|S|}\E_{\lX,\lY\sim\xi^{\otimes n}}\sbra{\tilde {\Ccal}(\lX,\lY) \lX_S \lY_S}$ where $c_{\xi} = (\BE_{\lx_i \sim \xi}[|\lx_i|])^{-2}$. In the case of level-two we will need to work with the truncated Gaussian distribution where each coordinate is sampled independently from the one dimensional standard Gaussian conditioned on being in some interval $[-T,T]$ for $T = \Omega(1)$ in which case $c_{\xi}$ is upper bounded by a universal constant.
\end{remark}

\subsection{Generalized Communication Protocols}\label{sec:gen-protocol}

In the protocol $\tilde{\Ccal}$ defined above, Alice and Bob's inputs $x$ and $y$ are real vectors in $\R^n$, but in each round they still exchange a single bit based on $\sgn(x)$ and $\sgn(y)$. In order to bound the Fourier growth, it will be more convenient for us to define a notion of generalized communication protocols where parties are also allowed to send real numbers with arbitrary precision in each round. To define this formally, we place certain restrictions on the real communication in the protocol. More formally, in a generalized communication protocol, in each round a player with input $z \in \R^n$ can either send:
\begin{enumerate}[label={(\roman*)}]
	\item a bit in $\bin$ which is purely a function of the Boolean input $\sgn(z)$ and the previous \emph{Boolean} messages, or
	\item a real number that is a measurable function of $z$ and the previous (real or Boolean) messages.
\end{enumerate} 

The \emph{depth} of a generalized communication protocol is defined to be the maximum number of rounds of communication. 

Note that a generalized protocol also generates a ``protocol tree'' where if in a round a real number is sent, the ``children'' of that particular ``node'' are indexed by all possible values in $\R$. A ``transcript'' of the protocol can be defined in an analogous way. The set of inputs that reach a particular node of this generalized protocol tree still form a rectangle $X \times Y$ where $X, Y \subseteq \R^n$. We say that a generalized protocol $\bar{\Ccal}$ is equivalent to the protocol $\tilde{\Ccal}$ if ${\bar{\Ccal}}(x,y) = \tilde{\Ccal}(x,y)$ for every $x, y \in \R^n$ except on a measure zero set. 

We will be interested in random walks on such generalized protocol trees when the inputs $\lx$ and $\ly$ are sampled from a product measure $\xi_x \times \xi_y$ on $\Rbb^n \times \Rbb^n$ and the parties send messages according to the protocol to reach a ``leaf''. The random variables corresponding to the messages until any time $t$ generate a filtration $(\supF{t})_t$ ---  this filtration can be thought of as specifying a particular node of the generalized protocol at depth $t$ (equivalently, a partial transcript of the protocol till time $t$) that was sampled by the process. In this case, conditioned on any event in $\supF{t}$, (e.g., any realization of the transcript till time $t$), almost surely the conditional probability measure on the inputs $\lx, \ly$ is some product measure on $\xi_x^{(t)} \times \xi_y^{(t)}$ supported on a rectangle $\supX{t} \times \supY{t}$ where $\supX{t}, \supY{t} \subseteq \Rbb^n$. We shall refer to the random variable $\supX{t} \times \supY{t}$ as the current rectangle determined by $\supF{t}$. Since we will be working with product measures on inputs $\lx, \ly$, the reader can think of conditioning on the filtration $\supF{t}$ as essentially conditioning on the inputs being in the rectangle $\supX{t} \times \supY{t}$ or equivalently a partial transcript till time $t$.

\subsection{Fourier Growth via Martingales}\label{sec:fourier_weights_via_martingales}

We will now relate Fourier growth to the quadratic variation of a martingale. Towards this end, we first note that in light of \Cref{fct:boolean_to_real}, the level-$k$ Fourier growth of the XOR-fiber $h$ of the original communication protocol is given by 
\begin{align}\label{eqn:boolean-to-real}
L_{1,k}(h) = \sum_{\substack{S \subseteq[n]\\|S|=k}} \abs{\BE_{\lZ \sim \unif_n}[h(\lZ) \lZ_S]} 
&= (\pi/2)^k \sum_{\substack{S \subseteq[n]\\|S|=k}} \abs{\BE_{\lX,\lY \sim \gamma_n}[\bar{\Ccal}(\lX,\lY)\lX_S\lY_S]} \notag \\
& = (\pi/2)^k \max_{(\eta_S)_{|S|=k}} \sum_{\substack{S \subseteq[n]\\|S|=k}} \eta_S {\BE_{\lX,\lY \sim \gamma_n}\sbra{\bar{\Ccal}(\lX,\lY)\lX_S\lY_S}},
\end{align} 
where $\bar{\Ccal}$ is any generalized protocol that is equivalent to $\tilde{\Ccal}$ and $\eta_S \in \pmone$.

We now express the right hand side above as an inner product. Let $\bell$ be a random leaf of the generalized protocol tree $\bar \Ccal$ induced by taking $\lX, \lY \sim \gamma_n$ and let $\X_\ell \times \Y_\ell$ be the corresponding rectangle in the generalized protocol tree. Then, 
\begin{align}\label{eqn:ip}
\sum_{\substack{S \subseteq[n],|S|=k}} \eta_S {\BE_{\lX,\lY \sim \gamma_n}\sbra{\bar{\Ccal}(\lX,\lY)\lX_S\lY_S}} 
&=\E_{\bell}\sbra{\E_{\lX,\lY\sim\gamma}\sbra{\sum_{\substack{S \subseteq[n],|S|=k}} \eta_S \cdot\bar\Ccal(\lX,\lY)\lX_S \lY_S\mid (\lX,\lY)\in \X_{\bell} \times \Y_{\bell}}} \notag\\
&=\E_{\bell}\sbra{\bar\Ccal(\bell)\E_{\lX,\lY\sim\gamma}\sbra{\sum_{\substack{S \subseteq[n],|S|=k}} \eta_S \cdot\lX_S \lY_S\mid (\lX,\lY)\in \X_{\bell} \times \Y_{\bell}}} \notag \\
&\le \E_{\bell}\sbra{~\left|\sum_{\substack{S \subseteq[n],|S|=k}} \eta_S \E\sbra{\lX_S \mid \lX \in \X_{\bell}} \cdot \E\sbra{\lY_S \mid \lY \in \Y_{\bell}}\right|~},
\end{align}
where the second line follows since $\bell$ is a leaf and determines the answer and the third line follows since $\lX$ and $\lY$ are independent conditioned on being in the rectangle $\X_{\bell} \times \Y_{\bell}$.

Thus, specializing \Cref{eqn:ip} to the level-one ($k=1$) and level-two cases ($k=2$), from \Cref{eqn:boolean-to-real} we get that 
\begin{align*}
L_{1,1}(h) &\le \frac{\pi}{2} \cdot \max_{\eta}~ \E_{\bell}\sbra{~\left|\sum_{i=1}^n \eta_i \cdot \E\sbra{\lX_i \mid \lX \in \X_{\bell}} \cdot \E\sbra{\lY_i \mid \lY \in \Y_{\bell}}\right|~},\\
L_{1,2}(h) &\le \frac{\pi^2}{4} \cdot \max_\eta~\E_{\bell}\sbra{~\left|\sum_{i,j=1}^n \eta_{ij} \cdot ~\E\sbra{\lX_{ij} \mid \lX \in \X_{\bell}} \cdot \E\sbra{\lY_{ij} \mid \lY \in \Y_{\bell}}\right|~},
\end{align*}
where for $L_{1,1}$ we optimize over $\eta \in \pmones$ and for $L_{1,2}$ we optimize over $\eta$ being an $n \times n$ symmetric matrix with zeros on the diagonals and $\pm 1$ entries otherwise.

To make the above more compact, we respectively define $\com(X) \in \R^n$ and $\comtwo(X) \in \R^{n\times n}$ to be the level-one and level-two centers of mass of a set $X \subseteq \R^n$:
\begin{equation}\label{def:com}
    \com(X) = \BE_{\lX \sim \gamma_n}\sbra{\lX \mid \lX \in X} 
    \quad\text{and}\quad
    \comtwo(X) = \BE_{\lX \sim \gamma_n}\sbra{\lX \tensor \lX \mid \lX \in X}.
\end{equation}
Then, upper bounding the constants in the above inequality ($\pi/2$ and $\pi^2/4$) by $4$, we get
\begin{equation}\label{eqn:fwt-to-ip}
\begin{split}
    L_{1,1}(h) &\le 4 \cdot \max_{\eta}~ \E_{\bell}\sbra{\left| \ip{\com(\X_{\bell})}{ \eta \odot \com(\Y_{\bell})} \right|},\\
    L_{1,2}(h) &\le 4 \cdot \max_\eta~\E_{\bell}\sbra{\left| \ip{\comtwo(\X_{\bell})}{ \eta \odot \comtwo(\Y_{\bell})}  \right|},
\end{split}
\end{equation}
where $\eta$ is understood to be the same as before. 

Moving forward, we fix an arbitrary $\eta$ for both cases $k \in \{1,2\}$ and define a martingale process $\left(\supZ{t}_k\right)_t$ that captures the right hand side above. For this we note that a generalized communication protocol, where Alice's and Bob's inputs are sampled from the Gaussian distribution, naturally induces a discrete-time random walk on the corresponding (generalized) protocol tree where at time $t$ we are at a node at depth $t$ with the corresponding rectangle $\supX{t}\times \supY{t}$. Then, we have the following proposition.

\begin{proposition}\label{prop:vec-martingale}
    $\com(\supX{t})$ and $\com(\supY{t})$ are vector-valued martingales taking values in $\R^n$ and $\comtwo(\supX{t})$ and  $\comtwo(\supY{t})$ are matrix-valued martingales taking values in $\R^{n \times n}$.
\end{proposition}

Note that if in the $t^{\text{th}}$ round Alice speaks, then $\com(\supY{t})$ and $\comtwo(\supY{t})$ do not change and similarly if Bob speaks, then $\com(\supX{t})$ and $\comtwo(\supX{t})$ do not change.
The above proposition implies that the real-valued processes 
\begin{equation}
    \label{eqn:def-martingale}
 \supZ{t}_1 = \ip{\com(\supX{t})}{\eta \odot \com(\supY{t})} \text{ and } \supZ{t}_2 = \ip{\comtwo(\supX{t})}{\eta \odot \comtwo(\supY{t})},
\end{equation}
each form a Doob martingale with respect to the natural filtration induced by the random walk on the protocol tree.
Note that taking a random walk on the tree until we hit a leaf generates the marginal distribution on $\bell$ given in \Cref{eqn:fwt-to-ip}. Let $\D$ be the stopping time when this martingale hits a leaf and stops (i.e., the depth of the random leaf). Thus, by the orthogonality of martingale differences $\Delta\supZ{t}_k = \supZ{t}_k - \supZ{t-1}_k$ from \Cref{eqn:martingale-orthogonality},  we get that for $k \in \{1,2\}$, one can upper bound the Fourier growth in terms of expected quadratic variation of the above martingales:
\begin{proposition}\label{prop:fwt-to-qv}
For $k\in\{1,2\}$, $\frac14\cdot  L_{1,k}(h) \le \max_{\eta}\sqrt{\BE\left[\left(\supZ{\D}_k\right)^2\right]} = \max_{\eta}\sqrt{\BE\left[\sum_{t=1}^{\D} \left(\Delta \supZ{t}_k\right)^2\right]}$.
\end{proposition}
The martingale implicitly depends on $\eta$ as used in \Cref{eqn:fwt-to-ip} and hence the maximum. Moreover, the martingale also depends on the underlying generalized communication protocol $\bar \Ccal$. In the next two sections, we will show that after transforming the original communication protocol into ``clean'' protocols, the expected quadratic variations of $(\supZ{t}_1)_t$ and $(\supZ{t}_2)_t$ are $O(d)$ and $O(d^3) \cdot \polylog(n)$ respectively. This will then imply our main theorems.

\begin{remark}\label{rem:martingale}
Note that \Cref{prop:vec-martingale} still holds even if the input distribution is not the Gaussian distribution, but some other product probability measure on the inputs $\lx, \ly$. This also implies that $\supZ{t}_k$ for $k \in \{1,2\}$ is a martingale. In particular, for the level-two case, we will need to use a truncated Gaussian distribution. In light of \Cref{rem:symmetric}, \Cref{prop:fwt-to-qv} still suffices for us with a different constant instead of $1/4$. We also remark that we shall also need to truncate the real messages being used in the protocol for the level-two case to a finite precision, so the generalized protocols for the level-two case only have Boolean communication. However, to obtain the optimal level-one bound allowing generalized protocols that communicate real values seems to be crucial.
\end{remark}
\section{Level-One Fourier Growth}\label{sec:proof_of_level_one}

In this section, we will give a proof of
\Cref{thm:boolean_bound_level_one} that $L_{1,1}(h) = O(\sqrt{d})$. We start with a $d$-round communication protocol $\tilde{\Ccal}$ over the Gaussian space as defined in \Cref{sec:boolean_to_real}. 
Given the discussion in the previous section and \Cref{prop:fwt-to-qv}, our task ultimately reduces to bounding the expected quadratic variation of the martingale that results from the protocol $\bar{\Ccal}$. For example, one can simply take $\bar\Ccal=\tilde\Ccal$, but, as discussed in \Cref{sec:overview}, the individual step sizes of this martingale can be quite large in the worst-case and it is not so easy to leverage cancellations here to bound the quadratic variation by $O(d)$. 

So, we first define a \emph{generalized} communication protocol $\bar{\Ccal}$ that is equivalent to the original protocol $\tilde\Ccal$ but has additional ``cleanup'' rounds where Alice and Bob reveal certain linear forms of their inputs so that their sets are pairwise clean in the sense described in the overview. These cleanup steps allow us to keep track of the quadratic variation more easily.

\subsection{Pairwise Clean Protocols}\label{sec:pairwise_clean_protocols}

To define a clean protocol, we first define the notion of a pairwise clean set.
Let $X\subseteq \Rbb^n$. We say that the set $X$ is \emph{pairwise clean in a direction $a \in \mathbb{S}^{n-1}$} with parameter $\lambda$ if 
\begin{equation}\label{eqn:pairwiseclean}
 \E_{\lX\sim \gamma}\sbra{ \abra{\lX-\com(X),a}^2 \mid \lX\in X }\le \lambda,
\end{equation}
where we recall that $\com(X) = \E_{\lX\sim \gamma}\sbra{\lX\mid \lX \in X}$ is the level-one center of mass of $X$. 

The above condition implies that for a random vector $\lX$ sampled from $\gamma$ conditioned on $X$, its variance along the direction $a$ is bounded by $\lambda$. We say that the set $X$ is \emph{pairwise clean} (with parameter $\lambda$) if it is clean in \emph{every direction $a \in \mathbb{S}^{n-1}$}. Equivalently, the operator norm of the covariance matrix of the random vector $\lX$ is bounded by $\lambda$.

We call a generalized communication protocol pairwise clean with parameter $\lambda$ if at the start of a new ``phase'' of the protocol, the corresponding rectangle $X \times Y$ satisfies that both $X$ and $Y$ are pairwise clean. Starting from a communication protocol $\tilde{\Ccal}$ in the Gaussian space, we will transform it into a pairwise clean protocol $\bar \Ccal$ by proceeding from top to bottom and adding certain Gram-Schmidt orthogonalization and cleanup steps.

In particular, consider an intermediate node in the protocol tree of $\tilde{\Ccal}$. Before Alice sends her bit as in the original protocol $\tilde\Ccal$, she first performs an orthogonalization step by revealing the inner-product between her input and Bob's current level-one center of mass. After this, she sends her bit according to the original protocol and afterwards she repeatedly cleans her current set $X$ by revealing $\abra{x,a}\in \Rbb$ while $X$ is not clean along the direction $a$ orthogonal to previous directions. 
Once $X$ becomes clean, they proceed to the next round. 
We now describe this formally.

\paragraph*{Construction of pairwise clean protocol $\bar \Ccal$ from $\tilde \Ccal$.}

We set $\lambda = 100$. The construction of the new protocol is recursive and we first define some notation. Consider an intermediate node of the new protocol $\bar \Ccal$ at depth $t$. We use the random variable $\supX{t}\subseteq\Rbb^n$ (resp., $\supY{t}\subseteq \Rbb^n$) to denote the set of inputs of Alice (resp., Bob) reaching the node. 
If Alice reveals a linear form in this step, we use $\supa{t}\in \Rbb^n$ to denote the vector of the linear form; otherwise, we set $\supa{t}$ to be the all-zeroes vector. 
We define $\supb{t}$ similarly for Bob. Throughout the protocol, we will abbreviate $\supu{t} = \com(\supX{t})$ and $\supv{t} = \com(\supY{t})$ for Alice's and Bob's current center of mass respectively. 
\begin{enumerate}
	\item At the beginning, Alice receives an input $x\in\Rbb^n$ and Bob receives an input $y\in\Rbb^n$.
	\item We initialize $t\gets0$, $\supX{0},\supY{0}\gets\Rbb^n$, and $\supa{0},\supb{0}\gets0^{n}$. 
	\item For each phase $i=1,2,\ldots,d$: suppose we are starting the cleanup for a node at depth $i$ in the original protocol $\tilde\Ccal$ and suppose we are at a node of depth $t$ in the new protocol $\bar\Ccal$. If it is Alice's turn to speak in $\tilde{\Ccal}$:
	\begin{enumerate}
		\item \textbf{Orthogonalization by revealing the correlation with Bob's center of mass.}\\
		Alice begins by revealing the inner product of her input $x$ with Bob's current (signed) center of mass $\Lambda\odot \supv{t}$. Since in the previous steps, she has already revealed the inner product with Bob's previous centers of mass, for technical reasons, we will only have Alice announce the inner product with the component of $\Lambda\odot \supv{t}$ that is orthogonal to the previous directions along which Alice announced the inner product. More formally, let $\la^{(t+1)}$ be the component of $\Lambda\odot \supv{t}$ that is orthonormal to all previous directions $\supa{1},\dots, \supa{t}$, i.e.,
		$$\textstyle
		\la^{(t+1)}=\unit\pbra{ 
		\Lambda\odot \supv{t}
		- \sum_{\tau=1}^{t}	\abra{\Lambda \odot \supv{t},\supa{\tau}} \cdot \supa{\tau}}.$$
		
		Alice computes $\bar \lc^{(t+1)}\gets \abra{x,\la^{(t+1)}}$ and sends $\bar \lc^{(t+1)}$ to Bob. Set $\lb^{(t+1)}\gets 0^ n$. Increment $t$ by 1 and go to step (b). 
		\item \textbf{Original communication.} Alice sends the bit $\supcbar{t+1}$ that she was supposed to send in $\tilde\Ccal$ based on previous messages and the input $x$. Set $\la^{(t+1)},\lb^{(t+1)}\gets 0^n$. Increment $t$ by 1 and go to step (c). 
		\item \textbf{Cleanup steps.} While there exists some direction $a\in\mathbb{S}^{n-1}$ orthogonal to the previous directions (i.e., satisfying $\abra{a,\la^{(\tau)}}=0$ for all $\tau\in [t]$) such that $\X^{(t)}$ is \emph{not pairwise clean} in direction $a$, Alice computes $\bar \lc^{(t+1)}\gets\abra{x,a}$ and sends this to Bob. 
		Set $\la^{(t+1)}\gets a$ and $\lb^{(t+1)}\gets0^n$. Increment $t$ by 1. Repeat step (c) as long as $\X^{(t)}$ is not pairwise clean; otherwise increment $i$ by 1 and go back to the for-loop in step 3 which starts the new phase.
	\end{enumerate}
	If it is Bob's turn to speak, we define everything similarly with the role of $x,\la,\X,\V$ switched with $y,\lb,\Y,\U$.
	\item Finally at the end of the protocol, the value $\bar\Ccal(x,y)$ is determined based on all the previous communication and the corresponding output it defines in $\tilde\Ccal$.
\end{enumerate}

We note some basic properties that directly follow from the description. First we note that the steps 3(a), 3(b), and 3(c) always occur in sequence for each party and we refer to such a sequence of steps as a \emph{phase} for that party. Note that there are at most $d$ phases. If a new phase starts at time $t$, then the current rectangle $\supX{t} \times \supY{t}$ is pairwise clean for both parties by construction. Also, note that the non-zero vectors in the sequence $(\supa{t})_t$ (resp., $(\supb{t})_t$) form an orthonormal set. We also note that the Boolean communication in step 3(b) is solely determined by the original protocol and hence only depends on the previous Boolean messages.

Lastly, each phase has one 3(a) and 3(b) step, followed by potentially many 3(c) steps. However, the following claim shows that it is always finite.
\begin{claim}\label{clm:finite_steps_level_one}
Let $\ell$ be an arbitrary leaf of the protocol $\bar\Ccal$ and $D(\ell)$ be its depth. Then $D(\ell) \le 2n + 2d$. 
Moreover, along this path there are at most $2d$ many steps 3(a) and 3(b).
\end{claim}
\begin{proof}
	We count the number of communication steps separately:
	\begin{itemize}
		\item \textbf{Steps 3(a) and 3(b).} Steps 3(a) and 3(b)  occur once in every phase, thus at most $d$ times.
		\item \textbf{Step 3(c).} For Alice, each time she communicates at step 3(c) $a\in\Rbb^n$, the direction is orthogonal to all previous $\supa{t}$'s. Since the dimension of $\Rbb^n$ is $n$, this happens at most $n$ times. Similar argument works for Bob.
	\end{itemize}
	Thus in total we have at most $2n+2d$ steps.
\end{proof}
We will eventually show that the \emph{expected} depth of the protocol $\bar \Ccal$ is $O(d)$ when $\lX, \lY \sim \gamma_n$.

\subsection{Bounding the Expected Quadratic Variation}

Consider a random walk on the protocol tree generated by the new protocol $\bar \Ccal$ when the parties are given independent inputs $\lx, \ly \sim \gamma_n$.
Consider the corresponding level-one martingale process defined in \Cref{eqn:def-martingale}. Formally, at time $t$ the process is defined by
\[ \supZ{t}_1 = \ip{\supu{t}}{\eta \odot \supv{t}},\]
where we recall that $\supu{t} = \com(\supX{t})$ and $\supv{t} = \com(\supY{t})$ and $\eta \in \pmones$ is a fixed sign vector. 

The martingale process stops once it hits a leaf of the protocol $\bar \Ccal$. Let $\D$ denote the (stopping) time when this happens. Note that $\BE[\D]$ is exactly the expected depth of the protocol $\bar \Ccal$. Then, in light of \Cref{prop:fwt-to-qv}, to prove \Cref{thm:boolean_bound_level_one}, it suffices to prove the following.

\begin{lemma}\label{lem:qv-level-one}
$\E\sbra{\sum_{t=1}^{\D} \pbra{\Delta\supZ{t}_1}^2} = O(d)$.
\end{lemma}

We will prove this in two steps. We first show that the only change in the value of the martingale occurs during the orthogonalization step 3(a). 
This is because in each phase, Alice's change of center of mass in steps 3(b) and 3(c) is always orthogonal to $\eta \odot \supv{t}$ so they do not change the value of the martingale $\supZ{t}_1$ as discussed in \Cref{sec:overview}. 
Moreover, recalling \Cref{eqn:overview}, since Alice's node was pairwise clean just before Alice sent the message in step 3(a), the expected change $\BE\left[\left(\Delta\supZ{t+1}_1 \right)^2\right]$ can be bounded in terms of the squared norm of the change that occurred in $\supu{t}$ between the current round and the last round where Alice was in step 3(a). 
A similar argument works for Bob.

Formally, this is encapsulated by the next lemma for which we need some additional definition.
Let $(\supF{t})_t$ be the natural filtration induced by the random walk on the generalized protocol tree with respect to which $\supZ{t}_1$ is a Doob martingale and also $\supu{t}, \supv{t}$ form vector-valued martingales (recall \Cref{prop:vec-martingale}). 
Note that $\supF{t}$ fixes all the rectangles encountered during times $0,\ldots, t$ and thus for $\tau \le t$, the random variables $\supu{\tau},\supv{\tau},\supZ{\tau}_1$ are determined, in particular, they are $\supF{t}$-measurable. Recalling that $\lambda = 100$ is the cleanup parameter, we then have the following. Below we assume without any loss of generality that Alice speaks first and, in particular, we note that Alice speaks in step 3(a) for the first time at time zero. 

\begin{lemma}[Step Size]\label{lem:step_size_square_level_one}
Let $0=\btau_1 < \btau_2 < \cdots \le \D$ be a sequence of stopping times with $\btau_m$ being the index of the round where Alice speaks in step 3(a) for the $m^\text{th}$ time or $\D$ if there is no such round. 
Then, for any integer $m \ge 2$, 
$$
\BE\left[ \left(\Delta\supZ{\btau_m+1}_1\right)^2 ~\bigg|~ \supF{\btau_m}\right] \le \lambda \cdot \vabs{\supv{\btau_m} - \supv{\btau_{m-1}}}^2,
$$
and moreover, for any $t \in \N$, we have that
$$
\BE\left[ \left(\Delta\supZ{t+1}_1\right)^2 ~\bigg|~ \supF{t}, \btau_{m-1} < t <\btau_{m}, \text{Alice speaks at time }t \right] = 0.
$$
A similar statement also holds if Bob speaks where $\V$ is replaced by $\U$ and the sequence $(\btau_m)$ is replaced by $(\btau'_m)$ where $\btau'_m$ is the index of the round where Bob speaks in step 3(a) for the $m^\text{th}$ time or $\D$ if there is no such round. 
\end{lemma}
In particular, we see that the steps 3(b) and 3(c) do not contribute to the quadratic variation and only the steps 3(a) do. Also, since the first time Alice and Bob speak, they start in step 3(a), we also note that $\supu{\btau_1}$ and $\supv{\btau'_1}$ are their initial centers of mass which are both zero.  

We shall prove the above lemma in \Cref{sec:step_size_level_one} and continue with the bound on the quadratic variation here. Using \Cref{lem:step_size_square_level_one}, we have
\begin{align*}
\E\sbra{\sum_{t=1}^{\D} \pbra{\Delta\supZ{t}_1}^2} \le \lambda \cdot\E\sbra{\sum_{m\ge 2} \vabs{\V^{(\btau_m)}-\V^{(\btau_{m-1})}}^2+\vabs{\U^{(\btau'_m)}-\U^{(\btau'_{m-1})}}^2}.
\end{align*}
On the other hand, by the orthogonality of vector-valued martingale differences from \Cref{eqn:martingale-orthogonality-vec}, we have
\begin{align*}
	\E\sbra{\sum_{m \ge 2} \vabs{\V^{(\btau_m)}-\V^{(\btau_{m-1})}}^2} = \E\sbra{\vabs{\V^{(\D)}}^2}.
\end{align*}
A similar statement holds for $(\supu{t})_t$. Therefore, 
\begin{align}\label{eqn:qv-upper-bound}
	 \E\sbra{\sum_{t=1}^{\D} \pbra{\Delta\supZ{t}_1}^2} \le \lambda \cdot\pbra{\E\sbra{\frob{\U^{(\D)}}^2}+\E\sbra{\frob{\V^{(\D)}}^2}}.
\end{align}

We prove the following in \Cref{sec:expected_cleanup_depth} to upper bound the quantity on the right hand side above. Loosely speaking, by an application of level-one inequalities (see \Cref{thm:level_k_ineq}), the lemma below ultimately boils down to a bound on the expected number of cleanup steps. 

\begin{lemma}[Final Center of Mass]\label{lem:expected_norm_level_one}
$
\E\sbra{\vabs{\supu{\D}}^2+\vabs{\supv{\D}}^2} = O(d).
$
\end{lemma}

Since $\lambda = 100$, plugging in the bounds from the above into \Cref{eqn:qv-upper-bound} readily implies \Cref{lem:qv-level-one}. Together with \Cref{prop:fwt-to-qv}, this completes the proof of \Cref{thm:boolean_bound_level_one}.

\subsection[Bounds on Step Sizes]{Bounds on Step Sizes (Proof of \Cref{lem:step_size_square_level_one})}\label{sec:step_size_level_one}

Let us abbreviate $\btau = \btau_m$. Observe that
\begin{align}
\E\sbra{\pbra{\Delta\lZ^{(\btau+1)}_1}^2\mid \Fcal^{(\btau)}}&=\E\sbra{\abra{\U^{(\btau+1)}-\U^{(\btau)}, \Lambda\odot \V^{(\btau)}}^2\mid \Fcal^{(\btau)}}\notag\\
&=\E\sbra{ \abra{\U^{(\btau+1)}, \Lambda\odot \V^{(\btau)} } ^2-\abra{\U^{(\btau)},\Lambda\odot \V^{(\btau)}}^2\mid \Fcal^{(\btau)}},
\label{step_size_level_one_alpha_3}
\end{align}
where the second line is due to $(\supu{t})_t$ being a vector-valued martingale and thus $\E\sbra{\U^{(\btau+1)}\mid \Fcal^{(\btau)}}=\U^{(\btau)}$.  

We first consider the case that at time $\btau$ a new phase starts for Alice. By construction, this means that the current rectangle $\supX{\btau} \times \supY{\btau}$ determined by $\supF{\btau}$ is pairwise clean with parameter $\lambda$, and since Alice is in step 3(a) at the start of a new phase, $\supa{\btau+1}$ is chosen to be the (normalized) component of $\Lambda\odot \V^{(\btau)}$ that is orthogonal to previous directions $\supa{0}, \ldots, \supa{\btau}$. Let $\balpha^{(\btau+1)}:= \abra{\Lambda\odot \V^{(\btau)},\la^{(\btau+1)}}$ be the length of this component before normalization. Note that $\balpha^{(\btau+1)}$ is $\Fcal^{(\btau)}$-measurable (i.e., it is determined by $\Fcal^{(\btau)}$). 

We now claim that components of $\supu{\btau+1}$ and $\supu{\btau}$ are the same along any of the previous directions $\supa{0}, \ldots, \supa{\btau}$. So in \Cref{step_size_level_one_alpha_3}, they cancel out and the only relevant quantity is the component in the direction $\supa{\btau+1}$. 
This follows since, in all the previous steps $t \le \btau$, Alice has already fixed $\sabra{x,\la^{(t)}}$. 
This implies that for any $\supX{\btau}$ and $\supX{\btau+1}$ that are determined by $\supF{\btau+1}$, the inner product with all the previous $\la^{(0)}, \ldots, \la^{(\btau)}$ is fixed over the choice of $x$ from both rectangles. 
Formally, we have that for any $x\in \X^{(\btau)}$ and $x'\in \X^{(\btau+1)} $, it holds that $\sabra{x,\la^{(t)}}=\sabra{x',\la^{(t)}}$ for any $t \le \btau$. In particular, since $\U^{(\btau)}=\com(\X^{(\btau)})$ and $\U^{(\btau+1)}=\com(\X^{(\btau+1)})$ are the corresponding centers of mass, we have that
\begin{equation}\label{step_size_level_one}
\abra{\U^{(\btau+1)}, \la^{(t)}}=\abra{\U^{(\btau)},\la^{(t)}} \text{ for all } t\le \btau.
\end{equation}

This, together with \Cref{step_size_level_one_alpha_3} and recalling that $\balpha^{(\btau+1)}$ is determined by $\Fcal^{(\btau)}$, implies that 
\begin{align}\label{eqn:step_size_level_one_alpha_4}
 \E\sbra{\pbra{\Delta\lZ^{(\btau+1)}_1}^2\mid \Fcal^{(\btau)}} &=\pbra{\balpha^{(\btau+1)}}^2\cdot \E\sbra{ \abra{\U^{(\btau+1)},\la^{(\btau+1)} } ^2-\abra{\U^{(\btau)}, \la^{(\btau+1)}}^2\mid \Fcal^{(\btau)}}.
\end{align}

We now bound the term outside the expectation by the change in the center of mass $\supv{\cdot}$ and the term inside the expectation by the fact that the set is pairwise clean.

\paragraph*{Term Outside the Expectation.} 
Recall that $\supa{\btau+1}$ is chosen to be the (normalized) component of $\Lambda\odot \V^{(\btau)}$ that is orthogonal to the span of $\supa{0}, \ldots, \supa{\btau}$. Since $\Lambda\odot \V^{(\btau_{m-1})}$ is in the span of $\supa{1}, \ldots, \supa{\btau_{m-1}+1}$ and $\btau_{m-1} + 1 \le \btau=\btau_m$, it is orthogonal to $\la^{(\btau+1)}$. Hence,
	\[ 
	\balpha^{(\btau+1)} = \abra{\Lambda\odot \V^{(\btau)},\la^{(\btau+1)}}= \abra{\Lambda\odot\pbra{ \V^{(\btau)}-\V^{(\btau_{m-1})}},\la^{(\btau+1)}}.
	\]
Since $\la^{(\btau+1)}$ is a unit vector and each entry of $\Lambda$ is in $\binpm$, this implies that
\begin{equation}\label{step_size_level_one_alpha}
	\pbra{\balpha^{(\btau+1)}}^2\le \vabs{\V^{(\btau)}-\V^{(\btau_{m-1})}}^2.
	\end{equation}
	
\paragraph*{Term Inside the Expectation.}
Since $(\supu{\tau})$ is a vector-valued martingale with respect to $\supF{\tau}$, and $\supa{\tau+1}$ is $\supF{\tau}$-measurable (determined by $\supF{\tau}$), we have that
\begin{align*}
    \E\sbra{ \abra{\U^{(\btau+1)},\la^{(\btau+1)} } ^2-\abra{\U^{(\btau)}, \la^{(\btau+1)}}^2\mid \Fcal^{(\btau)}}
 = \E\sbra{ \abra{\supu{\tau+1} - \supu{\tau} ,\la^{(\btau+1)} } ^2\mid \supF{\tau}}.
\end{align*}

Since Alice is in step 3(a), her message fixes $\abra{x,\la^{(\btau+1)}}$ at time $\btau$ for every $x \in \supX{\btau+1}$. Thus,
\begin{align}\label{eqn:step_size_level_one_alpha_2}
\E\sbra{\abra{\U^{(\btau+1)} - \U^{(\btau)}, \la^{(\btau+1)}}^2 \mid \supF{\btau}}
&= \E\sbra{\abra{\E_{\lx\sim \gamma}\sbra{\lx\mid \lx\in \X^{(\btau+1)}} - \supu{\tau},\la^{(\btau+1)}}^2 \mid \supF{\btau}} 
\notag\\
&= \E\sbra{\E_{\lx\sim \gamma}\sbra{\abra{\lx-\supu{\btau},\la^{(\btau+1)}}^2\mid \lx\in \X^{(\btau+1)}} \mid \supF{\btau}} \notag \\
&= \E\sbra{\abra{\lx-\supu{\btau},\la^{(\btau+1)}}^2\mid \supF{\btau}},
\end{align}
where the last line follows from the tower property of conditional expectation.

Recall that $\supu{\btau} = \mu(\supX{\btau})$ is the center of mass. Moreover, the unit vector $\supa{\tau+1}$ is determined by $\supF{\tau}$ and also the conditional distribution of $\lx$ conditioned on $\supF{\tau}$ is that of $\lx \sim \gamma$ conditioned on $\lx \in \supX{\tau}$. Thus, using the fact that $\X^{(\btau)}$ is pairwise clean since Alice is in step 3(a), the right hand side in \Cref{eqn:step_size_level_one_alpha_2} is at most $\lambda$.

\paragraph*{Final Bound.} 
Substituting the above in \Cref{eqn:step_size_level_one_alpha_4}, we have 
\begin{align*} 
\E\sbra{\pbra{\Delta\supZ{\btau+1}_1}^2\mid \Fcal^{(\btau)}}
&\le \lambda\cdot \pbra{\balpha^{(\btau+1)}}^2 \le \lambda \cdot \vabs{\V^{(\btau)}-\V^{(\btau_{m-1})}}^2,
\end{align*}
where the second inequality follows from \Cref{step_size_level_one_alpha}. This completes the proof of the first statement.

For the moreover part, let us condition on the event $\btau_{m-1} < t <\btau_{m}$ where Alice speaks at time $t$. Note that such $t$ must all lie in the same phase of the protocol where Alice is the only one speaking. 
So, Bob's center of mass does not change from the time $\btau_{m-1}$ till $t$, i.e., $\V^{(t+1)}=\V^{(\btau_{m-1})}$. 
Thus we have $\Delta\supZ{t+1}_1=\abra{\U^{(t+1)}-\U^{(t)}, \Lambda\odot \V^{(\btau_{m-1})}}$. 
Analogous to \Cref{step_size_level_one},
the component of Alice's center of mass along the previous directions are fixed.
Thus $\abra{\U^{(t+1)}, \la^{(r)}}=\abra{\U^{(t)},\la^{(r)}}$ for all $r \le t$. Furthermore, by construction, $\Lambda \odot \V^{(\btau_{m-1})}$ lies in the linear subspace spanned by $\la^{(0)},\ldots,\la^{(\btau_{m-1} +1)}$. Therefore, since $\btau_{m-1}+1\le t$, it follows that $\Delta\supZ{t+1}_1=0$. 

\subsection[Expected Norm of Final Center of Mass]{Expected Norm of Final Center of Mass (Proof of \Cref{lem:expected_norm_level_one})}
\label{sec:expected_cleanup_depth}

Let $\BH_A = \BH_A^{(\D)}$ be the (random) linear subspace spanned by the vectors $\supa{0}, \ldots, \supa{\D}$ and similarly, let $\BH_B = \BH_B^{(\D)}$ be the linear subspace spanned by the vectors $\supb{0}, \ldots, \supb{\D}$. 
For any linear subspace $V$ of $\Rbb^n$, we denote by $\bPi_V$ and $\bPi_{V^\bot}$ the projectors on the subspace $V$ and its orthogonal complement $V^\bot$ respectively. 
Then, we have that
\[
\vabs{\supu{\D}}^2 = \vabs{\bPi_{H_A} \supu{\D}}^2 + \vabs{\bPi_{H_A^\bot} \supu{\D}}^2 
\text{ and } 
\vabs{\supv{\D}}^2 = \vabs{\bPi_{H_B} \supv{\D}}^2 + \vabs{\bPi_{H_B^\bot} \supv{\D}}^2.
\]

Note that the non-zero vectors in $(\supa{t})_t$ and $(\supb{t})_t$ form an orthonormal basis for the subspaces $\BH_A$ and $\BH_B$ respectively. Moreover, for each $t \le \D$, the inner product $\ip{x}{\supa{t}}$ is fixed for every $x \in \supX{\D}$ and the inner product $\ip{y}{\supb{t}}$ is also fixed for every $y \in \supY{\D}$ where $\supX{\D} \times \supY{\D}$ is the current rectangle determined by $\supF{\D}$. In particular, since $\supu{\D}$ is the center of mass of $\supX{\D}$, this implies that 
\begin{align*}
\vabs{\bPi_{H_A} \supu{\D}}^2 = \sum_{t=1}^{\D} \abra{\supu{\D},\la^{(t)}}^2 &= \sum_{t=1}^{\D} \pbra{\E_{\lx\sim \gamma} \sbra{\abra{\lx,\la^{(t)}}\mid \lx\in \X^{(\D)}}}^2\\
& = \sum_{t=1}^{\D} \E_{\lx\sim \gamma} \sbra{\abra{\lx,\la^{(t)}}^2\mid \lx\in \X^{(\D)}},
\end{align*}
where the second line follows from the inner product being fixed in $\X^{(\D)}$. 
Therefore, we have 
\[ 
\vabs{\supu{\D}}^2 = \underbrace{\sum_{t=1}^{\D} {\E_{\lx\sim \gamma} \sbra{\abra{\lx,\la^{(t)}}^2\mid \lx\in \X^{(\D)}}}}_{\bP_A} + \underbrace{\vabs{\bPi_{H_A^\bot} \supu{\D}}^2}_{\bQ_A}.
\]
In an analogous fashion, 
\[ 
\vabs{\supv{\D}}^2 = \underbrace{\sum_{t=1}^{\D} {\E_{\ly\sim \gamma} \sbra{\abra{\ly,\lb^{(t)}}^2\mid \ly\in \Y^{(\D)}}}}_{\bP_B} + \underbrace{\vabs{\bPi_{H_B^\bot} \supv{\D}}^2}_{\bQ_B}.
\]

We next show that both $\E[\bP_A+\bP_B]$ and $\E[\Q_A+\Q_B]$ are at most $O(d)$. 
The former follows from stopping time and concentration arguments laid out in the overview that there cannot be too many orthogonal directions where $\BE\sbra{\ip{\lx}{\supa{t}}^2}$ is large. 
The latter follows from an application of level-one inequalities.

We will bound the norm of the projection on the subspaces $\BH_A$ and $\BH_B$, which corresponds to the quantity $\BE[\bP_A+\bP_B]$, in \Cref{sec:level_one_inside_subspace} and bound the norm of the projection on the orthogonal subspaces $\BH_A^\bot$ and $\BH_B^\bot$,  which corresponds to the quantity $\BE[\Q_A+\Q_B]$, in \Cref{sec:level_one_inside_complement_subspace}.

\subsubsection{Projection on the Subspaces \texorpdfstring{$\BH_A$}{H\textunderscore A} and \texorpdfstring{$\BH_B$}{H\textunderscore B}}\label{sec:level_one_inside_subspace}

We shall show that the expected norm of the final center of mass when projected on the subspaces $\BH_A$ and $\BH_B$ is 
\[ \E[\bP_A+\bP_B] = O(d).\]

Towards this end, define the random variable $\K_t = \K_t(\lx,\ly) =\abra{\lx,\la^{(t)}}^2+ \abra{\ly,\lb^{(t)}}^2$ for each $t \in \Nbb$. 
Note that the vectors $\la^{(t)}$'s are being chosen adaptively depending on the previous inner products $\ip{\lx}{\la^{(\tau)}}$ for $\tau < t$, as well as the Boolean communication bits from step 3(b), thus they are functions of $\lx$ and $\ly$ as well here. Observe that
\[
\E\sbra{\bP_A+\bP_B}= \E\sbra{\sum_{t=1}^{\D} \E\sbra{\K_t \mid \supF{\D}}}= \E_{\lx,\ly\sim \gamma} \sbra{ \sum_{t=1}^{\D} \K_t}.
\]

We now divide the time sequence into successive intervals of different lengths $r\cdot 4d$ for $r=1,2, \ldots$.
Then we bound the expected sum of $\K_t$ within each time interval by $O(r d)$. 
We further argue that the probability that the stopping time $\D$ lies in the $r$-th interval is at most $2\cdot 2^{-r}$. In particular, for $r\in\Nbb$, letting interval 
$I_r=\cbra{\binom{r}{2}\cdot 4d+1,\ldots,\binom{r+1}{2}\cdot 4d}$,
which is of length $4dr$, we show the following.
\begin{claim}\label{lem:depth_tail_bound_level_one}
For any $r\in \Nbb$, we have 
\[ \E_{\lx,\ly\sim \gamma}
\sbra{\sum_{t\in I_r} \K_t \mid \D>\binom{r}{2}\cdot 4d}
\le 20d r
+4 \ln\pbra{\dfrac{1}{\Pr\sbra{\D>\binom{r}{2}\cdot 4d}}}.\]
\end{claim}

We shall prove the above claim later since it is the most involved part of the proof. The previous claim readily implies the following probability bounds.
\begin{claim}\label{lem:depth_tail_bound_level_one_part_two}
For any $r\in \Nbb$, we have $\Pr\sbra{\D >\binom{r}{2}\cdot 4d}\le 2\cdot 2^{-r}$.
\end{claim}

\begin{proof}[Proof of \Cref{lem:depth_tail_bound_level_one_part_two}] 
We bound $\Pr\sbra{\D > \binom{r}{2}\cdot 4d}$ by induction on $r$. The claim trivially holds for $r=1$. 

Now we proceed to analyze the event $\D\ge\tbinom{r+1}2\cdot4d$.
Observe that \Cref{clm:finite_steps_level_one} implies that there are at most $2d$ many step 3(a) and 3(b) throughout the protocol.
Thus if the event above occurs, there are at least $4dr-2d\ge 2dr$ many time steps $t \in I_r$ where the process is in step 3(c).

By the definition of the cleanup step, if $X \times Y$ is a rectangle determined\footnote{It suffices to consider such events since we have a product measure on $\supX{t} \times \supY{t}$ conditioned on $\supF{t}$ and $\D$ is a stopping time and is $\supF{t}$-measurable (i.e., determined by the randomness in $\supF{t}$).} by $\supF{t-1} \cap \{\D > \binom{r}{2}\cdot 4d\}$ where the process is in step 3(c) and Alice speaks, then
\[ 
\BE_{\lx \sim \gamma}\sbra{\K_t \mid (\lx,\ly) \in X \times Y} = \BE_{\lx \sim \gamma}\sbra{\ip{\lx}{\supa{t}}^2 \mid \lx\in X} \ge \BE_{\lx \sim \gamma}\sbra{\ip{\lx - \com(X)}{\supa{t}}^2 \mid \lx\in X} \ge \lambda, 
\]
where $\lambda=100$ is the cleanup parameter and $\com(X) = \BE_{\lx \sim \gamma}[\lx \midd \lx \in X]$ is the center of mass. This is because $\supa{t}$ is chosen to be a unit vector in a direction where the current set (conditioned on the history) is not pairwise clean. 
A similar statement holds if Bob speaks in step 3(c) for the random variable $\ip{\ly}{\supb{t}}^2$ where $\ly$ is sampled from $\gamma$ conditioned on $Y$. 

By the tower property of conditional expectation, the above implies that
\[
100\cdot2dr\cdot \Pr\sbra{\D > {\textstyle\binom{r+1}{2}}\cdot 4d \mid \D > {\textstyle\binom{r}{2}}\cdot 4d} \le \E\sbra{\sum_{t\in I_r} \K_t \mid \D > {\textstyle\binom{r}{2}}\cdot 4d}.
\]
Recall that \Cref{lem:depth_tail_bound_level_one} implies that the right hand side is at most $\le 20d r + 4 \ln\pbra{\frac1{\Pr[\D>\tbinom{r}{2} \cdot 4d]}}$.
We consider two cases:
\begin{enumerate}
\item[(i)] if $\Pr[\D>\binom{r}{2} \cdot 4d] \le 2^{-r}$, then clearly $\Pr[\D>\binom{r+1}{2} \cdot 4d] \le 2^{-r}$ as well as required;
\item[(ii)] otherwise $\Pr[\D>\binom{r}{2} \cdot 4d] \ge 2^{-r}$ and $20d r + 4\pbra{\frac1{\Pr[\D>\tbinom{r}{2} \cdot 4d]}} \le 20dr + 4r$, then it follows that 
\[ 
\Pr\sbra{\D > \textstyle\binom{r+1}{2}\cdot 4d \mid \D > \textstyle\binom{r}{2}\cdot 4d}\le 1/2,
\]
and by induction this implies that $\Pr\sbra{\D > \textstyle\binom{r+1}{2}\cdot 4d} \le 1/2\cdot\Pr\sbra{\D > \textstyle\binom{r}{2}\cdot 4d}\le2^{-r}$.\qedhere
\end{enumerate} 
\end{proof}

These claims imply that 
\begin{align*} 
\E[\bP_A+\bP_B] & \le \E\sbra{\sum_{r=0}^\infty 1\sbra{\D>{\textstyle\binom{r}{2}}\cdot 4d}\cdot \sum_{t\in I_r} \K_t}
\\&= \sum_{r=0}^\infty \Pr[ \D> {\textstyle\binom{r}{2}}\cdot 4d] \cdot\E\sbra{\sum_{t\in I_r} \K_t\mid \D > {\textstyle\binom{r}{2}}\cdot 4d} \\
&\le \sum_{r=0}^\infty\pbra{2^{1-r}\cdot O(r d)  +   4 \cdot\Pr[ \D> {\textstyle\binom{r}{2}}\cdot 4d] \cdot \ln\pbra{\tfrac{1}{ \Pr\sbra{\D> {\textstyle\binom{r}{2}}\cdot 4d}}}}\\
&\le \sum_{r=0}^\infty\pbra{2^{1-r}\cdot O(r d)  +  O\pbra{(r+1)2^{-r}}}\le O(d),
\end{align*}
where the last line uses the fact that $x\ln(1/x)\le O((r+1)2^{-r})$ for $0\le x\le2\cdot2^{-r}$ and $r\in\N$.
This proves the desired bound on $\E[\bP_A+\bP_B]$ assuming \Cref{lem:depth_tail_bound_level_one}  which we prove next.

\begin{proof}[Proof of \Cref{lem:depth_tail_bound_level_one}] 
To prove the claim, we need to analyze the expectation of $\sum_{t \in I_r} \K_t$ under $\lx, \ly$ sampled from $\gamma$ conditioned on the event $\D \ge \binom{r}{2} \cdot 4d$. 

We first describe an equivalent way of sampling from this distribution which will be easier for analysis. 
First, we recall that the definition of the cleanup protocol implies that the Boolean communication in $\bar \Ccal$ is solely determined by the previous Boolean communication, since it is specified by the original protocol $\tilde{\Ccal}$ (and thus $\Ccal$) before the cleanup. 

Let us fix any Boolean string $c\in\{0,1\}^*$ that is a valid Boolean transcript in the original communication protocol $\tilde \Ccal$.
This defines a rectangle $X_c\times Y_c\subseteq\Rbb^n\times\Rbb^n$ consisting of all pairs of inputs to Alice and Bob that result in the Boolean transcript $c$ in $\tilde\Ccal$.
If we sample $\lx,\ly\sim \gamma$ conditioned on $\D>\binom{r}{2}\cdot 4d$ and output the unique $(\X_c,\Y_c)$ such that $(\lx,\ly)\in \X_c\times \Y_c$, 
we obtain a distribution on rectangles. We use $\gamma(X_c\times Y_c\,|\,\D >\binom{r}{2}\cdot 4d)$ 
to denote the probability of obtaining $X_c\times Y_c$ by this sampling process so that $\sum_c \gamma(X_c\times Y_c\,|\,\D >\binom{r}{2}\cdot 4d)=1$. 

Now consider the following two-stage sampling process. First, we sample a rectangle $X_c \times Y_c$ according to the above distribution, and then we sample  the inputs $\lx, \ly$ sampled from $\gamma_n$ conditioned on the event that $\{(\lx,\ly)\in X_c \times Y_c\} \wedge \{\D> \binom{r}{2}\cdot 4d\}$. We shall show the following claim for any rectangle $X_c \times Y_c$ that could be sampled in the first step. 

\begin{claim}\label{eq:depth_tail_bound_level_one_eq1}
$\E_{\lx,\ly\sim \gamma}\sbra{\sum_{t\in I_r} \K_t \mid \D > 4d\tbinom{r}{2},(\lx,\ly)\in X_c\times Y_c } \le 12dr + 4\ln
\pbra{\tfrac{1}{\Pr[\D > 4d\tbinom{r}{2},(\lx,\ly)\in X_c\times Y_c]}}$.
\end{claim}

Assuming the above, and taking an expectation over $X_c\times Y_c$ drawn with probability $\gamma(X_c\times Y_c\,|\,\D >\binom{r}{2}\cdot 4d)$, we immediately obtain \Cref{lem:depth_tail_bound_level_one}:
\begin{align*} 
&\E_{\lx,\ly\sim \gamma}\sbra{\sum_{t\in I_r} \K_t \mid \D> {\textstyle\binom{r}{2}}\cdot 4d}\\
&\le 12dr + 4\cdot\sum_{\substack{c\in\{0,1\}^*,|c|\le d}} \gamma(X_c\times Y_c|\D >{\textstyle\binom{r}{2}}\cdot 4d)\cdot 
\pbra{\ln\pbra{\tfrac{1}{\gamma(X_c\times Y_c|\D >\binom{r}{2}\cdot 4d)}}+
\ln\pbra{\tfrac{1}{\Pr[\D >\binom{r}{2}\cdot 4d]}}}\\
&\le 12dr +  4\cdot \ln(3^d) + 4\cdot \ln\pbra{\tfrac{1}{\Pr[\D >\binom{r}{2}\cdot 4d]}} 
\tag{by concavity of $\ln(\cdot)$}
\\
&\le 20dr + 4\cdot \ln\pbra{\tfrac{1}{\Pr[\D >\binom{r}{2}\cdot  4d]}}.
\tag*{\qedhere}
\end{align*}
\end{proof}

To complete the proof, we now prove \Cref{eq:depth_tail_bound_level_one_eq1}.

\begin{proof}[Proof of \cref{eq:depth_tail_bound_level_one_eq1}]
Fix any $c$ such that $\gamma(X_c\times Y_c\,|\,\D >\binom{r}{2}\cdot 4d)>0$. 
We will bound the expectation of the quantity $\sum_{t\in I_r} \K_t = \sum_{t\in I_r} \abra{\lx,\la^{(t)}}^2 +\abra{\ly,\lb^{(t)}}^2$ where $\lx, \ly$ are sampled from $\gamma_n$ conditioned on the event that $\{(\lx,\ly)\in X_c \times Y_c\} \wedge \{\D > \binom{r}{2}\cdot 4d\}$.
Note that $\supa{t}, \supb{t},\D$ are functions of the previous messages of the protocol and hence also the inputs $\lx, \ly$. Once we condition on the above event, the Boolean communication is also fixed to be $c$.

To analyze the above conditioning, we first do a thought experiment and consider a different protocol that takes standard Gaussian inputs (without any conditioning) and show a tail bound for the random variable $\sum_{t \in I_r} \K_t$ for this new protocol. In the last step, we will use it to compute the expectation we ultimately want.

\paragraph*{Protocol $\Ccal_c$.} 
The protocol $\Ccal_c$ always communicates according to the fixed transcript $c$ in a Boolean communication step and otherwise according to the cleanup protocol $\bar\Ccal$ on any input $x,y$. Consider the random walk on this new protocol tree where the inputs $\lx, \ly \sim \gamma$ (without any conditioning). Let $(\Gcal^{(t)})_t$ be the associated filtration of the new protocol $\Ccal_c$ which can be identified with the collection of all partial transcripts till time $t$. Note that the vectors $\supa{t}$ and $\supb{t}$ in this new protocol are determined only by the previous real communication since the Boolean communication is fixed to $c$. This also implies that the vectors $\supa{t}$ and $\supb{t}$ form a predictable sequence with respect to the filtration $(\Gcal^{(t)})_t$. Moreover, by the definition of the protocol the next non-zero vector $\supa{\cdot}$ is chosen to be a unit vector orthogonal to the previously chosen $\supa{\cdot}$'s and the same holds for the vectors $\supb{\cdot}$.

We denote by $\K_t^{(c)}$ the random variable that captures $\K_t$ for the protocol $\Ccal_c$, i.e., $\K_t^{(c)} = \abra{\lx,\la^{(t)}}^2 +\abra{\ly,\lb^{(t)}}^2$ for $\lx, \ly \sim \gamma$ and $\la^{(t)}, \lb^{(t)}$ defined by the protocol $\Ccal_c$.
Observe that if $(\lx, \ly) \in X_c \times Y_c$ then $\K_t^{(c)} = \K_t$.

Consider the behavior of the protocol $\Ccal_c$ at some fixed time $t$. The nice thing about the protocol $\Ccal_c$ is that conditioned on all previous real messages for $\tau < t$, both $\lx$ and $\ly$ are standard Gaussian distributions on an affine subspace of $\R^n$ (defined by the previous messages).
Then, at time $t$, since $\supa{t}$ is orthogonal to the directions used in all previous real messages, it follows that the distribution of $\abra{\lx,\la^{(t)}}$ conditioned on any event in $\Gcal^{(t-1)}$ is an independent standard Gaussian for every $t$ if $\supa{t}$ is non-zero. The same holds for $\abra{\ly,\lb^{(t)}}$ as well. This last fact uses that the projection of a multi-variate standard Gaussian $\gamma_n$ in orthonormal directions yields independent real-valued standard Gaussians.

This implies that each new $\abra{\lx,\la^{(t)}}^2$ and $\abra{\ly,\lb^{(t)}}^2$ is an independent chi-squared random variable conditioned on the history (up to depth $\binom r2\cdot4d$) of the random walk. Therefore, \Cref{thm:chi_squared_concentration} implies that
\[ 
\Pr_{\lx,\ly\sim \gamma}\sbra{ \sum_{t\in I_r} \K_t^{(c)}(\lx,\ly) \ge 2|I_r|+ s \mid \Gcal^{(\binom{r}{2}\cdot 4d)}} \le e^{-s/4}. 
\]
Since $|I_r|\le 4dr$, we have 
$\Pr_{\lx,\ly\sim \gamma}\sbra{ \sum_{t\in I_r} \K^{(c)}_t(\lx,\ly) \ge 8dr+ s } \le e^{-s/4}.$

\paragraph*{Computing the Original Expectation.} 
Let us compare the probability of the above tail event in the original protocol $\bar\Ccal$ 
where inputs $\lx, \ly$ are sampled from  $\gamma$ conditioned on the event that $\{(\lx,\ly)\in X_c \times Y_c\} \wedge \{\D > \binom{r}{2}\cdot 4d\}$. 
We can write
\begin{align}\label{eq:tail_of_kt}
&\phantom{\le}
\Pr_{(\lx,\ly)\sim \gamma}\sbra{ \sum_{t\in I_r} \K_t(\lx,\ly)\ge 8dr + s \mid \D> {\textstyle\binom{r}{2}}\cdot 4d, (\lx,\ly) \in X_c\times Y_c}\\
&= \frac{\Pr_{\lx,\ly\sim \gamma}\sbra{ \sum_{t\in I_r} \K_t(\lx,\ly)\ge 8dr+ s,(\lx,\ly) \in X_c\times Y_c, \D > \binom{r}{2}\cdot 4d}}{ \Pr_{\lx,\ly\sim \gamma}\sbra{ (\lx,\ly)\in X_c\times Y_c, \D >\binom{r}{2}\cdot 4d}}.
\notag
\end{align}
We then bound the numerator by
\begin{align*}
&\phantom{\le}
\Pr_{\lx,\ly\sim \gamma}\sbra{ \sum_{t\in I_r} \K_t(\lx,\ly)\ge 8dr+ s, (\lx,\ly) \in X_c\times Y_c, \D > {\textstyle\binom{r}{2}}\cdot 4d}\\
&= \Pr_{\lx,\ly\sim \gamma}\sbra{ \sum_{t\in I_r} \K_t^{(c)}(\lx,\ly)\ge 8dr+ s, (\lx,\ly) \in X_c\times Y_c, \D > {\textstyle\binom{r}{2}}\cdot 4d}
\tag{if $(\lx,\ly)\in X_c \times Y_c$ then $\K_t^{(c)} = \K_t$}
\\
&\le  \Pr_{\lx,\ly\sim \gamma}\sbra{ \sum_{t\in I_r} \K_t^{(c)}(\lx,\ly)\ge 8dr+ s}  \le e^{-s/4}.
\end{align*}

Note that the inequality gives us an exponential tail on \Cref{eq:tail_of_kt}:
$$
\Cref{eq:tail_of_kt}
\le e^{-s/4}\cdot\pbra{\Pr_{\lx,\ly\sim \gamma}\sbra{ (\lx,\ly)\in X_c\times Y_c, \D >\binom{r}{2}\cdot 4d}}^{-1}.
$$
We can now integrate the above inequality to get an upper bound on the expected value of $\sum_{t \in I_r} \K_t$ under the distribution of interest. In particular, since for any non-negative random variable $\lw$, the following holds for any parameter $\alpha\ge0$:
\[ 
\BE[\lw] = \int_0^{+\infty} \Pr[\lw \ge z]\sd{z} 
\le \alpha + \int_\alpha^{+\infty} \Pr[\lw \ge z]\sd{z}
=\alpha + \int_0^{+\infty} \Pr[\lw \ge\alpha+z]\sd{z},
\]
we derive the following by taking $\alpha = 8dr+4\ln\pbra{\frac{1}{\Pr_{\lx,\ly\sim \gamma}\sbra{(\lx,\ly)\in X_c\times Y_c, \D >\binom{r}{2}\cdot 4d}}}$:
\begin{align*}
&\E_{(\lx,\ly)\sim \gamma}\sbra{ \sum_{i\in I_r} \K_t(\lx,\ly) \mid \D> {\textstyle\binom{r}{2}}\cdot 4d, (\lx,\ly)\in X_c\times Y_c } \\
&\qquad\le\alpha+\int_0^{+\infty}e^{-z/4}\sd z
=\alpha + 4\\
&\qquad \le 12dr + 4\ln\pbra{\dfrac{1}{\Pr_{\lx,\ly\sim \gamma}\sbra{ (\lx,\ly)\in X_c\times Y_c, \D >\binom{r}{2}\cdot 4d}}}.
\end{align*}
This completes the proof of \Cref{eq:depth_tail_bound_level_one_eq1}.
\end{proof}

\subsubsection{Projection on the Orthogonal Subspaces \texorpdfstring{$\BH_A^\bot$}{of H\textunderscore A} and \texorpdfstring{$\BH_B^\bot$}{of H\textunderscore B}}\label{sec:level_one_inside_complement_subspace}

We shall show that the expected norm of the final center of mass when projected on the subspaces $\BH_A^\bot$ and $\BH_B^\bot$ is 
\[ \E[\Q_A+\Q_B] = O(d).\]

Recall that $\bQ_{A} = \vabs{\bPi_{\BH_A^\bot} \supu{\D}}^2$ where $\BH_A$ is the (random) linear subspace spanned by the orthonormal set of vectors $\supa{0}, \ldots, \supa{\D}$ and $\BH_A^\bot$ its orthogonal complement. 
Moreover, the vectors $\supa{t}$ are determined by the previous Boolean and real communication. A similar statement holds for $\bQ_B$ and the vectors $\supb{t}$ as well.

The proof will follow in two steps. We will first show that one can bound the norm of the projection $\bPi_{\BH_A^\bot} \supu{d}$, which turns out to be the Gaussian center of mass of a set that lives in the subspace $\BH_A^\bot$, in terms of the logarithm of the inverse relative measure with respect to the subspace. Note that the Gaussian measure here is the Gaussian measure $\gamma_{\BH_A^\bot}$ on the subspace $\BH_A^\bot$. The case for $\bPi_{\BH_B^\bot} \supu{d}$ will be similar. The second step will use information theory-esque convexity argument to show that on average the logarithm of the inverse relative measure is small.

For the first part, we observe that if we sample $\lx,\ly \sim \gamma$ and take a random walk on this protocol tree, we obtain a probability measure over transcripts which includes both real and Boolean values. Recall that the Boolean transcript is determined by the original protocol and only depends on the previous Boolean communication and the real transcript is sandwiched between the Boolean communication. 
Let $\bell = (\bc, \br)$ denote the random variable representing the full transcript of the generalized protocol where $\bc$ is the Boolean communication and $\br$ is the real communication. For any given transcript $\bell$, let $\X_{\bell} \times \Y_{\bell}$ denote the corresponding rectangle consists of inputs reaching the leaf, and let $\X_{\bc}\times \Y_{\bc}$ (for $\X_{\bc},\Y_{\bc} \subseteq\Rbb^n$) denote the rectangle consisting of all pairs of inputs to Alice and Bob that result in the Boolean transcript $\bc$. Note that the real communication $\br$ together with $\bc$ fixes the subspaces $\BH_A$ and $\BH_B$ and particular affine shifts $\bs_A$ and $\bs_B$ of those subspaces depending on the value of the inner products determined by the full transcript. In particular, the rectangle $\X_{\bell} \times \Y_{\bell}$ consistent with the full transcript $\bell = (\bc,\br)$ is given by $\X_{\bell} = \X_{\bc} \cap (\BH_A + \bs_A)$ and $\Y_{\bell} = \Y_{\bc} \cap (\BH_B + \bs_B)$, i.e., taking (random) affine slices of the original sets.

Note that $\supu{\D}$ and $\supv{\D}$ are distributed as the center of masses of the final rectangle $\X_{\bell} \times \Y_{\bell}$, and thus is suffices to look at the rectangles for the rest of the argument. Since $\X_{\bell}$ (resp., $\Y_{\bell}$) lies in some affine shift of $\BH_A^\bot$ (resp., $\BH_B^\bot$), defining the relative center of mass for a set $A$ that lives in the ambient linear subspace $V$, as $\com_V(A) = \BE_{\lx \sim \gamma_V}[\lx \midd \lx \in A]$ where the Gaussian measure $\gamma_V$ is on the subspace $V$, it follows that
\begin{align*}
	 \E\sbra{\bQ_A +\bQ_B} &= \BE\left[\vabs{\bPi_{\BH_A^\bot} \supu{\D}}^2 + \vabs{\bPi_{\BH_A^\bot} \supu{\D}}^2\right] =  \BE_{\bell}\left[\|\com_{\BH_A^\perp}(\bPi_{\BH_A^\bot}\X_{\bell})\|^2 + \|\com_{\BH_B^\perp}(\bPi_{\BH_B^\bot}\Y_{\bell})\|^2\right].
\end{align*}

Recalling that $\gamma_\rel$ is the Gaussian measure of a set relative to its ambient space, we will show: 

\begin{claim}\label{eqn:relative-measure}
   $\|\com_{\BH_A^\perp}(\bPi_{\BH_A^\bot}\X_{\bell})\|^2 \le 2e^2 \ln\pbra{\dfrac{e}{\gamma_\rel\pbra{\X_{\bell}}}}$ and $\|\com_{\BH_B^\perp}(\bPi_{\BH_B^\bot}\Y_{\bell})\|^2 \le 2e^2 \ln\pbra{\dfrac{e}{\gamma_\rel\pbra{\Y_{\bell}}}}$.
\end{claim}
Note that we can ignore the case when $\gamma_\rel(\X_{\bell})$ is zero above, since we will eventually take an expectation over $\bell$ and almost surely this measure is non-zero.  

Using the previous claim, 
\begin{align*}
	 \E\sbra{\bQ_A +\bQ_B} &= \BE\left[\vabs{\bPi_{\BH_A^\bot} \supu{\D}}^2 + \vabs{\bPi_{\BH_A^\bot} \supu{\D}}^2\right] \le2e^2 \cdot \E_{\bell}\sbra{\ln\pbra{\frac{e}{\gamma_\rel\pbra{{\X_{\bell} \times \Y_{\bell}}}}}}.
\end{align*}

For the second step of the proof, we show the next claim which relies on convexity arguments to bound  the right hand side above by $O(d)$. This is similar in spirit to chain-style arguments from information theory.

\begin{claim}\label{claim:convexity}
$\E_{\bell}\sbra{\ln\pbra{\dfrac{e}{\gamma_\rel\pbra{{\X_{\bell} \times \Y_{\bell}}}}}} = O(d)$.
\end{claim}

This gives us the final bound $\E\sbra{\bQ_A +\bQ_B} = O(d)$ assuming the claims which we now prove.

\begin{proof}[Proof of \Cref{eqn:relative-measure}]
    
We can bound the norm of the above projection by an application of the Gaussian level-one inequality (\Cref{thm:level_k_ineq}), which, by rotational symmetry, implies that if $A$ is a subset of a linear subspace $V$ with non-zero measure, then 
\begin{align}\label{eqn:level-one-inequality}
    \|\com_V(A)\|^2 \le 2e^2 \ln\left(\frac{e}{\gamma_V(A)}\right),
\end{align}
where recall that $\com_V(A) = \BE_{\lx \sim \gamma_V}[\lx \midd \lx \in A]$ is the center of mass with respect to the Gaussian measure $\gamma_V$ on the subspace $V$.

If we run the generalized protocol on $\lx, \ly \sim \gamma$ and condition on getting the full transcript $\bell$, the conditional probability measure on $\bPi_{\BH_A^\bot}\lx$ is that of the Gaussian measure $\gamma_{\BH_A^\bot}$ conditioned on $\lx \in \X_{\bell} - \bs_A$ and $\bPi_{\BH_A^\bot}\ly$ is that of the Gaussian measure $\gamma_{\BH_B^\bot}$ conditioned on $\ly \in \Y_{\bell} - \bs_B$ and they are independent. 
This follows from the fact that so far the parties have fixed inner products along a basis for the orthogonal subspaces $\BH_A$ and $\BH_B$ and the fact the projection of a standard Gaussian on orthogonal subspaces are independent.

Thus, applying \Cref{eqn:level-one-inequality}, we have 
\begin{align*}
    \|\com_{\BH_A^\bot}(\bPi_{\BH_A^\bot}\X_{\bell})\|^2 \le 2e^2 \ln\left(\frac{e}{\gamma_{\BH_A^\bot}(\X_{\bell}-\bs_A)}\right) = 2e^2 \ln\left(\frac{e}{\gamma_\rel(\X_{\bell})}\right),
\end{align*}
where the last line follows since $\BH_A+\bs_A$ is the ambient space for $\X_{\bell}$ (this holds almost surely) and $\gamma_\rel(S) = \gamma_V(S-t)$ if $V+t$ is the ambient space of $S$. A similar argument proves the bound on $\|\com_{\BH_B^\bot}(\bPi_{\BH_B^\bot}\Y_{\bell})\|^2$.
\end{proof}

\begin{proof}[Proof of \Cref{claim:convexity}]
    For this claim, it will be convenient to consider a different generalized protocol $\Ccal'$ that generates the same distribution on the leaves $\bell$. In particular, since the Boolean messages in the generalized protocol $\bar\Ccal$ only depend on the previous Boolean messages, one can first send all the Boolean messages $\bc$, and then send all the real messages $\br$ choosing them according to the protocol $\bar\Ccal$ depending on the previous real messages and the (partial) Boolean transcript. Note that the protocol $\Ccal'$ generates the same distribution on the leaves $\bell$ when the inputs $\lx, \ly \sim \gamma_n$. In particular, the real communication only partitions \footnote{We remark that this protocol $\Ccal'$ suffices for proving this claim since we are looking only at the leaves. However, unlike \Cref{lem:step_size_square_level_one}, directly bounding the expected quadratic variation of the martingale corresponding to the protocol $\Ccal'$ is difficult.} each rectangles $X_c \times Y_c$ that corresponds to the Boolean transcript $c$ into affine slices.

    For rest of the claim, we now work with the protocol $\Ccal'$ where the Boolean communication happens first. To prove the claim, we condition on a Boolean transcript $\bc=c$ and by induction show that 
    \begin{align}\label{eqn:rectangle-measure}
          \BE_{\br}\sbra{\ln\pbra{\dfrac{e}{\gamma_\rel(\X_{(c,\br)} \times \Y_{(c,\br)})}} \mid \bc=c} \le \ln\pbra{\dfrac{e}{\gamma_\rel(X_c \times Y_c)}},     
    \end{align}
     where $(c, r)$ is the full transcript and  $X_c \times Y_c$ is the rectangle containing all the inputs such that Boolean transcript is $c$. Note that $\gamma_\rel(X_c \times Y_c)$ is the probability of obtaining the Boolean transcript $c$ since the ambient space of $X_c$ and $Y_c$ is $\Rbb^n$.

    Then, taking expectation over the Boolean transcript $c$,
    \begin{align*}
         \BE_{\bell}\sbra{\ln\pbra{\dfrac{e}{\gamma_\rel(\X_{\bell} \times \Y_{\bell})}}} &\le \BE_{\bc}\sbra{\ln\pbra{\dfrac{e}{\gamma_\rel(\X_{\bc} \times \Y_{\bc})}}}  \\
         &= \sum_{\substack{c\in\{0,1\}^*,|c|\le d}} \Pr[\bc=c] \ln\pbra{\dfrac{e}{\Pr[\bc=c]}} \\
          & \le \ln(2e\cdot 2^d) = O(d),  
    \end{align*}
    where the last line follows from concavity.

    \paragraph*{Induction.} 
    To complete the proof, we now show \Cref{eqn:rectangle-measure} by induction. For this, let us look at an intermediate step $t$ in $\Ccal'$ where the Boolean communication is fixed to $c$ and Alice and Bob have exchanged some real messages $r_{<t} := r_1,\ldots, r_{t-1}$. Let the current rectangle be $X_{(c,r_{<t})} \times Y_{(c,r_{<t})}$ and it is Alice's turn to speak. Note that $X_{(c,r_{<t})}$ and $Y_{(c,r_{<t})}$ live in some affine subspaces at this point and in the current round, Alice sends the inner product of her input $x$ with a vector $a^{(t)}$ that is determined by the previous messages and orthogonal to the ambient space of $X_{(c,r_{<t})}$. At this step, Bob's set $Y_{(c,r_{<t})}$ does not change at all. 
    We shall show that in each step, the log of the inverse of the relative measure of the current rectangle does not increase on average over the next message:
    \begin{align}\label{eqn:rectangle-measure-one-step}
          \BE_{\br_{\le t}}\sbra{\ln\pbra{\dfrac{e}{\gamma_\rel(\X_{(c,\br_{\le t})})}} \mid \bc=c, \br_{<t}=r_{<t}} \le \ln\pbra{\dfrac{e}{\gamma_\rel(X_{(c,r_{<t})})}},   
    \end{align}
   and an analogous statement holds when Bob speaks. 
   Taking an expectation over $\br_{<t}$, the above directly applies \eqref{eqn:rectangle-measure} by a straightforward backward induction:
    \begin{align*}
          \BE_{\br_{\le t}}\sbra{\ln\pbra{\dfrac{e}{\gamma_\rel(\X_{(c,\br_{\le t})} \times \Y_{(c,{\br_{\le t})}})}} \mid \bc=c} &\le \BE_{\br_{<t }}\sbra{\ln\pbra{\dfrac{e}{\gamma_\rel(\X_{(c,\br_{< t})} \times \Y_{(c,{\br_{< t})}})}} \mid \bc=c}\\
          & \le \cdots \le  \ln\pbra{\dfrac{e}{\gamma_\rel(X_c \times Y_c)}}.   
    \end{align*}
    
   To see \Cref{eqn:rectangle-measure-one-step}, let us write $X := X_{(c,r_{<t})}$ for Alice's current set.  Recall that since we have fixed the history, Alice has fixed inner product with some orthogonal directions $a^{(1)}, \ldots, a^{(t-1)}$ and she has decided on the next direction $a := a^{(t)}$ along which she will send the next inner product.
   Thus, $X$ lives in some fixed affine subspace $V^\bot+s$ where $V$  is the span of $a^{(1)},\ldots,a^{(t-1)}$ and the next message $r := r_t = \ip{x}{a}$.
   Moreover, conditioned on the history till this point, the conditional probability distribution on Alice's input $\lx \in \Rbb^n$ can be described as follows: the projections corresponding to the non-zero vectors in the sequence $a^{(1)}, \ldots, a^{(t-1)}$, i.e., the inner products $\abra{\lx,a^{(\tau)}}$ where $a^{(\tau)}\neq0$ for $\tau < t$, are fixed according to the shift $s$, while the distribution on the orthogonal complement $V^\bot$ is that of the Gaussian measure $\gamma_{V^\bot}$ on the subspace $V^\bot$ after conditioning on the event that $\lx \in X - s$ (which lives in $V^\bot$). This uses that projections of a standard $n$-dimensional Gaussian in orthogonal directions are independent.

Let $k$ be the dimension of $V$ where $k<n$. Then, by doing a linear transformation, we may assume that $V^\bot= \Rbb^{n-k}$ (and thus $X \subseteq \Rbb^{n-k}$ and the shift $s$ fixes the coordinates $n-k+1$ through $n$) and $a = e_1$, i.e., in the current message Alice reveals the first coordinate of $\lx \in \Rbb^{n-k}$ where $\lx$ is sampled from $\gamma_{n-k}$ conditioned on $\lx \in X$. In this case, $\gamma_\rel$ in the left hand side of \Cref{eqn:rectangle-measure-one-step} is exactly $\gamma_\rel(X \cap \{x_1=r\})$ if Alice sends $r$ as the message, while for the right hand side of  \Cref{eqn:rectangle-measure-one-step}, we have $\gamma_\rel(X) = \gamma_{n-k}(X)$. Denoting by  $\sd\mu_{x_1}$ the probability density function of $\lx_1$, our statement boils down to showing
   \begin{align*}
       \int_{\Rbb} \ln\pbra{\dfrac{e}{\gamma_\rel(X \cap \{x_1=r\})}} \sd\mu_{x_1}(r) &\le  \ln\pbra{\dfrac{e}{\gamma_{n-k}(X)}}.  
   \end{align*}
   
    We show the above by explicitly writing the probability density function $\sd\mu_{x_1}$. Denote by $\sd\gamma_{n-k}(x_1,\ldots,x_{n-k})$ the standard Gaussian density function\footnote{Explicitly $\sd\gamma_m(x_1,\ldots,x_m)=\prod_{i=1}^m\sd\gamma_1(x_i)$ where $\sd\gamma_1(r) = \frac{1}{\sqrt{2\pi}} e^{-r^2/2}$ is the density function for one-dimensional standard Gaussian.} in $\Rbb^{n-k}$.
    The density function of the random vector $\lx$ sampled from $\gamma_{n-k}$ conditioned on $x \in X$, is given ${\gamma_{n-k}(X)}^{-1} \cdot {\sd\gamma_{n-k}(x_1,\ldots,x_{n-k})}$ for $x\in X$ and zero outside. Thus, we have 
   \begin{align*}
        \sd\mu_{x_1}(r) &= \frac{\int_{X \cap \{x_1=r\}} \sd\gamma_{n-k}(x_1,\ldots,x_{n-k})}{\gamma_{n-k}(X)}\\
        & = \sd\gamma_1(r) \cdot \frac{\int_{X \cap \{x_1=r\}} \sd\gamma_{n-k-1}(x_2,\ldots,x_{n-k})}{\gamma_{n-k}(X)}  = \sd\gamma_1(r) \cdot \frac{\gamma_\rel(X \cap \{x_1=r\})}{\gamma_{n-k}(X)}.   
   \end{align*}

   Then, by concavity, the left hand side of \Cref{eqn:rectangle-measure-one-step} is exactly given by
   \begin{align*}
       \int_{\Rbb} \ln\pbra{\dfrac{e}{\gamma_\rel(X \cap \{x_1=r\})}} \sd\mu_{x_1}(r) &\le  \ln\pbra{\int_{\Rbb} \dfrac{e}{\gamma_\rel(X \cap \{x_1=r\})} \sd\mu_{x_1}(r)}  \\
       &=  \ln\pbra{\dfrac{e}{\gamma_{n-k}(X)} \int_{\Rbb}  \sd\gamma_1(r)} = \ln\pbra{\dfrac{e}{\gamma_{n-k}(X)}}.  
       \qedhere
   \end{align*}

\end{proof}

\section{Level-Two Fourier Growth}\label{sec:proof_of_level_two}

In this section, we prove \Cref{thm:boolean_bound_level_two} that $L_{1,2}(h)=O\pbra{d^{3/2}\log^3(n)}$. Similar to the proof of level-one bound \Cref{thm:boolean_bound_level_one}, we start with a $d$-round communication protocol $\tilde\Ccal$ over the Gaussian space as defined in \Cref{sec:fourier_via_martingale}.
Note that $\tilde\Ccal$ in turn comes from the original Boolean communication protocol $\Ccal$. Thus in the following we assume without loss of generality $d\le n$.

Given the discussion in \Cref{sec:fourier_weights_via_martingales}, to bound the second-level Fourier growth, one can attempt to bound the expected quadratic variation of the martingale that results from the protocol $\bar{\Ccal}$ directly, but similar to the case of level-one it is hard to leverage cancellations here to prove the bound we aim for. So, starting from $\tilde{\Ccal}$, we will define a communication protocol $\bar{\Ccal}$ that computes the same function as $\tilde\Ccal$, but satisfies some additional ``clean" property where it is easier to control the quadratic variation. This new protocol will differ from $\tilde\Ccal$ in two ways. Firstly, the protocol $\bar{\Ccal}$ will consist of additional ``cleanup steps'' where Alice and Bob reveal certain {\em quadratic forms} of their input. Secondly, the protocol $\bar\Ccal$ will send the real value of the quadratic form {\em with certain precision}. Note that this protocol will not involve sending real messages at all, instead, any potential real messages will be truncated to a few bits of precision and be sent as Boolean messages. 

We emphasize that the main difference in the protocol $\bar\Ccal$ from the corresponding level-one variant comes from the precision control, which is not needed there due to the fact that Gaussian distribution remains a (lower-dimensional) Gaussian under linear projections. For technical reasons we shall also need to analyze the martingale under a truncated Gaussian distribution, where all coordinates are  bounded in some large interval $[-T, T]$. This intuitively doesn't incur a noticeable difference on the distribution since it is highly unlikely that coordinates drawn from Gaussian distribution will be outside such intervals and recalling \Cref{rem:symmetric} and \Cref{prop:fwt-to-qv}, it still suffices to analyze the corresponding martingale under the truncated Gaussian distribution. 

We next define the notion of a $4$-wise clean protocol. 

\subsection[4-Wise Clean Protocols]{$4$-Wise Clean Protocols}

Consider an intermediate node in the protocol and let $X\subseteq \Rbb^n$ refer to the set of Alice's inputs reaching this node.
We denote by  $\Sbb^{n\times n-1}$ the set of all matrices in $\Rbb^{n\times n}$ with zero diagonal and unit norm (when viewed as $n^2$-dimensional vectors).
For a parameter $\lambda>0$, we say that the set $X$ is \emph{$4$-wise clean in a direction $a\in \Sbb^{n\times n-1}$} if 
\[
\E_{\lX\sim \gamma}\sbra{ \abra{\lX\tensor\lX - \sigma(X), a}^2 \mid\lX\in X} < \lambda,
\]
where we recall that $\comtwo(X)=\E_{\lX\sim \gamma}\sbra{\lX\tensor\lX\mid\lX\in X}$ is the level-two center of mass of $X$ under the Gaussian measure.
We say that the set $X$ is \emph{$4$-wise clean} if it is $4$-wise clean in \emph{every direction $a$}. 
Our new protocol will consist of the original protocol, interspersed by cleaning steps. Once Alice sends her bit as in the original protocol, she cleans $X$ by revealing $\abra{x\tensor x, a}$ with a few bits of precision while there exists direction $a \in \Sbb^{n\times n-1}$  such that $X$ not clean in direction $a$. Once $X$ becomes clean, Alice proceeds to the next round and Bob does an analogous cleanup. We now describe this formally.
 
\paragraph*{Communication with Finite Precision.}
Let positive integer $L$ be a precision parameter that we will use for truncation.
In our new communication protocol, we will send real numbers with precision $2^{-L}$.
This is formalized as the $\SendReal_L(z)$ function defined at $z\in \R$ as
$$
\SendReal_L(z)=\floorbra{z\cdot 2^{L}}/2^L.
$$

\paragraph*{Construct $\bar \Ccal$ from $\tilde \Ccal$.} 
As described before, $\bar\Ccal$ will consist of the original protocol along with extra steps where Alice or Bob reveal the (approximate) value of a quadratic form on their input. Consider an intermediate node of this new protocol at depth $t$. We always use the random variable $\supX{t}$ (resp., $\supY{t}$) to denote the set of inputs of Alice (resp., Bob) reaching the node. If Alice is revealing a quadratic form in this step, we use $\supa{t}$ to denote the matrix of the quadratic form revealed at this node, otherwise set $\supa{t}$ to be the all-zeroes matrix. We define $\supb{t}$ similarly for Bob. Throughout the protocol, we will always set $\supu{t}$ and $\supv{t}$ to denote $\comtwo(\supX{t})$ and $\comtwo(\supY{t})$ respectively.

Recall that $\lambda>0$ is the parameter for cleanup to be optimized later. Since we will now send real numbers (with certain precision) as bit-strings, their magnitudes should also be well controlled to guarantee bounded message length.
This is managed by a parameter $T>0$ and we will restrict the inputs to the parties in $\bar\Ccal$ to come from the box $[-T,T]^n$. Note that, by Gaussian concentration, $T=\Theta\pbra{\sqrt{\log(n)}}$ suffices. 
\begin{enumerate}
\item At the beginning, Alice receives an input $x\in[-T,T]^n$ and Bob receives an input $y\in[-T,T]^n$.
\item We initialize $t\gets0$, $\supX{0},\supY{0}\gets[-T,T]^n$, and $\supa{0},\supb{0}\gets0^{n\times n}$.
\item For each phase $i=1,2,\ldots,d$: suppose we are starting the cleanup for a node at depth $i$ in the original protocol $\tilde\Ccal$ and suppose we are at a node of depth $t$ in the new protocol $\bar\Ccal$. If it is Alice's turn to speak in $\tilde\Ccal$:
\begin{enumerate}
\item \textbf{Orthogonalization by revealing the correlation with Bob's center of mass.}\\
Alice begins by revealing the inner product of her input $x$ with Bob's current (signed) level-two center of mass $\Lambda\odot \supv{t}$. Since in the previous steps, she has already revealed the inner product with Bob's previous centers of mass, for technical reasons, we will only have Alice announce the inner product with the component of $\Lambda\odot \supv{t}$ that is orthogonal to the previous directions along which Alice announced the inner product. More formally, let $\supa{t+1}$ be the component of $ \Lambda\odot \supv{t}$ that is orthonormal to the span of the previous directions $\supa{\tau}$ for $\tau\le t$, i.e.,
$$\textstyle
\supa{t+1}=\unit\pbra{ \Lambda\odot \supv{t}-\sum_{\tau=1}^t\abra{ \Lambda \odot \supv{t},\supa{\tau}}\cdot\supa{\tau}}.
$$
Alice computes $\supcbar{t+1}\gets\SendReal_L\pbra{\abra{x\tensor x,\supa{t+1}}}$ and sends $\supcbar{t+1}$ to Bob. 
Set $\supb{t+1}\gets 0^{n\times n}$.
Increment $t$ by $1$ and go to step (b). 
\item \textbf{Original communication.}
Alice sends the bit $\supcbar{t+1}$ that she was supposed to send in $\tilde\Ccal$ based on previous messages and $x$. Set $\supa{t+1},\supb{t+1}\gets 0^{n\times n}$. 
Increment $t$ by 1 and go to step (c). 
\item \textbf{Cleanup steps.}
While there exists some direction $a\in\Sbb^{n\times n-1}$ orthogonal to previous directions, i.e., $\abra{a,\supa{\tau}}=0$ for all $\tau\le t$, and $\supX{t}$ is \emph{not $4$-wise clean} in direction $a$, Alice computes $\supcbar{t+1}\gets\SendReal_L\pbra{\abra{x\tensor x,a}}$ and sends $\supcbar{t+1}$ to Bob. 
Set $\supa{t+1}\gets a$ and $\supb{t+1}\gets0^{n\times n}$. Increment $t$ by 1. 
Repeat step (c) while $\supX{t}$ is not $4$-wise clean; otherwise, increment $i$ by 1 and go back to the for-loop in step 3 which starts a new phase.
\end{enumerate}
If it is Bob's turn to speak, we define everything similarly with the role of $x,\lA,\X,\U$ switched with $y,\lB,\Y,\V$.
\item Finally at the end of the protocol, the value $\bar\Ccal(x,y)$ is determined based on all the previous communication and the corresponding output it defines in $\tilde\Ccal$.
\end{enumerate}

\begin{remark}
Note that by construction, the non-zero matrices among $\supa{1},\supa{2}, \ldots $ form an orthonormal set when viewed as $n^2$-dimensional vectors (similarly for $\supb{1},\supb{2}, \ldots $) and moreover, their diagonals are zero. Lastly, $\supa{t}$ and $\supb{t}$ are known to both Alice and Bob as they are canonically determined by previous messages.
\end{remark}

We remark that the steps 3(a), 3(b), and 3(c) always occur in sequence for each party and we refer to such a sequence of steps as a \emph{phase} for that party. Note that there are at most $d$ phases. 
If a new phase starts at time $t$, then the current rectangle $\supX{t} \times \supY{t}$ is $4$-wise clean for both parties by construction. 

Now we formalize a few useful properties regarding the communication protocol $\bar\Ccal$. The first fact below follows since each $\supu{t}$ is an expectation of $\lx\tensor \lx$ over some distribution and $\lx\tensor \lx$ has zero diagonal.

\begin{fact}\label{fct:starting_point}
$\supu{0}=\supv{0}=0^{n\times n}$ and each $\supu{t},\supv{t}$ has zero diagonal.
\end{fact}

The following follows from tail bounds for the univariate standard normal distribution.

\begin{fact}\label{fct:gammastar}
Let $\gamma^*=\gamma(\supX{0})\cdot\gamma(\supY{0})$. Then $\gamma^*\ge1-O\pbra{n\cdot e^{-T^2/2}}$.
\end{fact}

The next fact says that when a node fixes a quadratic form with $2^{-L}$ precision, for any two inputs that reach this node, the quadratic forms differ by at most $2^{-L}$. 

\begin{fact}\label{fct:sendreal_error}
In step 3(a) and 3(c), any $x,x'\in \supX{t+1}$ satisfies $\abs{\abra{x\tensor x,\supa{t+1}}-\abra{x'\tensor x',\supa{t+1}}}<2^{-L}$.
Similarly any $y,y'\in \supY{t+1}$ satisfies $\abs{\abra{y\tensor y,\supb{t+1}}-\abra{y'\tensor y',\supb{t+1}}}<2^{-L}$.
\end{fact}

The next claim bounds the maximum attainable norms for Alice and Bob's level-two center of masses at any point in the protocol. This uses the fact that the inputs come from the truncated Gaussian distribution.

\begin{claim}\label{clm:frob_norm_ub}
$\frob{\supu{t}}=\frob{ \Lambda\odot \supu{t}}<nT$ and $\frob{\supv{t}}=\frob{ \Lambda\odot \supv{t}}<nT$ for all possible $t$ and $\supu{t},\supv{t}$ throughout the communication.
\end{claim}
\begin{proof}
Since $\Lambda$ is a matrix with zero diagonal and $\binpm$ entries off diagonal and $\supu{t}$ also has zero diagonal, $\frob{\supu{t}}=\frob{ \Lambda\odot \supu{t}}$.
In addition, since $\supX{t}\subseteq \supX{0}=[-T,T]^n$, we have
$$
\frob{\supu{t}}
\le\E_{\lx\sim\gamma}\sbra{\frob{\pbra{\lx\tensor\lx}}\mid\lx\in \supX{t}}
\le\sqrt{(n^2-n)\cdot T^2}<nT.
$$
A similar analysis works for $\supv{t}$.
\end{proof}

The next claim gives a bound on the length of any message in the protocol $\bar \Ccal$.

\begin{claim}\label{clm:short_messages}
For any $x\in \supX{0}$ and $y\in \supY{0}$, any message in $\bar\Ccal(x,y)$ consists of at most $L + \log(Tn)$ many bits.
\end{claim}
\begin{proof}
Assume without loss of generality it is Alice's turn to speak. On step 3(b) she sends one bits. On steps 3(a) and 3(c), she computes $\SendReal_L(\sabra{x\tensor x,a})$ for some  $a\in \Sbb^{n\times n-1}$ and send the result. Since
$$
\abs{\abra{x\tensor x,a}}\le\frob{x\tensor x}\cdot \frob{a}\le\sqrt{(n^2-n)\cdot T^2}<nT,
$$
and the message is a multiple of $2^{-L}$
that means $\SendReal_L$ yields a message with $L+ \log(nT)$ many bits.
\end{proof}

The last claim bounds the maximum depth of the new protocol $\bar \Ccal$.

\begin{claim}\label{clm:finite_steps}
Let $\ell$ be an arbitrary leaf of the protocol $\bar\Ccal$ and $D(\ell)$ be its depth.
Then $D(\ell)\le2n^2$.
Moreover, along this path there are at most $n^2-n$ many non-zero $\supa{t}$ and at most $n^2-n$ many non-zero $\supb{t}$ for $t\in\{1,\ldots,D(\ell)\}$.
\end{claim}
\begin{proof}
We count the number of communication steps separately:
\begin{itemize}
\item \textbf{Steps 3(a) and 3(b).} Steps 3(a) and 3(b) occur once in every phase, thus at most $d$ times.
\item \textbf{Step 3(c).} For Alice, each time she communicates at step 3(c), the direction $a\in\Sbb^{n\times n-1}$ is non-zero and orthogonal to all previous $\supa{t}$'s. Since the dimension of $\Sbb^{n\times n-1}$ is $n^2-n$, this happens at most $n^2-n$ times. Similar argument works for Bob.
\end{itemize}
Thus in total we have at most $2(n^2-n)+2d \le 2n^2$ steps.
\end{proof}

We will eventually show that, with suitable choice of $\lambda,T,L$, typically $D(\ell)$ is at most $d\cdot\polylog(n)$.

\subsection{Bounding the Expected Quadratic Variation}\label{sec:expected_quadratic_variation}

Consider the martingale process defined in \Cref{eqn:def-martingale} from a random walk on the protocol tree generated by $\bar\Ccal$ where the inputs $\lx, \ly$ are sampled from $\gamma_n$ conditioned on being in the bounded cube $[-T,T]^n$. Recall that \Cref{prop:vec-martingale} still holds (see \Cref{rem:martingale}).

Formally, at time $t$ the process is defined by
$$
\supZ{t}_2=\abra{\supu{t},\eta\odot\supv{t}},
$$
where we recall that $\supu{t}=\comtwo(\supX{t})$ and $\supv{t}=\comtwo(\supY{t)})$ and $\eta$ is a fixed sign matrix with a zero diagonal.
The martingale process stops once it hits a leaf of $\bar\Ccal$.
Let $\D$ denote the (stopping) time when this happens.
Note that $\E[\D]$ is exactly the expected depth of the protocol $\bar\Ccal$.

In light of \Cref{rem:symmetric} and \Cref{prop:fwt-to-qv}, to prove \Cref{thm:boolean_bound_level_two}, it suffices to prove the following.
\begin{lemma}\label{lem:qv-level-two}
$\E\sbra{\sum_{t=1}^{\D} \pbra{\Delta\supZ{t}_2}^2} = O\pbra{d^3\log^6(n)}.$
\end{lemma}

\Cref{lem:qv-level-two} is proved in three steps.
We first show that essentially the only change in the value of the martingale is the orthogonalization step 3(a).
The reason is the same as the level-one bound: Alice's messages sent in step 3(b) and 3(c) are always near-orthogonal to Bob's current level-two center of mass, thus they do not change the value of the martingale $\supZ{t}_2$ much.
Moreover, by level-two analog of \Cref{eqn:overview}, since Alice's current node was clean just before Alice sent the message in step 3(a), the expected change $\E\sbra{\pbra{\Delta\supZ{t+1}_2}^2}$ can be bounded in terms of the squared norm of the change that occurred in $\supu{t}$ (or $\supv{t}$) between the current round and the last round where Alice was in step 3(a). Similar argument works for Bob.

Formally, this is encapsulated by the next lemma for which we need some additional definitions. Let $(\supF{t})_t$ denote the natural filtration induced by the random walk on the generalized protocol tree with respect to which $\supZ{t}_2$ is a Doob martingale and also $\supu{t}, \supv{t}$ form vector-valued martingales (recall \Cref{prop:vec-martingale}). Note that $\supF{t}$ fixes all the rectangles encountered during times $0,\ldots, t$ and thus for $\tau \le t$, the random variables $\supu{\tau},\supv{\tau},\supZ{\tau}_2$ are determined, in particular, they are $\supF{t}$-measurable. Recalling that $\lambda$ is the cleanup parameter to be optimized later, we then have the following. Below we assume without any loss of generality that Alice speaks first and, in particular, we note that Alice speaks in step 3(a) for the first time at time zero when both Alice and Bob's center of masses are at zero: $\supu{0}=\supv{0}=0$. 

\begin{lemma}[Step Size]\label{lem:step_size_square_level_two}
    Let $0= \btau_1 < \btau_2 < \cdots \le \D$ be a sequence of stopping times with $\btau_m$ being the index of the round where Alice speaks in step 3(a) for the $m^\text{th}$ time or $\D$ if there is no such round. 
    Then, for any integer $m \ge 2$, 
	$$
	\E\sbra{\pbra{\Delta\supZ{\btau_m+1}_2}^2 \mid \supF{\btau_m}}  \le \lambda \cdot \vabs{\supv{\btau_m} - \supv{\btau_{m-1}}}^2+ 16n^7T^3 \cdot  2^{-L}.
	$$
	and moreover, for any $t \in \N$, we have that 
	$$
	\E\sbra{\pbra{\Delta\supZ{t+1}_2}^2\mid\supF{t}, \btau_{m-1} < t <\btau_{m}, \text{Alice speaks at time }t}\le 4 n^6T^2 \cdot 2^{-2L}
	$$
 A similar statement also holds if Bob speaks where $\V$ is replaced by $\U$ and the sequence $(\btau_m)$ is replaced by $(\btau'_m)$ where $\btau'_m$ is the index of the round where Bob speaks in step 3(a) for the $m^\text{th}$ time or $\D$ if there is no such round. 
\end{lemma}

We indeed see that, if $L=\Omega(\log(n))$ and $T=O(\sqrt{\log(n)})$, then $\poly(T,n)\cdot 2^{-L} = o(1)$, and steps~3(b) and~3(c) do not contribute much to the quadratic variation and only the steps 3(a) do.  Also, since the first time Alice and Bob speak, they start in step 3(a), we also note that $\supu{\btau_1}$ and $\supv{\btau'_1}$ are their initial centers of mass which are both zero.  

We shall prove the above lemma in \Cref{sec:step_size_leve_two} and continue with the bound on the quadratic variation here.
Using the bounds on the step sizes from \Cref{lem:step_size_square_level_two},
\begin{align*}
\E\sbra{\sum_{t=1}^{\D} \pbra{\Delta\supZ{t}_2}^2} 
&\le\lambda \cdot\E\sbra{\sum_{m\ge 2} \vabs{\V^{(\btau_m)}-\V^{(\btau_{m-1})}}^2+\vabs{\U^{(\btau'_m)}-\U^{(\btau'_{m-1})}}^2}+
16n^7T^3 \cdot  2^{-L}
\cdot\E[\D]\\
&\le\lambda \cdot\E\sbra{\sum_{m \ge 2} \vabs{\V^{(\btau_m)}-\V^{(\btau_{m-1})}}^2+\vabs{\U^{(\btau'_m)}-\U^{(\btau'_{m-1})}}^2}+16n^7T^3 \cdot  2^{-L}
\cdot2n^2.
\tag{by \Cref{clm:finite_steps}}
\end{align*}
On the other hand, using the orthogonality of vector-valued martingale differences from  \Cref{eqn:martingale-orthogonality-vec},
\begin{align*}
	\E\sbra{\sum_{m \ge 2} \vabs{\V^{(\btau_m)}-\V^{(\btau_{m-1})}}^2} = \E\sbra{\vabs{\V^{(\D)}}^2}.
\end{align*}
A similar statement holds for $(\supu{t})$ as well. Therefore, 
\begin{align}\label{eqn:qv-upper-bound-level-two}
\E\sbra{\sum_{t=1}^{\D} \pbra{\Delta\supZ{t}_2}^2} \le\lambda \cdot\pbra{\E\sbra{\frob{\U^{(\D)}}^2}+\E\sbra{\frob{\V^{(\D)}}^2}}+64n^9T^3 \cdot  2^{-L}.
\end{align}

Then in \Cref{sec:to_depth} we will apply level-two inequalities (see \Cref{thm:level_k_ineq}) to convert the bounding $\E\sbra{\frob{\U^{(\D)}}^2+\frob{\V^{(\D)}}^2}$ into bounding the second moment $\E[\D^2]$. This reduction is formalized as \Cref{lem:to_depth} below and its proof is similar to \cite[Claim 1]{GRT21}.

For each leaf $\ell$, let $\gamma(\ell)=\gamma(\supX{D(\ell)})\cdot\gamma(\supY{D(\ell)})$ be the Gaussian measure of the rectangle at $\ell$. 
Recall $\gamma^*=\gamma(\supX{0})\times\gamma(\supY{0})$.
\begin{lemma}\label{lem:to_depth}
$\E\sbra{\frob{\supu{\D}}^2+\frob{\supv{\D}}^2}\le O\pbra{\frac1{\gamma^*}+L^2\E[\D^2]}$.
\end{lemma}

Finally, in \Cref{sec:depth_tail_bounds}, we bound the second moment $\E[\D^2]$ for a suitable choice of parameters.
\begin{lemma}\label{lem:second_moment} It holds that
$\E[\D^2]=O(d^2)$ and $\gamma^*\ge\frac{3}{4}$ for $L=\Theta(\log(n))$, $T=\Theta(\sqrt{\log(n)})$, and $\lambda=\Theta(d\log^4(n))$. 
\end{lemma}

Given \Cref{lem:to_depth,lem:second_moment},the proof of \Cref{lem:qv-level-two} naturally follows.
\begin{proof}[Proof of \Cref{lem:qv-level-two}]
With the parameters chosen in \Cref{lem:second_moment}, we have
\begin{align*}
\E\sbra{\sum_{t=1}^{\D} \pbra{\Delta\supZ{t}_2}^2} 
&\le O(d\log^4(n))\cdot\pbra{\E\sbra{\frob{\U^{(\D)}}^2}+\E\sbra{\frob{\V^{(\D)}}^2}}+1
\tag{by \Cref{eqn:qv-upper-bound-level-two}}\\
&\le O(d\log^4(n))\cdot\pbra{1+\log^2(n)\cdot\E[\D^2]}+1
\tag{by \Cref{lem:to_depth}}\\
&\le O(d\log^4(n))\cdot\pbra{1+\log^2(n)\cdot d^2}+1
\tag{by \Cref{lem:second_moment}}\\
&=O(d^3\log^6(n)).
\tag*{\qedhere}
\end{align*}
\end{proof}

\begin{remark}
Note that our proof for level-two Fourier growth actually holds for a slightly more general setting, where Alice and Bob are allowed to send $O(L)=O(\log(n))$ bits during each original communication round.
This can be viewed as balancing the length of the messages in step 3(b) with step 3(a) and step 3(c).

Since one can always convert a $d$-round $1$-bit communication protocol into a $\frac{2d}{\log\log(n)}$-round $\log(n)$-bit communication protocol, we obtain a slightly better level-two Fourier growth bound of 
$$
O\pbra{\frac{d^{3/2}\log^3(n)}{\pbra{\log\log(n)}^{3/2}}}.
$$ 
The conversion is done by Alice (resp., Bob) enumerating the next $\log\log(n)/2$ bits from Bob (resp., Alice), and providing the corresponding $\log\log(n)/2$ bits responses for each possibility.

It is also possible to improve the $\log^3(n)$ factor to $\log^2(n)$ by varying the cleanup parameter $\lambda$ with depth.
For example, for depth in the interval $[4rd, 4(r+1)d]$, one could pick $\lambda_r = \Theta( d \cdot \log^2(n) \cdot r^2)$.
Since our focus is mostly on improving the polynomial dependence in $d$ where there is still room for improvement, we do not make an effort here to improve the polylog terms.
\end{remark}

\subsection[Bounds on Step Sizes]{Bounds on Step Sizes (Proof of \Cref{lem:step_size_square_level_two})}\label{sec:step_size_leve_two}

Let us abbreviate $\btau = \btau_m$ and note that at time $\btau$ a new phase starts for Alice. 
By construction, this means that the current rectangle $\supX{\btau} \times \supY{\btau}$ determined by $\supF{\btau}$ is $4$-wise clean with parameter $\lambda$, and since Alice is in step 3(a) at the start of a new phase, $\supa{\btau+1}$ is chosen to be the (normalized) component of $\Lambda\odot \V^{(\btau)}$ that is orthogonal to previous directions $\supa{1}, \ldots, \supa{\btau}$. 

For each $r=1,\ldots,\btau+1$, let $\balpha^{(r)}:= \abra{\Lambda\odot \V^{(\btau)},\la^{(r)}}$ be the length of $\Lambda\odot \V^{(\btau)}$ along direction $\la^{(r)}$.
Each $\balpha^{(r)}$ is $\Fcal^{(\btau)}$-measurable (i.e., it is determined by $\Fcal^{(\btau)}$) and $\eta\odot\supv{\btau}=\sum_{r\le\btau+1}\balpha^{(r)}\cdot\supa{r}$. 
In this case, we have
\begin{align}\label{eq:step_size_level_two}
\E\sbra{\pbra{\Delta\lZ^{(\btau+1)}_2}^2\mid \Fcal^{(\btau)}}&=\E\sbra{\abra{\U^{(\btau+1)}-\U^{(\btau)}, \Lambda\odot \V^{(\btau)}}^2\mid \Fcal^{(\btau)}}\notag\\
&=\E\sbra{\pbra{\sum_{r=1}^{\btau+1}\balpha^{(r)}\cdot\abra{\supu{\btau+1}-\supu{\btau},\supa{r}}}^2\mid \Fcal^{(\btau)}}.
\end{align}

Similar to the level-one proof, the components of $\supu{\btau+1}$ and $\supu{\btau}$ are roughly the same along any of the previous directions $\supa{1},\ldots,\supa{\btau}$ and so they almost cancel out and the major quantity is in the direction $\supa{\btau+1}$.
This follows since, in all the previous steps $r\le\btau$, Alice has already fixed $\abra{x\tensor x,\supa{r}}$ with precision $2^{-L}$.
This implies that for any $\supX{\btau}$ and $\supX{\btau+1}$ that are determined by $\supF{\btau+1}$, the inner product with all the previous $\supa{1},\ldots,\supa{\btau}$ is fixed with precision $2^{-L}$ over the choice of $x$.
Formally, by \Cref{fct:sendreal_error}, we have that for any $x\in\supX{\btau}$ and $x'\in\supX{\btau+1}$, it holds that $\abs{\abra{x\tensor x,\supa{r}}-\abra{x'\tensor x',\supa{r}}}\le2^{-L}$ for all $r\le\btau$.
In particular, since $\supu{\btau}=\comtwo(\supX{\btau})$ and $\supu{\btau+1}=\comtwo(\supX{\btau+1})$ are the corresponding centers of mass, we have that
\begin{equation}\label{eq:step_size_level_two_error_term}
\abs{\abra{\supu{\btau+1}-\supu{\btau},\supa{r}}}\le2^{-L}
\quad\text{for all $r\le\btau$.}
\end{equation}
On the other hand, 
since $\supX{\btau+1}\subseteq\supX{\btau}\subseteq\supX{0}=[-T,T]^n$ and $\supa{\btau+1}$ is a unit direction, we have
\begin{equation}\label{eq:step_size_level_two_frob_bound}
\abs{\abra{\supu{\btau+1}-\supu{\btau},\supa{\btau+1}}}
\le\vabs{\supu{\btau+1}-\supu{\btau}}
\le2nT.
\end{equation}
Similarly, 
noting that $\eta$ is a sign matrix, we can bound
\begin{equation}\label{eq:step_size_level_two_beta}
\abs{\balpha^{(r)}}
=\abs{\abra{\eta\odot\supv{\btau},\supa{r}}}
\le\vabs{\eta\odot\supv{\btau}}
\le\vabs{\supv{\btau}}
\le nT
\quad\text{for all $r\le\btau+1$.}
\end{equation}
Expanding the square in \Cref{eq:step_size_level_two} and plugging these estimates to each one of the $(\btau+1)^2$ terms gives
\begin{align}
\E\sbra{\pbra{\Delta\lZ^{(\btau+1)}_2}^2\mid \Fcal^{(\btau)}}
&\le\E\sbra{\pbra{\balpha^{(\btau+1)}}^2\abra{\supu{\btau+1}-\supu{\btau},\supa{\btau+1}}^2 
+ ((\btau+1)^2-1)\cdot  \tfrac{2(nT)^3}{2^{L}}
\mid \Fcal^{(\btau)}}\notag
\\
&\le\pbra{\balpha^{(\btau+1)}}^2\E\sbra{\abra{\supu{\btau+1}-\supu{\btau},\supa{\btau+1}}^2\mid\supF{\btau}}+12n^7T^3 \cdot  2^{-L},
\label{eq:step_size_level_two_no_error_term}
\end{align}
where the second line follows from \Cref{clm:finite_steps}.

We now bound the term outside the expectation by the change in the center of mass $\supv{\cdot}$ and the term inside the expectation by the fact that the set is $4$-wise clean.

\paragraph*{Term Outside the Expectation.}
Recall that $\supa{\btau+1}$ is chosen to be the (normalized) component of $\eta\odot\supv{\btau}$ that is orthogonal to the span of $\supa{1},\ldots,\supa{\btau}$.
Since $\eta\odot\supv{\btau_m-1}$ is in the span of $\supa{1},\ldots,\supa{\btau_{m-1}+1}$ and $\btau_{m-1}+1\le\btau=\btau_m$, it is orthogonal to $\supa{\btau+1}$. Hence
$$
\balpha^{(\btau+1)}=\abra{\eta\odot\supv{\btau},\supa{\btau+1}}=\abra{\eta\odot\pbra{\supv{\btau}-\supv{\btau_{m-1}}},\supa{\btau+1}}.
$$
Since $\supa{\btau+1}$ is a unit direction and $\eta$ is a sign matrix, this implies that
\begin{equation}\label{eq:level_two_term_outside}
\pbra{\balpha^{(\btau+1)}}^2\le\vabs{\supv{\btau}-\supv{\btau_{m-1}}}^2.
\end{equation}

\paragraph*{Term Inside the Expectation.}
Recall that Alice is in step 3(a), she sends $\abra{x\tensor x,\supa{\btau+1}}$ with precision $2^{-L}$ at time $\btau$, and thus the same inner product with $\supa{\btau+1}$ is fixed with precision $2^{-L}$ for every point in $\supX{\btau+1}$ determined by $\supF{\btau+1}$.
Thus
\begin{align}
\abra{\supu{\btau+1},\supa{\btau+1}}^2
&=\pbra{\E_{\lx\sim \gamma}\sbra{\abra{\lx\tensor\lx,\la^{(\btau+1)}}\mid \lx\in \X^{(\btau+1)}}}^2\notag\\
&=\pbra{\abra{x\tensor x,\la^{(\btau+1)}}+\E_{\lx\sim \gamma}\sbra{\eps_{\lx}\mid \lx\in \X^{(\btau+1)}}}^2
\tag{$|\eps_{\lx}|\le2^{-L}$ is the truncation error by \Cref{fct:sendreal_error}}\\
&\le\abra{x\tensor x,\la^{(\btau+1)}}^2+2^{-2L}+2^{1-L}\cdot\abs{\abra{x\tensor x,\la^{(\btau+1)}}}\notag\\
&\le\abra{x\tensor x,\la^{(\btau+1)}}^2+nT\cdot2^{2-L},
\label{eq:level_two_term_inside}
\end{align}
where the last line follows from $\abs{\abra{x\tensor x,\la^{(\btau+1)}}}\le\vabs{x\tensor x}$ and $x\in\supX{0}=[-T,T]^n$.

\paragraph*{Final Bound.}

Since $(\supu{r})_r$ is a matrix-valued martingale and thus $\E\sbra{\supu{\btau+1}\mid\supF{\btau}}=\supu{\btau}$, we have
$$
\E\sbra{\abra{\supu{\btau+1}-\supu{\btau},\supa{\btau+1}}^2\mid\supF{\btau}}
=\E\sbra{\abra{\supu{\btau+1},\supa{\btau+1}}^2-\abra{\supu{\btau},\supa{\btau+1}}^2\mid\supF{\btau}}
$$
Then by \Cref{eq:level_two_term_inside}, we upper bound the right hand side by
\begin{align*}
nT\cdot2^{2-L}+\E_{\lx\sim\gamma}\sbra{\abra{\lx\tensor\lx,\supa{\btau+1}}^2-\abra{\supu{\btau},\supa{\btau+1}}^2\mid\supF{\btau}}.
\end{align*}
Since $\supX{\btau}$ is $4$-wise clean with parameter $\lambda$, it can be bounded by $nT\cdot2^{2-L}+\lambda$:
\begin{equation}\label{eq:level_two_exp_inside}
\E\sbra{\abra{\supu{\btau+1}-\supu{\btau},\supa{\btau+1}}^2\mid\supF{\btau}}
\le nT\cdot2^{2-L}+\lambda
\end{equation}
Putting everything together, we have
\begin{align*}
\E\sbra{\pbra{\Delta\lZ^{(\btau+1)}_2}^2\mid \Fcal^{(\btau)}}
&\le\pbra{\balpha^{(\btau+1)}}^2\E\sbra{\abra{\supu{\btau+1}-\supu{\btau},\supa{\btau+1}}^2\mid\supF{\btau}}+
12n^7T^3 \cdot  2^{-L}
\tag{by \Cref{eq:step_size_level_two_no_error_term}}\\
&\le\pbra{\balpha^{(\btau+1)}}^2\cdot\pbra{nT\cdot2^{2-L}+\lambda}+12n^7T^3 \cdot  2^{-L}
\tag{by \Cref{eq:level_two_exp_inside}}\\
&\le\lambda\cdot\pbra{\balpha^{(\btau+1)}}^2+n^3T^3 \cdot 2^{2-L}+12n^7T^3 \cdot  2^{-L}
\tag{by \Cref{eq:step_size_level_two_beta}}\\
&\le\lambda\cdot\vabs{\supv{\btau}-\supv{\btau_{m-1}}}^2+n^3T^3 \cdot 2^{2-L}+12n^7T^3 \cdot  2^{-L}
\tag{by \Cref{eq:level_two_term_outside}}\\
&\le\lambda\cdot\vabs{\supv{\btau}-\supv{\btau_{m-1}}}^2+16n^7T^3 \cdot  2^{-L}.
\end{align*}
This completes the proof of the first statement in the lemma.

For the moreover part, let us condition on the event $\btau_{m-1}<t<\btau_m$ where Alice speaks at time $t$. 
Note that such $t$ must all lie in the same phase of the protocol where Alice is the only one speaking.
So, Bob's center of mass does not change from the time $\btau_{m-1}$ till $t$, i.e., $\supv{t+1}=\supv{\btau_{m-1}}$.
Thus we have 
\begin{equation}\label{eq:level_two_moreover}
\Delta\supZ{t+1}_2=\abra{\supu{t+1}-\supu{t},\eta\odot\supv{\btau_{m-1}}}.
\end{equation}
Analogous to \Cref{eq:step_size_level_two_error_term}, the component of Alice's center of mass along the previous directions are fixed with precision $2^{-L}$.
Thus by \Cref{fct:sendreal_error}, 
\begin{equation}\label{eq:level_two_moreover_error}
\abs{\abra{\supu{t+1}-\supu{t},\supa{r}}}\le2^{-L}
\quad\text{for all $r\le t$.}
\end{equation}
Furthermore, by construction, $\eta\odot\supv{\btau_{m-1}}$ lies in the space spanned by $\supa{1},\ldots,\supa{\btau_{m-1}+1}$.
Note that $\btau_{m-1}+1\le t$.
Similar to the previous analysis, for each $r=1,\ldots,t$, let $\balpha^{(r)}:=\abra{\eta\odot\supv{t},\supa{r}}$ be the length of $\eta\odot\supv{t}$ along direction $\supa{r}$.
Then \Cref{eq:step_size_level_two_beta} also holds here.
Therefore
\begin{align*}
\abs{\Delta\supZ{t+1}_2}
&=\abs{\sum_{r=1}^t\balpha^{(r)}\cdot\abra{\supu{t+1}-\supu{t},\supa{r}}}
\tag{by \Cref{eq:level_two_moreover}}\\
&\le\sum_{r=1}^t\abs{\balpha^{(r)}}\cdot\abs{\abra{\supu{t+1}-\supu{t},\supa{r}}}
\le\sum_{r=1}^tnT\cdot2^{-L}
\tag{by \Cref{eq:step_size_level_two_beta} and \Cref{eq:level_two_moreover_error}}\\
&\le 2n^3T\cdot2^{-L}.
\tag{by \Cref{clm:finite_steps}}
\end{align*}

\subsection[Conversion to Second Moment Bounds of the Depth]{Conversion to Second Moment Bounds of the Depth (Proof of \Cref{lem:to_depth})}\label{sec:to_depth}

Recall $\gamma^*=\gamma(\supX{0})\times\gamma(\supY{0})$ and $\gamma(\ell)=\gamma(\supX{D(\ell}))\cdot\gamma(\supY{D(\ell)})$ for each leaf $\ell$.
The goal of this subsection is to prove \Cref{lem:to_depth}.

We first note the following basic fact.

\begin{fact}\label{fct:path_probability}
$\sum_\ell\gamma(\ell)=\gamma^*$ and
$$
\Pr_{\lx\sim \supX{0},\ly\sim \supY{0}}\sbra{\bar\Ccal(\lx,\ly)\text{ reaches leaf }\ell}=\gamma(\ell)/\gamma^*.
$$
\end{fact}

Now we apply \Cref{thm:level_k_ineq} with $k=2$ to relate the LHS of \Cref{lem:to_depth} with an entropy-type bound.

\begin{lemma}\label{lem:frob_to_path_prob}
$\E\sbra{\frob{\supu{\D}}^2+\frob{\supv{\D}}^2}\le \frac{4e^2}{\gamma^*}\sum_\ell\gamma(\ell)\cdot\ln^2\pbra{\frac{e}{\gamma(\ell)}}$.
\end{lemma}
\begin{proof}
Let $\ell$ be a fixed leaf and $D=D(\ell)$ be its depth.
Note that this also fixes the rectangle $X^{(D)}\times Y^{(D)}$ and thus the centers of mass $u^{(D)},v^{(D)}$.
Define the indicator function $\indicator_\ell\colon\Rbb^{2n}\to\bin$ by
$$
\indicator_\ell(x,y)=\begin{cases}
1 & (x,y)\in X^{(D)}\times Y^{(D)},\\
0 & \text{otherwise.}
\end{cases}
$$
Then we have
\begin{align*}
&\phantom{\le}\frob{u^{(D)}}^2+\frob{v^{(D)}}^2\\
&=\frob{\E_{\lx\sim\gamma}\sbra{\lx\tensor \lx\mid \lx\in X^{(D)}}}^2+\frob{\E_{\ly\sim\gamma}\sbra{\ly\tensor \ly\mid \ly\in Y^{(D)}}}^2\\
&=\sum_{\substack{i,j=1\\i\neq j}}^n\pbra{\E_{\lx\sim\gamma}\sbra{\lx_i\lx_j\mid \lx\in X^{(D)}}}^2+\sum_{\substack{i,j=1\\i\neq j}}^n\pbra{\E_{\ly\sim\gamma}\sbra{\ly_i\ly_j\mid \ly\in Y^{(D)}}}^2\\
&=\sum_{\substack{i,j=1\\i\neq j}}^n\pbra{\E_{\lx,\ly\sim\gamma}\sbra{\lx_i\lx_j\mid(\lx,\ly)\in X^{(D)}\times Y^{(D)}}}^2+\sum_{\substack{i,j=1\\i\neq j}}^n\pbra{\E_{\lx,\ly\sim\gamma}\sbra{\ly_i\ly_j\mid(\lx,\ly)\in X^{(D)}\times Y^{(D)}}}^2\\
&=\frac2{\gamma(\ell)^2}\pbra{\sum_{S\in\binom{[n]}2}\pbra{\E_{\lx\sim\gamma,\ly\sim\gamma}\sbra{\indicator_\ell(\lx,\ly)\lx_S}}^2+\sum_{S\in\binom{[n]}2}\pbra{\E_{\lx\sim\gamma,\ly\sim\gamma}\sbra{\indicator_\ell(\lx,\ly)\ly_S}}^2}\\
&\le\frac2{\gamma(\ell)^2}\sum_{S\in\binom{[2n]}2}\pbra{\E_{\lw\sim\gamma_n\times\gamma_n}\sbra{\indicator_\ell(\lw)\lw_S}}^2\\
&\le\frac2{\gamma(\ell)^2}\cdot 2e^2\gamma(\ell)^2\cdot\ln^2\pbra{\frac{e}{\gamma(\ell)}}
\tag{by \Cref{thm:level_k_ineq}}\\
&=4e^2\cdot\ln^2\pbra{\frac{e}{\gamma(\ell)}}.
\end{align*}
Therefore taking expectation over a random $\ell$, by \Cref{fct:path_probability}, we have
\begin{equation*}
\E\sbra{\frob{\supu{\D}}^2+\frob{\supv{\D}}^2}
\le 4e^2\cdot\E_{\bell}\sbra{\ln^2\pbra{\frac{e}{\gamma(\bell)}}}
=\frac{4e^2}{\gamma^*}\sum_\ell\gamma(\ell)\cdot\ln^2\pbra{\frac{e}{\gamma(\ell)}}.
\tag*{\qedhere}
\end{equation*}
\end{proof}

Now in the next lemma, we bound the right hand side of \Cref{lem:frob_to_path_prob} in terms of the second moment of the depth, which immediately proves \Cref{lem:to_depth}.
\begin{lemma}\label{lem:path_prob_to_depth}
Assume that $Tn \le  2^L$. Then,
$\sum_\ell\gamma(\ell)\cdot\ln^2\pbra{e/{\gamma(\ell)}}\le O(1+\gamma^*\cdot L^2\E[\D^2])$.
\end{lemma}
\begin{proof}
By \Cref{clm:short_messages}, and the assumption $Tn \le 2^{L}$ each message is of length at most $L+\log(Tn)\le 2L$.
We divide $\ell$ into two cases based on $\gamma(\ell)$:
\begin{align*}
&\sum_{\ell:\gamma(\ell)<2^{-3L\cdot D(\ell)}}\gamma(\ell)\cdot\ln^2\pbra{\frac{e}{\gamma(\ell)}}\\
&\le\sum_{\ell:\gamma(\ell)<2^{-3L\cdot D(\ell)}}2^{-3L\cdot D(\ell)}\cdot\ln^2\pbra{e \cdot 2^{3L\cdot D(\ell)}}
\tag{$x\ln^2(e/x)$ is increasing when $0\le x\le 0.2$}\\
&\le\sum_{t=1}^{\infty}2^{-3L\cdot t}\cdot 2(9L^2t^2+1)\cdot\abs{\cbra{\ell:D(\ell)=t}} 
\tag{since $\ln^2(ab) \le 2\ln^2(a) + 2\ln^2(b)$}\\
&\le\sum_{t=1}^{\infty}2^{-3L\cdot t}\cdot 2(9L^2t^2+1)\cdot2^{(2L)\cdot t}
\tag{each message is of length $\le 2L$}\\
&\le\sum_{t=1}^{\infty}2(9L^2t^2+1)\cdot2^{-Lt}=O(1)
\tag{since $L\ge2$}
\end{align*}
and
\begin{align*}
\sum_{\ell:\gamma(\ell)\ge2^{-3L\cdot D(\ell)}}\gamma(\ell)\cdot\ln^2\pbra{\frac{e}{\gamma(\ell)}}
&\le\sum_{\ell:\gamma(\ell)\ge2^{-3L\cdot D(\ell)}}\gamma(\ell)\cdot\ln^2\pbra{e \cdot 2^{3L\cdot D(\ell)}}\\
&\le 2 \cdot 9L^2 \sum_\ell\gamma(\ell) D(\ell)^2 + 2\sum_\ell\gamma(\ell)\\
&=18L^2\gamma^*\cdot\E_{\bell}\sbra{D(\bell)^2} + 2\\
&=18L^2\gamma^*\cdot\E\sbra{\D^2} + 2.
\end{align*}
Adding up the two estimates above gives the desired bound.
\end{proof}

\subsection[Second Moment Bounds for the Depth]{Second Moment Bounds for the Depth (Proof of \Cref{lem:second_moment})}\label{sec:depth_tail_bounds}

The final ingredient is an estimate for the second moment $\E[\D^2]$.
This subsection is devoted to this goal and proving \Cref{lem:second_moment}.

For messages $\ell'=(\supcbar{1},\ldots,\supcbar{t})$, we define $\gamma(\ell')=\gamma(\supX{t})\cdot\gamma(\supY{t})$ where $\supX{t},\supY{t}$ is defined by the protocol using the messages $\ell'$.
Note that this definition is consistent with $\gamma(\ell)$ from \Cref{sec:to_depth} for a leaf $\ell$.

\begin{lemma}\label{lem:depth_tail_bound}
There exists a universal constant $\alpha>0$ such that the following holds.
Let $0\le d_1<d_2$ be two arbitrary integers with $d_2-d_1\ge2d+1$.
Let $\ell^*=(\supcbar{1},\ldots,\supcbar{d_1})$ be arbitrary messages of the first $d_1$ communication steps.
Assume $2^L\ge8n^4T^2$. Then
$$
\Pr\sbra{\D\ge d_2\mid\ell^*}\le\frac{\alpha\cdot d_2^2L^2}{\lambda\cdot(d_2-d_1-2d)}+\frac14\cdot\frac{2^{-3L\cdot d_1}}{\gamma(\ell^*)}.
$$
\end{lemma}
\begin{proof}
Let $\lx,\ly$ be sampled from $\gamma$ conditioned on $\lx\in\supX{0},\ly\in\supY{0}$.
Let $\bell$ be its corresponding leaf in $\bar\Ccal$ and $\D$ be the depth of $\bell$.
By \Cref{clm:finite_steps}, $\bell$ always has finite depth.
We extend $\supa{t}=\supb{t}=0^{n\times n}$ and $\supX{t}=\supX{\D},\supY{t}=\supY{\D}$ for all $t>\D$.
Then define
$$
\K(\lx,\ly)=\sum_{t=d_1+1}^{d_2}\pbra{\abra{\lx\tensor \lx,\supa{t}}^2+\abra{\ly\tensor\ly,\supb{t}}^2}
\quad\text{and}\quad
K=\E_{\lx,\ly\sim\gamma}\sbra{\K(\lx,\ly)\mid\ell^*},
$$
where $\supa{\cdot}$'s and $\supb{\cdot}$'s depend only on $\bell$.\footnote{Note that $\bell$ specifies all the communication messages, which allows us to simulate the protocol and obtain each $\supa{\cdot}$ and $\supb{\cdot}$.}
Equivalently, we can write $K$ as
$$
K=\E_{\lx,\ly\sim\gamma}\sbra{\K(\lx,\ly)\mid(\lx,\ly)\in X^{(d_1)}\times Y^{(d_1)}},
$$
where $X^{(d_1)}$ and $Y^{(d_1)}$ are fixed due to $\ell^*$.

Observe that for any fixed $t\ge d_1$, $\supX{t}\times \supY{t}$ induced by different $\bell$, conditioned on $\ell^*$, is a disjoint partition of $X^{(d_1)}\times Y^{(d_1)}$. 
Therefore sampling $\lx,\ly\sim\gamma$ conditioned on $(\lx,\ly)\in X^{(d_1)}\times Y^{(d_1)}$ is equivalent to 
\begin{itemize}
\item first sample random messages $\bell'=(\supcbar{d_1+1},\ldots,\supcbar{t})$ conditioned on $\ell^*$,
\item then sample $\lx,\ly\sim\gamma$ conditioned on $(\lx,\ly)\in \supX{t}\times \supY{t}$ given $\bell'$.
\end{itemize}
Note that we can further expand $\bell'$ to a leaf $\bell$ as a full communication path, and obtain the following equivalent sampling process:
\begin{itemize}
\item Sample a random leaf $\bell$ conditioned on $\ell^*$.
\item Sample $\lx,\ly\sim\gamma$ conditioned on $(\lx,\ly)\in \supX{t}\times \supY{t}$ defined by the first $t$ messages of $\bell$.
\end{itemize}
As a result, we have
\begin{align*}
K
&=\sum_{t=d_1+1}^{d_2}\E_{\bell}\sbra{\E_{\lx,\ly\sim\gamma}\sbra{\abra{\lx\tensor \lx,\supa{t}}^2+\abra{\ly\tensor \ly,\supb{t}}^2\mid(\lx,\ly)\in \supX{t}\times \supY{t}}\mid\ell^*}\\
&=\E_{\bell}\sbra{\sum_{t=d_1+1}^{d_2}\E_{\lx\sim\gamma}\sbra{\abra{\lx\tensor \lx,\supa{t}}^2\mid \lx\in \supX{t}}+\E_{\ly\sim\gamma}\sbra{\abra{\ly\tensor \ly,\supb{t}}^2\mid \ly\in \supY{t}}\mid\ell^*}.
\end{align*}
Observe that there are at most $2d$ many step 3(a) and 3(b) in $\bell$.
This means, if $\D\ge d_2$, then from the $(d_1+1)$-th to the $d_2$-th communication steps, there are at least $d_2-d_1-2d$ cleanup steps (i.e., step 3(c)), each of which contributes at least $\lambda$ to $K$.
Thus we can lower bound $K$ by
\begin{equation}\label{eq:lem:depth_tail_bound_1}
K\ge \lambda\cdot(d_2-d_1-2d)\cdot\Pr\sbra{\D\ge d_2\mid\ell^*}.
\end{equation}

On the other hand by \Cref{clm:finite_steps}, there are at most $n^2$ non-zero $\supa{\cdot}$'s and at most $n^2$ non-zero $\supb{\cdot}$'s in each communication path.
Thus
\begin{equation}\label{eq:lem:depth_tail_bound_2}
\K(\lx,\ly)\le n^2\cdot\pbra{\max_{x\in \supX{0}}\frob{x\tensor x}^2+\max_{y\in \supY{0}}\frob{y\tensor y}^2}<2n^4T^2.
\end{equation}

We now obtain another upper bound using \Cref{thm:quadratic_concentration}.
Let $\bar\bell=(\supcbar{1},\ldots,\supcbar{d_2})$ extend $\ell^*$ for the next $d_2-d_1$ messages.\footnote{If $\bar\bell$ becomes a leaf before $d_2$, then we can simply pad dummy messages to it.}
Then $K=\E_{\bar\bell}\sbra{\K(\bar\bell)\mid\ell^*}$ where $
\K(\bar\ell):=\E_{\lx,\ly\sim\gamma}\sbra{\K(\lx,\ly)\mid\bar\ell}$.
Note that $\bar\ell$ fixes $a^{(\cdot)}$'s and $b^{(\cdot)}$'s in $\K(\lx,\ly)$.
Therefore we use $\K_{\bar\ell}(\lx,\ly)$ to denote $\K(\lx,\ly)$ with the directions $a^{(\cdot)}$'s and $b^{(\cdot)}$'s fixed by $\bar\ell$.
We now continue the bound on $\K(\bar\ell)$:
\begin{align}
\K(\bar\ell)
&\le\sum_{t=0}^{\infty}\Pr_{\lx,\ly\sim\gamma}\sbra{\K_{\bar\ell}(\lx,\ly)\ge t\mid\bar\ell}
=\sum_{t=0}^{\infty}\frac{\Pr_{\lx,\ly\sim\gamma}\sbra{\K_{\bar\ell}(\lx,\ly)\ge t,\bar\ell}}{\Pr_{\lx,\ly\sim\gamma}\sbra{\bar\ell}}
\notag\\
&=\sum_{t=0}^{\infty}\min\cbra{1,\frac{\Pr_{\lx,\ly\sim\gamma}\sbra{\K_{\bar\ell}(\lx,\ly)\ge t,\bar\ell}}{\gamma(\bar\ell)}}
\tag{by the definition of $\gamma(\cdot)$}\\
&\le\sum_{t=0}^{\infty}\min\cbra{1,\frac{\Pr_{\lx,\ly\sim\gamma}\sbra{\K_{\bar\ell}(\lx,\ly)\ge t}}{\gamma(\bar\ell)}}.
\label{eq:lem:depth_tail_bound_3}
\end{align}

We now analyze $\Pr_{\lx,\ly\sim\gamma}\sbra{\K_{\bar\ell}(\lx,\ly)\ge t}$ using \Cref{thm:quadratic_concentration}.
Since $a^{(t)},b^{(t)}$ cannot be non-zero simultaneously, we rearrange the matrices and assume $a^{(d_1+1)},\ldots,a^{(d')},b^{(d'+1)},\ldots,b^{(d'')}$ are the only non-zero matrices where $d''\le d_2$.
Then
$$
\K_{\bar\ell}(\lx,\ly)=\sum_{t=d_1+1}^{d'}\abra{\lx\tensor \lx,a^{(t)}}^2+\sum_{t=d'+1}^{d''}\abra{\ly\tensor \ly,b^{(t)}}^2.
$$
Note that $a$'s (resp., $b$'s) satisfy the condition in \Cref{thm:quadratic_concentration}. 
Let $1/\kappa$ be the constant\footnote{In particular $\kappa=56448$ suffices from our proof in \Cref{app:thm:quadratic_concentration}.} in $\Omega$ in \Cref{thm:quadratic_concentration}.
Hence
\begin{align*}
\Pr\sbra{\K_{\bar\ell}(\lx,\ly)\ge t}
&\le\Pr\sbra{\sum_{t=d_1+1}^{d'}\abra{\lx\tensor \lx,a^{(t)}}^2\ge t/2}+\Pr\sbra{\sum_{t=d'+1}^{d''}\abra{\ly\tensor \ly,b^{(t)}}^2\ge t/2}\\
&\le2\exp\cbra{-\frac1\kappa\cdot\frac{t/2}{d'-d_1+\sqrt{t/2}}}+2\exp\cbra{-\frac1\kappa\cdot\frac{t/2}{d''-d'+\sqrt{t/2}}}
\tag{by \Cref{thm:quadratic_concentration} and assuming $t\ge196\cdot\max\cbra{d'-d_1,d''-d'}$}\\
&\le4\exp\cbra{-\frac1\kappa\cdot\frac{t/2}{d_2-d_1+\sqrt{t/2}}}.
\tag{since $d_1\le d'\le d''\le d_2$}
\end{align*}
Thus for any $t\ge196\cdot(d_2-d_1)\ge196\cdot\max\cbra{d'-d_1,d''-d'}$, we have
\begin{equation}\label{eq:lem:depth_tail_bound_5}
\Pr\sbra{\K_{\bar\ell}(\lx,\ly)\ge t}\le4\exp\cbra{-\frac1\kappa\cdot\frac{t/2}{d_2-d_1+\sqrt{t/2}}}.
\end{equation}

For $\gamma(\bar\ell)\ge2^{-3L\cdot d_2}$, we plug \Cref{eq:lem:depth_tail_bound_5} into \Cref{eq:lem:depth_tail_bound_3} and obtain
\begin{align}
\K(\bar\ell)
&\le\sum_{t=0}^{196\cdot(d_2-d_1)^2}1+\sum_{t>196\cdot(d_2-d_1)^2}\min\cbra{1,2^{3L\cdot d_2+1}\cdot\exp\cbra{-\frac1\kappa\cdot\frac{t/2}{d_2-d_1+\sqrt{t/2}}}}
\tag{by \Cref{eq:lem:depth_tail_bound_5}}\\
&\le196\cdot(d_2-d_1)^2+1+\sum_{t\ge196\cdot(d_2-d_1)^2}\min\cbra{1,2^{3L\cdot d_2+1}\cdot e^{-\frac1\kappa\cdot\frac{t/2}{2\sqrt{t/2}}}}
\notag\\
&\le197\cdot d_2^2+\sum_{t\ge1}\min\cbra{1,2^{3L\cdot d_2+1}\cdot e^{-\frac{\sqrt{t/2}}{2\kappa}}}
\notag\\
&\le\alpha\cdot d_2^2L^2,
\label{eq:lem:depth_tail_bound_6}
\end{align}
where $\alpha$ is another universal constant.
Now we have
$$
K
=\E_{\bar\bell}\sbra{\K(\bar\bell)\mid\ell^*}
=\sum_{\bar\ell}\frac{\gamma(\bar\ell)}{\gamma(\ell^*)}\cdot \K(\bar\ell)
=\sum_{\bar\ell:\gamma(\bar\ell)<2^{-3L\cdot d_2}}\frac{\gamma(\bar\ell)}{\gamma(\ell^*)}\cdot \K(\bar\ell)+\sum_{\bar\ell:\gamma(\bar\ell)\ge2^{-3L\cdot d_2}}\frac{\gamma(\bar\ell)}{\gamma(\ell^*)}\cdot \K(\bar\ell),
$$
where the first summation can be bounded by
\begin{align*}
\sum_{\bar\ell:\gamma(\bar\ell)<2^{-3L\cdot d_2}}\frac{\gamma(\bar\ell)}{\gamma(\ell^*)}\cdot \K(\bar\ell)
&\le\frac{2^{-3L\cdot d_1}}{\gamma(\ell^*)}\cdot\sum_{\bar\ell}2^{-3L\cdot(d_2-d_1)}\cdot n^4T^2
\tag{by \Cref{eq:lem:depth_tail_bound_2}}\\
&\le\frac{2^{-3L\cdot d_1}}{\gamma(\ell^*)}\cdot2^{2L\cdot(d_2-d_1)}\cdot2^{-3L\cdot(d_2-d_1)}\cdot n^4T^2
\tag{since $\ell^*$ is fixed and each message is at most $2L$ bits}\\
&=\frac{2^{-3L\cdot d_1}}{\gamma(\ell^*)}\cdot\frac{2n^4T^2}{2^L}
\tag{since $d_2-d_1\ge1$}
\end{align*}
and the second summation is bounded by
\begin{equation*}
\sum_{\bar\ell:\gamma(\bar\ell)\ge2^{-3L\cdot d_2}}\frac{\gamma(\bar\ell)}{\gamma(\ell^*)}\cdot \K(\bar\ell)
\le\sum_{\bar\ell}\frac{\gamma(\bar\ell)}{\gamma(\ell^*)}\cdot\alpha\cdot d_2^2L^2
=\alpha\cdot d_2^2L^2.
\tag{by \Cref{eq:lem:depth_tail_bound_6}}
\end{equation*}
Then combining \Cref{eq:lem:depth_tail_bound_1}, we have
$$
\lambda\cdot(d_2-d_1-2d)\cdot\Pr\sbra{\D\ge d_2\mid\ell^*}\le \alpha\cdot d_2^2L^2+\frac{2^{-3L\cdot d_1}}{\gamma(\ell^*)}\cdot\frac{2n^4T^2}{2^L}.
$$
Assume $2^L\ge8n^4T^2$ and $d_2-d_1\ge2d+1$. Then
\begin{equation*}
\Pr\sbra{\D\ge d_2\mid\ell^*}\le\frac{\alpha\cdot d_2^2L^2}{\lambda\cdot(d_2-d_1-2d)}+\frac14\cdot\frac{2^{-3L\cdot d_1}}{\gamma(\ell^*)}.
\tag*{\qedhere}
\end{equation*}
\end{proof}

\begin{corollary}\label{cor:depth_tail_bound}
Assume $\gamma^*\ge3/4$, $T\le n$, $L\ge\Theta(\log(n))$, and $\lambda\ge\Theta(dL^2\log^2(n))$. Then for each $k=0,1,\ldots,4\log(n)$, we have
$$
\Pr\sbra{\D\ge4kd}\le2^{-k}+\frac k{n^5}.
$$
\end{corollary}
\begin{proof}
We prove the bound by induction on $k$. The base case $k=0$ is trivial.
For the inductive case, let $\ell^*$ be the first $4(k-1)d$ communication messages. Then we bound
$$
P:=\sum_{\ell^*:\gamma(\ell^*)/\gamma^*<2^{-3L\cdot4(k-1)d}}\frac{\gamma(\ell^*)}{\gamma^*}\cdot\Pr\sbra{\D\ge4kd\mid\ell^*}
$$
and
$$
Q:=\sum_{\ell^*:\gamma(\ell^*)/\gamma^*\ge2^{-3L\cdot4(k-1)d}}\frac{\gamma(\ell^*)}{\gamma^*}\cdot\Pr\sbra{\D\ge4kd\mid\ell^*}
$$
separately.

For $P$, observe that if $k=1$ then $\ell^*$ is root of the protocol, thus $\gamma(\ell^*)=\gamma^*$ and $P=0$. 
On the other hand, if $k\ge2$, then
\begin{align*}
P
&\le\sum_{\ell^*:\gamma(\ell^*)/\gamma^*<2^{-3L\cdot4(k-1)d}}2^{-3L\cdot4(k-1)d}
\le\sum_{\ell^*}2^{-3L\cdot4(k-1)d}\\
&\le2^{2L\cdot4(k-1)d}\cdot2^{-3L\cdot4(k-1)d}
\tag{each communication message is at most $2L$ bits}\\
&=2^{-L\cdot4(k-1)d}\le n^{-5}.
\tag{since $k\ge2$ and $L\ge\Theta(\log(n))$}
\end{align*}
Now we turn to $Q$.
Applying \Cref{lem:depth_tail_bound} with $\ell^*$ and $d_1=4(k-1)d,d_2=4kd$, we have
\begin{align*}
Q
&\le\sum_{\ell^*:\gamma(\ell^*)/\gamma^*\ge2^{-3L\cdot4(k-1)d}}\frac{\gamma(\ell^*)}{\gamma^*}\cdot\pbra{\frac{16\alpha\cdot k^2d^2L^2}{2dR}+\frac14\cdot\frac{2^{-3L\cdot4(k-1)d}}{\gamma(\ell^*)}}\\
&\le\sum_{\ell^*}\frac{\gamma(\ell^*)}{\gamma^*}\cdot\pbra{\frac{8\alpha\cdot k^2dL^2}{\lambda}+\frac1{4\gamma^*}}\\
&=\Pr\sbra{\D\ge4(k-1)d}\cdot\pbra{\frac{8\alpha\cdot k^2dL^2}{\lambda}+\frac1{4\gamma^*}}\\
&\le\Pr\sbra{\D\ge4(k-1)d}\cdot\frac12
\tag{since $\gamma^*\ge3/4$ and $\lambda\ge\Theta(dL^2\log^2(n)),k\le4\log(n)$}\\
&\le\pbra{2^{-(k-1)}+\frac{k-1}{n^5}}\cdot\frac12
\le2^{-k}+\frac{k-1}{n^5}.
\tag{by induction hypothesis}
\end{align*}
By adding up $P$ and $Q$, we complete the induction.
\end{proof}

Given \Cref{cor:depth_tail_bound} and suitable choice of the parameters, we now prove the second moment bound.
\begin{proof}[Proof of \Cref{lem:second_moment}]
With $L=\Theta(\log(n))$, $T=\Theta(\sqrt{\log(n)})$, and $\lambda=\Theta(d\log^4(n))$, by \Cref{fct:gammastar}, we have $\gamma^*\ge3/4$.
Therefore the second moment of $\D$ is 
\begin{align*}
\E[\D^2]
&\le\sum_{k=0}^{4\log(n)}\pbra{4(k+1)d}^2\cdot\Pr\sbra{\D\ge4kd}+\Pr\sbra{\D\ge16 d\log(n)}\cdot(2n^2)^2
\tag{by \Cref{clm:finite_steps}}\\
&\le\sum_{k=0}^{4\log(n)}\pbra{4(k+1)d}^2\cdot\pbra{2^{-k}+\frac k{n^5}}+\pbra{n^{-4}+\frac{4\log(n)}{n^5}}\cdot(2n^2)^2
\tag{by \Cref{cor:depth_tail_bound}}\\
&=O(d^2).
\tag*{\qedhere}
\end{align*}
\end{proof}
\section{Fourier Growth Reductions For General Gadgets}\label{sec:gadget}

In this section, we show that Fourier growth bounds of communication protocols for general (constant-sized) gadgets can be reduced to the bounds of XOR-fiber, and vice versa.
This implies that in the study of Fourier growth, they are all equivalent.

Let $m_1,m_2$ be two positive integers.
Let $g\colon\binpm^{m_1}\times\binpm^{m_2}\to\binpm$ be a gadget.
Recall that $\unif$ is the uniform distribution over $\binpm^n$.
We now use $\unif_1,\unif_2,\bar\unif_1,\bar\unif_2$ to denote the uniform distributions over $\binpm^{m_1},\binpm^{m_2},(\binpm^{m_1})^n,(\binpm^{m_2})^n$ respectively.
We define the $g$-fiber of communication protocols similar to the XOR-fiber:

\begin{definition}\label{def:g-fiber}
For any randomized two-party protocol $\Ccal\colon(\binpm^{m_1})^n\times(\binpm^{m_2})^n\to[-1,1]$, its $g$-fiber, denoted by $\Ccal_{\downarrow g}\colon\binpm^n\to[-1,1]$, is defined by
$$
\Ccal_{\downarrow g}(z)=\E_{\xbm\sim\bar\unif_1,\ybm\sim\bar\unif_2}\sbra{\Ccal(\xbm,\ybm)\mid g(\xbm_i,\ybm_i)=z_i,~\forall i},
$$
where the expectation is also over the internal randomness of $\Ccal$.
\end{definition}

To compare the Fourier growth bounds between gadgets, we use $L_{1,k}(g,d,m_1,m_2,n)$ to denote the upper bound of the level-$k$ Fourier growth for the $g$-fiber of an arbitrary randomized communication protocol $\Ccal\colon(\binpm^{m_1})^n\times(\binpm^{m_2})^n\to[-1,1]$ with at most $d$ bits of communication, where $g\colon\binpm^{m_1}\times\binpm^{m_2}\to\binpm$ is the gadget.
Since randomized protocols are convex combinations of deterministic protocols of the same cost, using this notation, our main results \Cref{thm:boolean_bound_level_one,thm:boolean_bound_level_two} can be rephrased as
$$
L_{1,1}(\mathrm{XOR},d,1,1,n)\le O\pbra{\sqrt d}
\quad\text{and}\quad
L_{1,2}(\mathrm{XOR},d,1,1,n)\le O\pbra{d^{3/2}\log^3(n)}.
$$

For any set $S\subseteq[m_1]$, define $x_S=\prod_{i\in S}x_i$, and similarly for $y_T$ with $T\subseteq[m_2]$.
Similar to the standard Fourier representation of Boolean functions, the gadget $g$, which is a two-party function, also has Fourier representation:
$$
g(x,y)=\sum_{S\subseteq[m_1],T\subseteq[m_2]}\hat g(S,T)\cdot x_Sy_T,
\quad\text{where}\quad
\hat g(S,T)=\E_{\xbm\sim\unif_1,\ybm\sim\unif_2}\sbra{g(\xbm,\ybm)\cdot\xbm_S\ybm_T}.
$$

For convenience, we will assume $g$ satisfies the following assumption.
It's easy to see that the XOR gadget satisfies it.
\begin{assumption}\label{as:balance}
$\hat g(S,T)=0$ if $S=\emptyset$ or $T=\emptyset$.
\end{assumption}
\begin{remark}
This assumption is equivalent to the fact that, restricted on any input to Alice's side, the remaining function on Bob's side is balanced, and vice versa.

Even if $g$ does not satisfy the assumption, then we can embed it inside a similar gadget $g'\colon\binpm^{m_1+1}\times\binpm^{m_2+1}\to\binpm$, where we XOR the last bit of Alice and the last bit of Bob to the old gadget $g$ applied to Alice's first $m_1$ bits and Bob's first $m_2$ bits, i.e.,
$$
g'(x,y)=x_{m_1+1}y_{m_2+1}\cdot g(x_{\le m_1},y_{\le m_2}).
$$
Then $g'$ satisfies the assumption and inherits most properties of $g$ sufficient for studies in communication complexity tasks.
\end{remark}

Now for a protocol $\Ccal\colon(\binpm^{m_1})^n\times(\binpm^{m_2})^n\to[-1,1]$, it is also a two-party function and thus admitting similar Fourier representation.
We view an input from $(\binpm^{m_1})^n$ as indexed by a tuple in $[n]\times[m_1]$.
Therefore any subset of $(\binpm^{m_1})^n$ is uniquely identified as $\bigcup_{i\in[n]}\cbra{i}\times S_i$, where each $S_i\subseteq[m_1]$.
We use $S^{[n]}$ to denote $(S_i)_{i\in[n]}$.
Thus the Fourier coefficients of $\Ccal$ can be written as
$$
\hat\Ccal(S^{[n]},T^{[n]}):=\hat\Ccal\pbra{\bigcup_{i\in[n]}\cbra{i}\times S_i,\bigcup_{i\in[n]}\cbra{i}\times T_i},
$$
and the Fourier representation of $\Ccal$ is
$$
\Ccal(x,y)=
\sum_{S^{[n]},J^{[n]}}\hat\Ccal(S^{[n]},T^{[n]})\cdot\prod_{i\in[n]}x_{i,S_i}\cdot\prod_{j\in[n]}y_{j,T_j},
$$
where $x_{i,S}=\prod_{j\in S}x_{i,j}$ and similar for $y_{j,T}$.

Under this notation and assuming \Cref{as:balance}, we can effectively compute the Fourier coefficients of any $g$-fiber.
\begin{fact}\label{fct:g-fiber_fourier}
Assume gadget $g\colon\binpm^{m_1}\times\binpm^{m_2}\to\binpm$ satisfies \Cref{as:balance}.
Then we have
$$
\hat{\Ccal_{\downarrow g}}(I)
=\sum_{\substack{S^I,T^I\\S_i\neq\emptyset,T_i\neq\emptyset,\forall i\in I}}\hat\Ccal(S^I,T^I)\cdot\prod_{i\in I}\hat g(S_i,T_i)
\quad
\text{for any $I\subseteq[n]$,}
$$
where we use $S^I$ to denote $S^{[n]}$ with $S_j$ fixed to $\emptyset$ for all $j\notin I$.
\end{fact}
\begin{proof}
Observe that
\begin{align*}
\hat{\Ccal_{\downarrow g}}(I)
&= \E_{\zbm \sim \unif}\sbra{\Ccal_{\downarrow g}(\zbm) \cdot \prod_{i\in I}\zbm_i}\\
&= \E_{\zbm \sim \unif}\sbra{\E_{\xbm\sim\bar\unif_1,\ybm \sim \bar\unif_2}\sbra{\Ccal(\xbm,\ybm) \mid g(\xbm_i, \ybm_i)=\zbm_i,~\forall i} \cdot \prod_{i\in I} \zbm_i}\\
&= \E_{\zbm \sim \unif}\sbra{  \E_{\xbm\sim\bar\unif_1,\ybm \sim \bar\unif_2}\sbra{ \Ccal(\xbm,\ybm) \cdot \prod_{i\in I} g(\xbm_i, \ybm_i) \mid g(\xbm_i, \ybm_i)=\zbm_i,~\forall i } }.
\end{align*}
Since $\hat{g}(\emptyset, \emptyset) = 0$ by \Cref{as:balance}, every pair $(x,y)$ is sampled with the same probability under the conditional distribution. 
Thus we get 
$$
\hat{\Ccal_{\downarrow g}}(I) = \E_{\xbm\sim\bar\unif_1,\ybm\sim \bar\unif_2}\sbra{\Ccal(\xbm,\ybm) \cdot \prod_{i\in I} g(\xbm_i, \ybm_i)}.
$$
Now we expand $\Ccal$ and $g$ in the Fourier basis and obtain
\begin{align*}
\hat{\Ccal_{\downarrow g}}(I) 
&=
\E_{\xbm\sim\bar\unif_1,\ybm\sim\bar\unif_2}\sbra{
\pbra{\sum_{S^{[n]},T^{[n]}}
\hat\Ccal(S^{[n]},T^{[n]})\prod_{i\in[n]}\xbm_{i,S_i}\prod_{j\in[n]}\ybm_{j,T_j}}
\cdot
\prod_{i\in I}\pbra{
\sum_{S_i,T_i}\hat g(S_i,T_i)\xbm_{i,S_i}\ybm_{i,T_i}}}\\
&=
\E_{\xbm\sim\bar\unif_1,\ybm\sim\bar\unif_2}\sbra{
\pbra{\sum_{S^{[n]},T^{[n]}}
\hat\Ccal(S^{[n]},T^{[n]})\prod_{i\in[n]}\xbm_{i,S_i}\prod_{j\in[n]}\ybm_{j,T_j}}
\pbra{
\sum_{S^I,T^I}
\prod_{i\in I}\hat g(S_i,T_i)\xbm_{i,S_i}\ybm_{i,T_i}}}\\
&=
\sum_{S^I,T^I}\hat\Ccal(S^I,T^I)\cdot\prod_{i\in I}\hat g(S_i,T_i)\\
&=\sum_{\substack{S^I,T^I\\S_i\neq\emptyset,T_i\neq\emptyset,\forall i\in I}}\hat\Ccal(S^I,T^I)\cdot\prod_{i\in I}\hat g(S_i,T_i),
\tag{by \Cref{as:balance}}
\end{align*}
as desired.
\end{proof}

Now we present the reduction from XOR-fiber to a general $g$-fiber.
\begin{theorem}\label{thm:xor_to_g}
Assume gadget $g\colon\binpm^{m_1}\times\binpm^{m_2}\to\binpm$ satisfies \Cref{as:balance}. Then
\begin{align*}
L_{1,k}(\mathrm{XOR},d,1,1,n)
&\le\pbra{\max_{S,T}|\hat g(S,T)|}^{-k}\cdot L_{1,k}(g,d,m_1,m_2,n)\\
&\le2^{(m_1+m_2)\cdot k/2}\cdot L_{1,k}(g,d,m_1,m_2,n).
\end{align*}
\end{theorem}
\begin{proof}
Let $\Ccal\colon\binpm^n\times\binpm^n\to[-1,1]$ be an arbitrary protocol of cost at most $d$.
Then for a fixed set $I\subseteq[n]$, by \Cref{fct:g-fiber_fourier} applied to the XOR gadget, we have
\begin{equation}\label{eq:lem:xor_to_g_1}
\hat{\Ccal_{\downarrow\mathrm{XOR}}}(I)=\hat\Ccal(1^I,1^I).
\end{equation}
Let $S\subseteq[m_1]$ and $T\subseteq[m_2]$ maximize $|\hat g(S,T)|$.
Since $g$ satisfies \Cref{as:balance}, we know $S$ and $T$ are not empty sets.

Now define a different protocol $\Ccal'\colon(\binpm^{m_1})^n\times(\binpm^{m_2})^n\to[-1,1]$ as follows:
After receiving input $x$, Alice computes $x'_i=x_{i,S}$ for each block $x_i$; and Bob computes similarly $y'_i=y_{i,T}$ upon receiving input $y$.
Then they execute the protocol $\Ccal$ on $x'$ and $y'$.
That is, $\Ccal'(x,y)=\Ccal(x',y')$.
Therefore, for any $I\subseteq[n]$ and $S^I,T^I$ satisfying $S_i\neq\emptyset,T_i\neq\emptyset$ for $i\in I$, we have
$$
\hat{\Ccal'}(S^I,T^I)=
\begin{cases}
\hat\Ccal(1^I,1^I) & S_i=S,T_i=T,~\forall i\in I,\\
0 & \text{otherwise.}
\end{cases}
$$
Then by \Cref{eq:lem:xor_to_g_1} and \Cref{fct:g-fiber_fourier} applied to $\Ccal'$ with gadget $g$, we have
$$
\hat{\Ccal_{\downarrow g}'}(I)
=\hat\Ccal(1^I,1^I)\cdot\hat g(S,T)^{|I|}
=\hat{\Ccal_{\downarrow\mathrm{XOR}}}(I)\cdot\hat g(S,T)^{|I|}.
$$
Now summing over all $I\subseteq[n]$ of size $k$, we have
\begin{align*}
L_{1,k}(\Ccal_{\downarrow\mathrm{XOR}})
&=\sum_{I\subseteq[n]:|I|=k}\abs{\hat{\Ccal_{\downarrow\mathrm{XOR}}}(I)}
=|\hat g(S,T)|^{-k}\cdot\sum_{I\subseteq[n]:|I|=k}\abs{\hat{\Ccal_{\downarrow g}'}(I)}
=|\hat g(S,T)|^{-k}\cdot L_{1,k}(\Ccal'_{\downarrow g})\\
&\le
|\hat g(S,T)|^{-k}\cdot L_{1,k}(g,d,m_1,m_2,n).
\tag{since $\Ccal'$ has cost at most $d$}
\end{align*}
Since $\Ccal$ is arbitrary, this proves the first half of \Cref{thm:xor_to_g}.
To prove the second half, we use an averaging argument and Parseval's identity on $g$:
\begin{equation*}
|\hat g(S,T)|
\ge\sqrt{2^{-m_1-m_2}\sum_{S',T'}\hat g(S',T')^2}
=\sqrt{2^{-m_1-m_2}}.
\tag*{\qedhere}
\end{equation*}
\end{proof}

Using similar analysis, we also have a reduction from a general $g$-fiber to XOR-fiber.
\begin{theorem}\label{thm:g_to_xor}
Assume gadget $g\colon\binpm^{m_1}\times\binpm^{m_2}\to\binpm$ satisfies \Cref{as:balance}. Then
\begin{align*}
L_{1,k}(g,d,m_1,m_2,n)
&\le\pbra{\sum_{S,T}|\hat g(S,T)|}^k\cdot L_{1,k}(\mathrm{XOR},d,1,1,n)\\
&\le2^{(m_1+m_2)\cdot k/2}\cdot L_{1,k}(\mathrm{XOR},d,1,1,n).
\end{align*}
\end{theorem}
\begin{proof}
Let $\Ccal\colon(\binpm^{m_1})^n\times(\binpm^{m_2})^n\to[-1,1]$ be an arbitrary protocol of cost at most $d$.
Then for a fixed set $I\subseteq[n]$, by \Cref{fct:g-fiber_fourier} applied to gadget $g$ and using \Cref{as:balance}, we have
$$
\hat{\Ccal_{\downarrow g}}(I)
=\sum_{S^I,T^I}\hat\Ccal(S^I,T^I)\cdot\prod_{i\in I}\hat g(S_i,T_i).
$$
Therefore
$$
L_{1,k}(\Ccal_{\downarrow g})
\le\sum_{I\subseteq[n]:|I|=k}\sum_{S^I,T^I}\abs{\hat\Ccal(S^I,T^I)}\cdot\abs{\prod_{i\in I}\hat g(S_i,T_i)}.
$$

Now let $M=\sum_{S,T}|\hat g(S,T)|$.
Let $\rho$ be a distribution over subsets of $[m_1]\times[m_2]$ and its probability density function is defined as:
$$
\rho(S,T)=|\hat g(S,T)|/M.
$$
Then we can rewrite $L_{1,k}(\Ccal_{\downarrow g})$ as
\begin{align}
L_{1,k}(\Ccal_{\downarrow g})
&\le\sum_{I\subseteq[n]:|I|=k}\E_{(\Sbm^I,\Tbm^I)\sim\rho^I}\sbra{\abs{\hat\Ccal(\Sbm^I,\Tbm^I)}\cdot M^k}
\notag\\
&=M^k\cdot\E_{(\Sbm^{[n]},\Tbm^{[n]})\sim\rho^{[n]}}\sbra{\sum_{I\subseteq[n]:|I|=k}\abs{\hat\Ccal(\Sbm^I,\Tbm^I)}}.
\label{eq:lem:g_to_xor_1}
\end{align}

Now we fix an arbitrary $(S^{[n]},T^{[n]})$ sampled from $\rho^{[n]}$.
Note that $S_i$ and $T_i$ are not empty by the definition of $\rho$ and \Cref{as:balance}.
Then define a different protocol $\Ccal'\colon\binpm^n\times\binpm^n\to[-1,1]$ as follows:
After receiving input $x$, Alice samples $x'\in(\binpm^{m_1})^n$ uniformly conditioned on $x'_{i,S_i}=x_i$ for all $i\in[n]$; and Bob samples similarly $y'\in(\binpm^{m_2})^n$ conditioned on $y'_{i,T_i}=y_i$ for all $i\in[n]$.
Then they execute the protocol $\Ccal$ on $x'$ and $y'$.
That is, $\Ccal'(x,y)=\E_{\xbm',\ybm'}[\Ccal(\xbm',\ybm')]$.
Therefore, for any $I\subseteq[n]$, we have
$$
\hat{\Ccal'}(1^I,1^I)=\hat\Ccal(S^I,T^I).
$$
By \Cref{fct:g-fiber_fourier} applied to $\Ccal'$ and the XOR gadget, we have
$$
\hat{\Ccal'_{\downarrow\mathrm{XOR}}}(I)=\hat{\Ccal'}(1^I,1^I)=\hat\Ccal(S^I,T^I).
$$
Since $\Ccal'$ has cost at most $d$, we have
$$
\sum_{I\subseteq[n]:|I|=k}\abs{\hat\Ccal(S^I,T^I)}
=\sum_{I\subseteq[n]:|I|=k}\abs{\hat{\Ccal'_{\downarrow\mathrm{XOR}}}(I)}
=L_{1,k}(\Ccal'_{\downarrow\mathrm{XOR}})
\le L_{1,k}(\mathrm{XOR},d,1,1,n).
$$
Putting back to \Cref{eq:lem:g_to_xor_1}, we have
$$
L_{1,k}(\Ccal_{\downarrow g})
\le M^k\cdot L_{1,k}(\mathrm{XOR},d,1,1,n),
$$
which proves the first half of \Cref{thm:g_to_xor} since $\Ccal$ is arbitrary.
To prove the second half, we use Cauchy-Schwarz inequality and Parseval's identity on $g$:
\begin{equation*}
M=\sum_{S,T}|\hat g(S,T)|\le\sqrt{2^{m_1+m_2}\sum_{S,T}\hat g(S,T)^2}=\sqrt{2^{m_1+m_2}}.
\tag*{\qedhere}
\end{equation*}
\end{proof}

As a corollary, to study the Fourier growth bounds, we can switch between gadgets conveniently, as long as the gadgets have small size.

\begin{corollary}\label{cor:g_to_g'}
Assume gadgets $g\colon\binpm^{m_1}\times\binpm^{m_2}\to\binpm$ and $g'\colon\binpm^{m_1'}\times\binpm^{m_2'}\to\binpm$ satisfy \Cref{as:balance}.
Then
$$
L_{1,k}(g,d,m_1,m_2,n)\le2^{(m_1+m_2+m_1'+m_2')\cdot k/2}\cdot L_{1,k}(g',d,m_1',m_2',n).
$$
\end{corollary}
\section{Directions Towards Further Improvements}\label{sec:future}

In this section we propose potential directions for further improving our second level bounds.
In \Cref{sec:lift}, we show that better Fourier growth bounds can be obtained from strong lifting theorems in a black-box way. This relies on the Fourier growth reductions in \Cref{sec:gadget}.
In \Cref{sec:improved-hw}, we examine the bottleneck in our analysis and identify major obstacles within.

\subsection{Better Lifting Theorems Imply Better Fourier Growth}\label{sec:lift}

Let $f:\pmone^n \to \pmone$ be a Boolean function. Let $g: \pmone^{m_1} \times \pmone^{m_2} \to \pmone$ be a gadget.
A lifting theorem connects the communication complexity of $f \circ g$ with the query complexity of $f$.
Some lifting theorems show that a low-cost communication protocol can be simulated by a low-cost query algorithm.

To be more precise, let $\Ccal: (\pmone^{m_1})^n \times (\pmone^{m_2})^n \to [-1,1]$ be a randomized two-party protocol.
Recall \Cref{def:g-fiber}, the $g$-fiber of $\Ccal$, denoted $\Ccal_{\downarrow g}(z): \pmone^{n} \to [-1,1]$, is defined by
$$
\Ccal_{\downarrow g}(z) = \E_{\bm{x}\sim\bar\unif_1, \bm{y}\sim \bar\unif_2}\left[ \Ccal(\bm{x},\bm{y})\mid g(\bm{x}_i, \bm{y}_i)=z_i,~\forall i\right].
$$
We say that $g$ satisfies a strong lifting theorem if for all randomized protocols $\Ccal$ of small communication bits, there is a randomized decision tree of small depth that approximates $\Ccal_{\downarrow g}$ on each input with error $1/\poly(n)$ (see e.g., \cite{GPW20}).

\begin{theorem}\label{thm:gadget}
Assume gadget $g\colon\binpm^{m_1}\times\binpm^{m_2}\to\binpm$ satisfies \Cref{as:balance}.
Assume for any randomized protocol $\Ccal\colon(\binpm^{m_1})^n\times(\binpm^{m_2})^n\to[-1,1]$ with at most $d$ bits of communication, there exists a randomized decision tree $\Tcal$ of depth at most $D$ that approximates $\Ccal_{\downarrow g}$ with pointwise error at most $1/n^k$, i.e.,
$$
\abs{\Tcal(z)-\Ccal_{\downarrow g}(z)}\le n^{-k}
\quad\forall z\in\binpm^n.
$$

Then, for any randomized protocol $\Ccal'\colon\binpm^n\times\binpm^n\to[-1,1]$ with at most $d$ bits of communication, its XOR-fiber $\Ccal'_{\downarrow\mathrm{XOR}}$ has level-$k$ Fourier growth
\begin{align*}
L_{1,k}(\Ccal'_{\downarrow\mathrm{XOR}})
&\le\pbra{\max_{S,T}|\hat g(S,T)|}^{-k}\cdot\sqrt{D^k\cdot O\pbra{\log(n)}^{k-1}}\\
&\le2^{(m_1+m_2)\cdot k/2}\cdot\sqrt{D^k\cdot O\pbra{\log(n)}^{k-1}}.
\end{align*}
\end{theorem}

As a simple corollary, we see that if the assumption of \Cref{thm:gadget} holds with $k=2$, $D= d \cdot \polylog(n)$, and a polylogarithmic-sized gadget $g$ (i.e., $2^{m_1},2^{m_2}\le\polylog(n)$), then the second level Fourier growth of the XOR-fiber of any randomized protocol of cost $d$ is at most $d\cdot\polylog(n)$ as desired.

We also remark that state-of-the-art lifting results hold with the gadget $g$ being either:
\begin{itemize}
\item The inner product on $m_1 = m_2 =  O(\log(n))$ bits~\cite{CFKMP19}. 
However, for such $g$ the largest Fourier coefficient squared is $1/\poly(n)$, which yields a trivial bound in Theorem~\ref{thm:gadget}.
\item The index function with $m_1 = \poly(n)$, $m_2 = \log(m_1)$~\cite{GPW20}.\footnote{For deterministic lifting, a better bound $m_1=O(n\log(n))$ is known \cite{lovett2022lifting}, but it doesn't suffice for our reduction.} 
In this case the largest Fourier coefficient squared is $1/m_1^2$, which again yields a trivial bound in Theorem~\ref{thm:gadget}.
Nonetheless, even a polynomial improvement on $m_1$, say $m_1 = n^{0.01}$, would give new non-trivial bounds in Theorem~\ref{thm:gadget} and in turn improves our lower bound on the XOR-lift of Forrelation.
\end{itemize}

\begin{proof}[Proof of \Cref{thm:gadget}]
Let $\Ccal\colon(\binpm^{m_1})^n\times(\binpm^{m_2})^n\to[-1,1]$ be a randomized protocol of cost at most $d$.
Then by assumption, $\Ccal_{\downarrow g}$ can be approximated up to error $1/n^k$ by a randomized decision tree $\Tcal$ of depth at most $D$.
Thus any Fourier coefficient of $\Ccal_{\downarrow g}$ and $\Tcal$ differs by at most $1/n^k$.
Therefore by the level-$k$ Fourier growth bounds on randomized decision trees \cite{Tal20,SSW21}, we have
$$
L_{1,k}(\Ccal_{\downarrow g})
\le \sum_{S\subseteq[n]:|S|=k}\pbra{n^{-k} + \abs{\hat{\Tcal}(S)}}
\le \sqrt{D^k\cdot O(\log(n))^{k-1}}.
$$
Since $\Ccal$ is arbitrary, the claimed bound for $\Ccal'_{\downarrow\mathrm{XOR}}$ follows from \Cref{thm:xor_to_g}.
 \end{proof}

\subsection{Sums of Squares of Quadratic Forms for Pairwise Clean Sets}
\label{sec:improved-hw}

In our analysis for the level-two bound, we showed that one can transform a general protocol to a $4$-wise clean protocol with parameter $\lambda = d\cdot\polylog(n)$ by adding $O(d)$ additional cleanup steps in expectation. If one could show that with essentially the same number of steps, one could take $\lambda = \polylog(n)$, then we would obtain the optimal level-two bound of $d \cdot \polylog(n)$.

We recall that to bound the number of cleanup steps, we rely on a concentration inequality for sums of squares of orthonormal quadratic forms (\Cref{thm:quadratic_concentration}), which says that if $M_1, \ldots, M_m$ are matrices with zero diagonal and form an orthonormal set when viewed as $n^2$ dimensional vectors,
then the random variable $\lQ = \sum_{i=1}^m \ip{\lX \tensor \lX}{M_i}^2$ satisfies $\Pr_{\lx \sim \gamma_n}[\lQ \ge t] \le e^{-\Omega(\sqrt{t})}$ for any  $t\gtrsim m^2$. 
Using this tail bound for $m= \Theta(d)$ and conditioning on $\lx \in X$ where $X$ is an arbitrary subset of $\Rbb^n$ with Gaussian measure $\approx 2^{-d}$, we obtained a bound $\BE_{\lx \sim \gamma}[\lQ \midd \lx \in X] \lesssim d^2$. 
This shows that there can be at most $O(d)$ such quadratic forms $M_i$'s where the value $\BE_{\lx \sim \gamma}\sbra{\ip{\lX \tensor \lX}{M_i}^2 \mid \lx \in X}$ can be larger than $d$ and hence, the reason we can only take $\lambda \approx d$. We note that the argument just described is for the non-adaptive setting, while in our case the $M_i$'s are also being chosen adaptively, so additional work is needed. 

The next example shows that the aforementioned statement is tight even in the non-adaptive setting where the $M_i$'s are fixed: in particular, there is a set $X$ of large measure and $\approx d$ such orthonormal quadratic forms where the above expectation after conditioning on $\lx \in X$ is $\Theta(d^2)$.

\begin{example} 
For $1\le i< j\le\sqrt d$, let $M_{ij} = E_{ij}$ for $i < j$ where $E_{ij}$ denotes the $n \times n$ matrix where only the $(i,j)$ entry is one. Note that the matrices $M_{ij}$ form an orthonormal set and they all have a zero diagonal. Let $X = \cbra{x \in \Rbb^n \mid |x_i| \gtrsim d^{1/4} \text{ for all $i \le d^{1/2}$}}$. Then, the Gaussian measure $\gamma(X) = 2^{-\Theta(d)}$ but 
\[ 
\BE_{\lx \sim \gamma}\sbra{ \sum_{1\le i< j\le \sqrt d} \ip{\lX \tensor \lX}{M_{ij}}^2 \mid \lx \in X} = \Theta(d^2).
\]
\end{example}

Note that the set $X$ in the example above is not pairwise clean and for our application, one can get around it by first ensuring that the protocol is pairwise clean and then proceeding with the 4-wise cleanup process. Motivated by this, we speculate that  when the set is pairwise clean, then the expected value of the sum of squares of orthonormal quadratic forms is much smaller unlike the example above.
Assuming such a statement and combining it with our ideas for handling the adaptivity suggests a potential way of improving the level-two bounds.

\bibliographystyle{alpha} 
\bibliography{ref}

\appendix
\section{Gap-Hamming Lower Bounds}\label{app:thm:gap_hamming}

As an immediate consequence of~\Cref{thm:coin_problem}, we can derive optimal lower bounds against the Gap-Hamming problem as in~\Cref{thm:gap_hamming}.
\begin{proof}[Proof of~\Cref{thm:gap_hamming}]
Set $\rho=10/\sqrt{n}$. Fix the randomness to be any $r\in\bin^*$ and let $\Ccal_r$ refer to the deterministic protocol $\Ccal$ with randomness fixed to $r$. Suppose $d\le \tau \cdot n$ for a sufficiently small constant $\tau$, we apply \Cref{thm:coin_problem} on $\rho$ as well as $-\rho$, and apply triangle inequality to conclude that 
\[ 
\abs{\E_{\lz\sim \biased{\rho}}[h_r(\lz)]-\E_{\lz\sim \biased{-\rho}}[h_r(\lz)]}\le 2\cdot O\pbra{\sqrt{d/n}}< 1/9.
\]
Let $\sigma_\rho$ be the distribution of $(\lx,\ly)$ induced by sampling $\lx\sim \biased0$ and $\lz\sim \biased{\rho}$ and letting $\ly=\lx\odot \lz$, similarly define $\sigma_{-\rho}$ but with $\lz\sim \biased{-\rho}$. We now expand $h_r(z)$ in terms of $\Ccal(x,y)$, take an expectation over $r$ and apply triangle inequality to conclude that
\begin{equation}\label{eq:gap_hamming} 
\abs{\E_{(\lx,\ly)\sim \sigma_\rho}[\Ccal(\lx,\ly)]-\E_{(\lx,\ly)\sim \sigma_{-\rho}}[\Ccal(\lx,\ly)]}<1/9.
\end{equation}

Hoeffding's inequality implies that for $\lz\sim \biased{\rho}$, we have
\[ 
\Pr\sbra{\abs{\sum_i \lz_i - 10 \sqrt{n}}\ge 5\sqrt{n}}\le 2\exp\cbra{\tfrac{-2\cdot (5\sqrt{n})^2}{4n}} < 1/18.
\]
This implies that a random $(\lx,\ly)\sim \sigma_\rho$ is a \textsc{yes} instance of the Gap-Hamming problem with probability larger than $17/18$. Let $\tilde \sigma_\rho$ denote $\sigma_\rho$ conditioned on \textsc{Yes} instances of the Gap-Hamming problem. Similarly define $\tilde \sigma_{-\rho}$ to be $\sigma_{-\rho}$ conditioned on \textsc{No} instances of the Gap-Hamming problem. Since $\Ccal(x,y)$ has outputs in $[-1,1]$, we have 
\[ 
\abs{\E_{(\lx,\ly)\sim \sigma_\rho}[\Ccal(\lx,\ly)]-\E_{(\lx,\ly)\sim \tilde \sigma_\rho}[\Ccal(\lx,\ly)]}<1/9
\]
and
\[
\abs{\E_{(\lx,\ly)\sim \sigma_{-\rho}}[\Ccal(\lx,\ly)]-\E_{(\lx,\ly)\sim \tilde \sigma_{-\rho}}[\Ccal(\lx,\ly)]}<1/9. 
\]
This, along with \Cref{eq:gap_hamming} and triangle inequality, implies that 
\begin{equation*}
\abs{\E_{(\lx,\ly)\sim \tilde \sigma_\rho}[\Ccal(\lx,\ly)]-\E_{(\lx,\ly)\sim\tilde \sigma_{-\rho}}[\Ccal(\lx,\ly)]}< 1/3.
\end{equation*}
However, this contradicts the assumption that the protocol $\Ccal$ solves the Gap-Hamming problem with advantage at least $2/3$. 
\end{proof}

\section{Concentration for Sum of Squares of Quadratic Forms}\label{app:thm:quadratic_concentration}

Here we prove \Cref{thm:quadratic_concentration}.
While it follows from \cite[Theorem 6]{A20} which is a Banach space-valued version of the Hanson-Wright inequality, in our setting a weaker statement suffices, for which we give a self-contained proof following \cite{A20}.

For any integer $n\ge1$, we use $\Bcal^n=\cbra{x\in\Rbb^n\mid\vabs{x}\le1}$ to denote the unit Euclidean ball in $\Rbb^n$.
For any two sets $A,B\subseteq\Rbb^n$, we define $A+B=\cbra{x+y\mid x\in A,y\in B}$.
For any set $A\in\Rbb^n$ and any number $t\in\Rbb$, we define $tA=\cbra{t\cdot x\mid x\in A}$.
Let $\Phi\colon\Rbb\to[0,1]$ be the cumulative distribution function of the standard Gaussian distribution, i.e., $\Phi(a)=\frac1{\sqrt{2\pi}}\int_{-\infty}^ae^{-u^2/2}\sd u$.

Now we cite the famous Gaussian isoperimetric inequality \cite{borell1975brunn,sudakov1978extremal}.
\begin{theorem}[Gaussian Isoperimetric Inequality]\label{thm:gaussian_iso_ineq}
Let $A\subseteq\Rbb^n$ be a measurable set and assume $\gamma_n(A)\ge\Phi(a)$ for some $a\in\Rbb$.
Then for any $t\ge0$, we have $\gamma_n(A+t\Bcal^n)\ge\Phi(a+t)$.
\end{theorem}

In particular, if $\gamma_n(A)\ge1/2$, then we can pick $a=0$ in \Cref{thm:gaussian_iso_ineq} and have 
\begin{equation}\label{eq:gaussian_iso_ineq_1/2}
\gamma_n(A+t\Bcal^n)\ge\Phi(t)\ge1-e^{-t^2/2}.
\end{equation}

Now we are ready to prove \Cref{thm:quadratic_concentration}.

\begin{proof}[Proof of \Cref{thm:quadratic_concentration}]
Note that the bound is trivial when $m=0$. Thus from now on we assume without loss of generality $m\ge1$.

For each $x\in\Rbb^n$, let $K_x=\sum_{i=1}^m\abra{x\tensor x,M_i}^2$. 
We first write $K_x$ as a squared Euclidean norm of a vector:
\begin{itemize}
\item For $i\in[m]$, we view $M_i$ as a length-$n^2$ row vector.
\item Let $M\in\Rbb^{m\times n^2}$ be a matrix where the $i$-th row is $M_i$.
\end{itemize}
Therefore we have
\begin{equation}\label{eq:quadratic_concentration_k}
K_x=\vabs{M(x\tensor x)}^2=\vabs{M(x\otimes x)}^2,
\end{equation}
where $\otimes$ is the standard tensor product and the second equality follows since each $M_i$ has zero diagonal.

Define $f(y)=\vabs{M(y\otimes y)}$, $g(y)=\sup_{z\in\Sbb^{n-1}}\vabs{M(z\otimes y)}$, and $h(y)=\sup_{z\in\Sbb^{n-1}}\vabs{M(y\otimes z)}$.
Let $F=\E_{\ly\sim\gamma_n}[f(\ly)]$, $G=\E_{\ly\sim\gamma_n}[g(\ly)]$, and $H=\E_{\ly\sim\gamma_n}[h(\ly)]$ be their mean.
Define the set 
$$
A=\cbra{y\in\Rbb^n\mid f(y)<6F,\ g(y)<6G,\text{ and }h(y)<6H}.
$$
By Markov's inequality and union bound, we have the Gaussian measure of $A$ is $\gamma_n(A)\ge1/2$.
Then by \Cref{eq:gaussian_iso_ineq_1/2}, we have
\begin{equation}\label{eq:quadratic_concentration_1}
\gamma_n(A+t\Bcal^n)\ge1-e^{-t^2/2}
\quad\text{holds for all $t\ge0$.}
\end{equation}
Now for an arbitrary $x\in A+t\Bcal^n$, we write $x=y+tz$ where $y\in A$ and $z\in\Bcal^n$.
Then 
\begin{align*}
\vabs{M(x\otimes x)}
&\le\vabs{M(y\otimes y)}+t\cdot\vabs{M(y\otimes z)}+t\cdot\vabs{M(z\otimes y)}+t^2\cdot\vabs{M(z\otimes z)}\\
&<6F+6t(G+H)+t^2V,
\end{align*}
where $V=\sup_{z\in\Sbb^{n-1}}\vabs{M(z\otimes z)}$.
This, together with \Cref{eq:quadratic_concentration_k} and \Cref{eq:quadratic_concentration_1}, implies
\begin{equation}\label{eq:quadratic_concentration_2}
\Pr_{\lx\sim\gamma_n}\sbra{K_{\lx}\ge\pbra{6F+6t(G+H)+t^2V}^2}
\le\Pr_{\lx\sim\gamma_n}\sbra{\lx\notin A+t\Bcal^n}
=1-\gamma_n(A+t\Bcal^n)
\le e^{-t^2/2}.
\end{equation}
Now we calculate $F,G,H,V$ in the following claim, the proof of which will be presented later.
\begin{claim}\label{clm:fghv}
$F\le\sqrt{2m}$, $G,H\le\sqrt m$, and $V\le1$.
\end{claim}
Plugging \Cref{clm:fghv} into \Cref{eq:quadratic_concentration_2}, we have
$$
\Pr_{\lx\sim\gamma_n}\sbra{K_{\lx}\ge\pbra{6\sqrt{2m}+12t\sqrt m+t^2}^2}\le e^{-t^2/2}
\quad\text{holds for any $t\ge0$.}
$$
Now we set 
$$
t=\frac1{168}\sqrt{\frac r{m+\sqrt r}}\ge0
$$
and assume $r\ge98m$.
Then $6\sqrt{2m}\le\frac67\sqrt r$, $12t\sqrt m\le\frac1{14}\sqrt r$, and $t^2\le\frac1{14}\sqrt r$.
Therefore
\begin{equation*}
\Pr_{\lX\sim\gamma_n}\sbra{\sum_{i=1}^m\abra{\lX\tensor \lX,M_i}^2\ge r}
=\Pr_{\lx\sim\gamma_n}\sbra{K_{\lx}\ge r}
\le e^{-t^2/2}
=\exp\cbra{-\frac1{56448}\cdot\frac r{m+\sqrt r}}.
\tag*{\qedhere}
\end{equation*}
\end{proof}

Finally we present the missing proof of \Cref{clm:fghv}.
\begin{proof}[Proof of \Cref{clm:fghv}]
First we observe that rows of $M$ are unit vectors, therefore
\begin{equation}\label{eq:clm:fghv_1}
\frob{M}=\sqrt m.
\end{equation}
In addition, rows of $M$ are orthogonal to each other, therefore the operator norm of $M$ is 
\begin{equation}\label{eq:clm:fghv_2}
\opnorm{M}\le1.
\end{equation}

We index the columns of $M$ by $[n]^2$ and let the column vectors of $M$ be $\pbra{b_{i,j}}_{i,j\in[n]}$. 
Since rows of $M$ are flattened matrices with zero diagonal, we have
\begin{equation}\label{eq:clm:fghv_3}
b_{i,i}=0^m\quad\text{for all $i\in[n]$.}
\end{equation}
Now we bound $F,G,H,V$ separately.

\paragraph*{Bounding $F$.}
Observe that
\begin{align*}
F^2
&=\pbra{\E_{\ly\sim\gamma_n}\sbra{\vabs{M(\ly\otimes \ly)}}}^2
\le\E_{\ly\sim\gamma_n}\sbra{\vabs{M(\ly\otimes \ly)}^2}
=\E_{\ly\sim\gamma_n}\sbra{\vabs{\sum_{i,j\in[n]}b_{i,j}\ly_i\ly_j}^2}
\tag{by convexity}\\
&=\E_{\ly\sim\gamma_n}\sbra{\sum_{i,j,i',j'\in[n]}\abra{b_{i,j},b_{i',j'}}\ly_i\ly_j\ly_{i'}\ly_{j'}}
=\sum_{i,j\in[n]}\pbra{\vabs{b_{i,j}}^2+\abra{b_{i,j},b_{j,i}}}
\tag{by \Cref{eq:clm:fghv_3}}\\
&\le\sum_{i,j\in[n]}\pbra{\vabs{b_{i,j}}^2+\frac12\pbra{\vabs{b_{i,j}}^2+\vabs{b_{j,i}}^2}}
=2\sum_{i,j\in[n]}\vabs{b_{i,j}}^2\\
&=2\frob{M}^2=2m.
\tag{by \Cref{eq:clm:fghv_1}}
\end{align*}

\paragraph*{Bounding $G$ and $H$.}
Fix an arbitrary $y\in\Rbb^n$ and we first simplify $g(y)$.
For each $i\in[n]$, define vector $b_i=\sum_{j\in[n]}b_{i,j}y_j$ and let $B$ be the matrix with $b_i$'s as column vectors.
Then
\begin{equation}\label{eq:clm:fghv_4}
g(y)
=\sup_{z\in\Sbb^{n-1}}\vabs{\sum_{i,j\in[n]}b_{i,j}z_iy_j}
=\sup_{z\in\Sbb^{n-1}}\vabs{\sum_{i\in[n]}b_iz_i}
=\opnorm{B}
\le\frob{B}
=\sqrt{\sum_{i\in[n]}\vabs{\sum_{j\in[n]}b_{i,j}y_j}^2}.
\end{equation}
Now we bound $G$:
\begin{align*}
G^2
&=\pbra{\E_{\ly\sim\gamma_n}\sbra{g(\ly)}}^2
\le\E_{\ly\sim\gamma_n}\sbra{g(\ly)^2}
\tag{by convexity}\\
&\le\E_{\ly\sim\gamma_n}\sbra{\sum_{i\in[n]}\vabs{\sum_{j\in[n]}b_{i,j}\ly_j}^2}
=\E_{\ly\sim\gamma_n}\sbra{\sum_{i\in[n]}\sum_{j,j'\in[n]}\abra{b_{i,j},b_{i,j'}}\ly_j\ly_{j'}}
\tag{by \Cref{eq:clm:fghv_4}}\\
&=\sum_{i,j\in[n]}\vabs{b_{i,j}}^2=\frob{M}^2=m.
\tag{by \Cref{eq:clm:fghv_1}}
\end{align*}
Similar argument works for $H$.

\paragraph*{Bounding $V$.}
Note that for any $z\in\Sbb^{n-1}$, we have $\vabs{z\otimes z}=\vabs{z}^2=1$.
Thus, by \Cref{eq:clm:fghv_2}, we have
\begin{equation*}
V=\sup_{z\in\Sbb^{n-1}}\vabs{M(z\otimes z)}\le\opnorm{M}\le1.
\tag*{\qedhere}
\end{equation*}
\end{proof}

\end{document}